\documentclass[11pt,reqno]{amsart}

\usepackage[utf8]{inputenc}
\usepackage[T1]{fontenc}
\usepackage{amsmath}
\usepackage{amsfonts}
\usepackage{amssymb}
\usepackage{amsthm}
\usepackage{mathrsfs}
\usepackage{graphicx}
\usepackage{xcolor}
\colorlet{PURPLE}{purple}
\usepackage{slashed}
\usepackage{tikz}
\usetikzlibrary{shadings,intersections}
\usepackage{enumitem}
\usepackage{hyperref}
\usepackage[final,expansion=false]{microtype}
\usepackage[a4paper,top=2.5cm,bottom=2.5cm,left=2.5cm,right=2.5cm]{geometry}

\numberwithin{equation}{section}

\theoremstyle{plain}
\newtheorem{theorem}{Theorem}[section]
\newtheorem{proposition}[theorem]{Proposition}
\newtheorem{lemma}[theorem]{Lemma}
\newtheorem{corollary}[theorem]{Corollary}
\newtheorem{conjecture}[theorem]{Conjecture}

\theoremstyle{definition}
\newtheorem{definition}[theorem]{Definition}

\theoremstyle{remark}
\newtheorem{remark}[theorem]{Remark}

\DeclareMathOperator{\Ric}{Ric}

\DeclareMathOperator{\tr}{tr}

\newcommand{\R}{\mathcal{R}}

\newcommand{\chic}{\check{\chi}}
\newcommand{\chibc}{\underline{\check{\chi}}}
\newcommand{\Hb}{\underline{H}}
\newcommand{\Lb}{\underline{L}}
\newcommand{\chib}{\underline{\chi}}
\newcommand{\psib}{\underline{\psi}}
\newcommand{\chih}{\hat{\chi}}
\newcommand{\chibh}{\hat{\underline{\chi}}}
\newcommand{\etab}{\underline{\eta}}
\newcommand{\omegab}{\underline{\omega}}
\newcommand{\betab}{\underline{\beta}}
\newcommand{\alphab}{\underline{\alpha}}
\newcommand{\ot}{\overset{\circ}}

\newcommand{\wt}{\widetilde}
\newcommand{\LS}{\mathcal{L}_S}
\newcommand{\LH}{\mathcal{L}_H}
\newcommand{\LHb}{\mathcal{L}_{\Hb}}
\newcommand{\LR}{\mathcal{L}_R}
\newcommand{\lnm}{\left\|}
\newcommand{\rnm}{\right\|}

\newcommand{\snab}{\nabla\mkern-13mu /\,\,}
\newcommand{\sdiv}{\mbox{div}\mkern-19mu /\,\,\,\,}

\newcommand{\strc}{\mbox{tr}\mkern-13mu /\,\,}
\newcommand{\slap}{\mbox{$\triangle  \mkern-13mu / \,$}}

\newcommand{\dv}{{\rm div}}
\newcommand{\cl}{{\rm curl}}

\def\Rc{{\rm Ric}}
\def\k{\kappa}
\def\Lie{\mathcal{L}}
\def\SS{\mathbb{S}}

\def\LS{\mathcal{L}_S}
\def\LH{\mathcal{L}_{H}}
\def\LHb{\mathcal{L}_{\Hb}}
\def\LR{\mathcal{L}_{R}}

\title[Instability of $\kappa$-Self-Similar Naked Singularities]{Naked Singularities beyond Spherical Symmetry: Instability of $\kappa$-Self-Similar Solutions\\ via an Iteration Scheme}

\author{Xinliang An}
\address{Department of Mathematics, National University of Singapore, Singapore}
\email{matax@nus.edu.sg}

\author{Shengrong Wu}
\address{Department of Mathematics, National University of Singapore, Singapore}
\email{shengrong\_wu@u.nus.edu}

\date{}


\begin{document}

\begin{abstract}

This paper provides the instability counterpart to our recent construction of nonspherically symmetric approximating $\kappa$-self-similar naked-singularity solutions for the Einstein--scalar field system. These singular solutions contain pervasive nonspherical borderline terms, and to prove instability the delicate renormalization procedure developed in \cite{An2025} does not extend to {the more singular setting considered here}. To overcome these difficulties, we introduce a new iteration scheme adapted to singular backgrounds whose leading-order geometry depends on the angular variables. At each step, the nonlinear coefficients are frozen using the preceding double-null geometry, and the resulting equations are solved in a triangular order. In this way, the nonspherical borderline terms are incorporated into the approximate geometry rather than treated as perturbative errors, yielding successively sharper estimates. After sufficiently many iterations, the scheme controls the singular angular structure and produces an approximate spacetime whose difference from the exact solution satisfies the required bounds. In particular, these bounds provide an existence region large enough to carry out the instability argument. We further prove that anisotropic perturbations of the outgoing data, arbitrarily small in a scale-critical norm, lead to the formation of a trapped surface. We also formulate and verify a matter-focusing condition under which a parabolic flow argument guarantees the existence of a corresponding marginally outer trapped surface (MOTS). Together, these results establish the nonlinear instability of $\kappa$-self-similar naked singularities beyond spherical symmetry in the Einstein--scalar field system and introduce a framework for applications across Einstein systems.

\end{abstract}

\maketitle

\tableofcontents

\section{Introduction}\label{Section_Introduction}

\subsection{Background}

\textcolor{black}{{During gravitational collapse, naked singularities may arise as the Einstein field equations evolve. Information from such a singularity may not remain confined and may be perceived by distant observers. This poses a crisis for classical general relativity.} In 1965, Penrose {introduced} the concept of cosmic censorship, and in 1993 it was further formulated by Christodoulou as follows.}

\begin{conjecture}[Weak Cosmic Censorship]
\textit{For generic, regular, asymptotically flat initial data for the Einstein field equations, the maximal future development has complete future null infinity. In particular, singularities formed in gravitational collapse should not be visible from infinity.}
\end{conjecture}

\textcolor{black}{This is a vast and fundamental conjecture in the field of general relativity. Christodoulou initiated the rigorous mathematical exploration of this conjecture {in the 1990s}. He first considered the {following} $3+1$-dimensional Einstein-scalar field system:}

\begin{equation}\label{eq:esf-intro}
\begin{split}
    R_{\mu\nu}-\frac{1}{2}R g_{\mu\nu}&=T_{\mu\nu},\\
    T_{\mu\nu}&=D_\mu\phi D_\nu\phi
    -\frac{1}{2}g_{\mu\nu}D^\lambda\phi D_\lambda\phi.
\end{split}
\end{equation}
Here $(\mathcal{M},g)$ is a $3+1$-dimensional Lorentzian spacetime, $\phi:\mathcal{M}\to\mathbb{R}$ is a real scalar field, $R_{\mu\nu}$ and $R$ denote the Ricci and scalar curvatures of $g$, respectively, and $T_{\mu\nu}$ is the stress--energy tensor of $\phi$.

 {In} spherical symmetry, Christodoulou designed an ingenious four-step argument {\color{black}\cite{Christodoulou1991,Christodoulou1993,Christodoulou1994,Christodoulou1999}} to attack this problem. {An--Tan later generalized these four steps in \cite{An-Tan} to the $3+1$-dimensional Einstein--Maxwell--charged-scalar-field system.} For the circularly symmetric $2+1$-dimensional Einstein-scalar field system, {Cicortas employed a two-step argument in \cite{Cicortas2026}}. {In} spherical symmetry, we also mention the related works {\color{black} \cite{An-Zhang,CicortasKehle2026,Guo--Hadzic--Jang,Li-Liu,Liu-Li2018,Singh-Zheng,Zheng}} and the references therein.
  
 \vspace{1mm}
{\color{black}Outside spherical symmetry, the road toward this conjecture is much harder. In this paper, we study the above system \eqref{eq:esf-intro} {using double-null coordinates} $(u,v,\theta^A)$, {in which the metric takes the form}
\begin{equation}\label{metric_expression}
    g=-2\Omega^2\left(du\otimes dv+dv\otimes du\right)+g_{AB}\left(d\theta^A-b^Adu\right)\otimes \left(d\theta^B-b^Bdu\right).
\end{equation}
We write
\[
    e_3=\Omega^{-1}\left(\partial_u+b^A\partial_{\theta^A}\right),\qquad e_4=\Omega^{-1}\partial_v,\qquad e_A=\partial_{\theta^A},
\]
and denote by $S_{u,v}$ the level spheres of $u$ and $v$. The Ricci coefficients used below are
\begin{equation}\label{Ricci_coefficients_definition}
    \begin{aligned}
    \chi_{A B}=g\left(D_A e_4, e_B\right), &\qquad \underline{\chi}_{A B}=g\left(D_A e_3, e_B\right), \\
    \eta_A=-\frac{1}{2} g\left(D_3 e_A, e_4\right), &\qquad \underline{\eta}_A=-\frac{1}{2} g\left(D_4 e_A, e_3\right), \\
    {\omega}=-\frac{1}{4} g\left(D_4 e_3, e_4\right), &\qquad \underline{\omega}=-\frac{1}{4} g\left(D_3 e_4, e_3\right),\\
    \zeta_A=\frac{1}{2}g\left(D_A e_4,e_3\right).
    \end{aligned}
\end{equation}
The trace and trace-free parts of $\chi$ are denoted by $\operatorname{tr}\chi$ and $\hat{\chi}$, and similarly for $\underline{\chi}$.

{Building on} Christodoulou's spherically symmetric naked singularity, the first author proved {a naked-singularity-censoring result with an anisotropic apparent horizon in \cite{An2025}}, {\color{black}showing that anisotropic outgoing perturbations {that} are arbitrarily small in {a} scale-critical norm can generate a spacelike anisotropic apparent horizon that encloses the singularity. Moreover, for every $N\in\mathbb{N}$, the construction exhibits $2N$ independent unstable directions and yields codimension-$2N$ nonlinear instability}. {See also the work of Li--Liu \cite{Li-Liu}, in which} {\color{black} singularities in their spherically symmetric class are shown to be unstable under sufficiently close-to-spherically-symmetric gravitational perturbations.} To construct naked-singularity solutions beyond spherical symmetry, Rodnianski and Shlapentokh-Rothman introduced $\kappa$-self-similarity for the Einstein vacuum equations. {An and Wu generalized this notion in \cite{AnWu2026}} to the Einstein-scalar field system. The nonspherical character of these self-similar solutions is already visible in the constraint identities on $v=0$. In particular, the incoming expansion and shear identities in \cite[Lemma~3.1]{AnWu2026} give
\begin{equation*}
    \Omega\operatorname{tr}\underline{\chi}
    =\frac{2}{u}+\operatorname{div} b,
    \qquad
    \Omega\underline{\hat{\chi}}_{AB}
    =\frac12(\nabla\hat{\otimes} b)_{AB}{,}
\end{equation*}
with $\nabla$ being the angular derivative on $S_{u,v}$, so the angular variation of $b$ affects both the incoming null expansion and its trace-free part. In particular, any shift $b$ with a nonvanishing trace-free symmetrized angular derivative generates nonzero incoming shear, whereas this quantity vanishes in spherical symmetry. Moreover, the constraint equation \cite[(3.2)]{AnWu2026} reads
\begin{equation*}
\begin{aligned}
    \frac{2\kappa}{(-u)^2}
    &+\frac{1}{-u}\operatorname{div}b
    -\mathcal L_b\operatorname{div}b
    -\frac12(\operatorname{div}b)^2
    -\frac14\lvert\nabla\hat{\otimes} b\rvert^2 \\
    &-\frac{4}{-u}\mathcal L_b\log\Omega
    +2\operatorname{div}b\,\mathcal L_b\log\Omega
    -(\Omega D_3\phi)^2=0.
\end{aligned}
\end{equation*}
Thus the angularly dependent shift is coupled nonlinearly to the lapse and the scalar field. Within spherical symmetry, one has $b=0$ and all of these angular terms disappear. {Returning} to the Einstein vacuum equations, Rodnianski and Shlapentokh-Rothman constructed {naked-singularity solutions without symmetry assumptions} in \cite{RodnianskiShlapentokhRothman2019}. {Controlling} the Einstein vacuum equations in this regime is quite delicate. The regularity of the inner Cauchy horizon remains open for these spacetimes. {In \cite{AnWu2026}, An and Wu} {\color{black}constructed global nonspherically symmetric naked singularity solutions to the Einstein--scalar field system near Christodoulou's spherically symmetric $\kappa$-self-similar solution, proved that these spacetimes retain an incomplete future null infinity, and derived precise asymptotics at their singular inner Cauchy horizons. In particular, they established the desired $C^{1,\frac{\kappa}{1-\kappa}+}$-inextendibility across the inner Cauchy horizon}.}

{Extending} the naked-singularity censoring results in \cite{An2025} to these {nonspherically symmetric} $\kappa$-{self-similar} naked singularities is quite challenging. For the $\kappa$-self-similar naked singularities, the coefficients of the borderline singular terms depend on angular variables. The energy-estimate framework and the renormalization procedures {used} in the works of Christodoulou \cite{Christodoulou2009}, An and Luk \cite{AnLuk2017}, and An \cite{An2025} are not robust enough. 

In this paper, we introduce a {robust new} framework: we design an iteration scheme and construct a sequence of controlled approximating solutions to the Einstein field equations, which absorb the angularly dependent borderline singular terms into the approximate geometry instead of treating them as perturbative errors. At each stage, the nonlinear coefficients are frozen {using} the preceding double-null geometry and the equations are solved in a triangular order, so that the Einstein--scalar field residuals gain an additional small factor with every iteration. \textcolor{black}{The iteration procedure then improves the estimates for the difference between the exact {solution} and the iterated approximate solutions, and {yields} an existence region for the full Einstein field equations. After {several} iterations, the {resulting} existence region is large enough to carry out the trapped-surface-formation argument.}

\begin{sloppypar}
\textit{Within the resulting \textcolor{black}{existence} region, we {prove that anisotropic perturbations of the outgoing data that are arbitrarily small in a scale-critical norm lead to trapped-surface formation.}} The point is not merely that trapped surfaces form, but that the instability mechanism is genuinely nonspherical and is measured in scale-invariant quantities natural for the self-similar geometry. In the anisotropic case, the trapped surface is obtained by {a} carefully designed deformation of the sphere along an incoming null hypersurface. {\color{black}In the gravitational-collapse setting, An introduced an elliptic method for {constructing MOTSs} in \cite{An2020}; An--Han extended it to fully anisotropic apparent horizons in \cite{AnHan2021}; An--He introduced the null comparison principle in \cite{AnHe2024}; and, more recently, An--He used new elliptic arguments to determine the complete apparent-horizon dynamics in Kerr black-hole-formation spacetimes \cite{AnHe2025}.} {In these works, the existence and uniqueness of MOTSs and the achronality of the apparent horizon are established in the same framework.}
\end{sloppypar}

In the highly singular regime governing the instability of $\kappa$-self-similar naked singularities, the proofs {of} these properties are separate. In this paper, we {prove the existence of a MOTS} via a mean-curvature-flow approach. Key calculations and estimates are derived in Section~5. With these estimates{,} we show that the {positive contribution from the scalar field} along the incoming null direction dominates the shear and trace-free curvature contributions. This quantitative dominance is the matter-focusing condition required to control the flow. With these estimates in hand, the remaining mean-curvature-flow arguments are relatively standard{; we refer the reader to Roesch--Scheuer \cite{RoeschScheuer2021} for details}. {\color{black}In Section~5{,} we also include a new second-variation formula for the area functional. As demonstrated in Section~5.3 {for} less singular regimes, our derived formula yields the coercive positivity underlying {the} uniqueness of the MOTS and explains how the flow construction can be combined with the null comparison principle of An--He \cite{AnHe2024} to prove the achronality of the resulting apparent horizon.}

\subsection{Main results}

{\color{black}
We now state our main results in the characteristic setting. As {indicated by} \eqref{metric_expression},
the double-null foliation is used both to prescribe the initial data and to evolve them. The incoming null hypersurface $\underline{H}_0$ carries the $\kappa$-self-similar background geometry, whereas the outgoing null hypersurface $H_{-1}$ carries the perturbation responsible for the subsequent focusing. 

\vspace{2mm}
We begin by fixing the terminology for the null expansions and then define the scale-invariant class of initial data appearing in the theorem statements.

\begin{definition}
Let $S$ be a closed spacelike sphere contained in an incoming null hypersurface $\underline{H}_v$. As given in \eqref{Ricci_coefficients_definition}{,} we denote its future-directed outgoing and incoming null expansions by $\tr\chi$ and $\tr\chib$, respectively. We call $S$ a trapped surface if
\[
    \tr\chi<0,\quad \tr\chib<0{,}
\]
everywhere on $S$. We call $S$ a marginally outer trapped surface, or a MOTS, if
\[
    \tr\chi=0,\quad \tr\chib<0{,}
\]
everywhere on $S$.
\end{definition}

Since $u=0$ is a singular boundary, the appropriate notion of smallness on $\underline{H}_0$ is introduced in a scale-invariant way. Namely, after the leading self-similar powers of $-u$ have been factored out, the weighted angular and incoming derivatives are required to remain small. The following definition makes this notion precise:

\begin{definition}
We say that an initial data set $(g_{AB},\Omega,\zeta_A,b^A,\phi)$ on $\underline{H}_0$ satisfies an \textbf{$(\tilde\epsilon,\tilde N)$-small scale-invariant bound} if the following conditions hold:
\begin{enumerate}
\item {Each} of the quantities
\[
(-u)^{-2}g_{AB}(u,0)-{g^{\mathbb{S}}}_{AB},\qquad
(-u)b^A(u,0),\qquad \zeta_A(u,0),
\]
denoted generically by $\psi$, satisfies
\[
\sum_{i+j\leq \tilde N+1}
\left\|(-u)^j(\partial_u+b)^j(\nabla^\mathbb{S})^i\psi\right\|_{L^2(\mathbb{S})}
\leq\tilde\epsilon.
\]
\item We require
\[
\sum_{1\leq i+j\leq \tilde N+1}
\left\|(-u)^j(\partial_u+b)^j(\nabla^\mathbb{S})^i
\left(\log\Omega(u,0),\phi(u,0)\right)\right\|_{L^2(\mathbb{S})}
\leq\tilde\epsilon.
\]
\item Along $\underline{H}_0${,} the {characteristic} initial data satisfy the corresponding constraint equations for the Einstein--scalar field system, namely, the $\nabla_3\operatorname{tr}\underline{\chi}$ and $\nabla_3(\zeta-\nabla\log\Omega)$ equations.
\end{enumerate}
\end{definition}

{\color{black}Using a new iteration scheme, we first establish the existence result needed for trapped-surface formation. The solution is constructed in a scale-invariant region that narrows toward the self-similar boundary. The statement involves two large parameters: $N_1$ measures the regularity of the initial data, whereas $N$ determines the depth of the iteration scheme. More precisely, the proof carries out $4N$ iteration steps and uses the $4N$-th iterate as the reference spacetime, yielding the exponent $\kappa/N$ in the description of the existence region.

\begin{theorem}
Fix $\kappa>0$. Choose integers $1\ll N\ll N_1$, with $N_1$ sufficiently large in terms of $N$, and then choose $\delta,\epsilon>0$ so that $\epsilon\ll{\delta}/{N^2}\ll{\kappa}/{N^3}.$
Suppose that $\Omega^2(u,0)\lesssim(-u)^\kappa$, and consider the characteristic initial value problem whose data $(g_{AB},\Omega,\zeta_A,b^A,\phi)$ on $\underline{H}_0$ obey an $(\epsilon,N_1)$-small scale-invariant bound. Prescribe continuous data $\Omega\chi$, $\Omega\omega$, and $\partial_v\phi$ on $H_{-1}$ for $0\leq v\leq v_*$, with
\[
A:=\sup_{0\leq v\leq v_*}
\left\|\left(\Omega\chi,\Omega\omega,\partial_v\phi\right)(-1,v)\right\|_{H^{N_1}(\mathbb{S})}<\infty.
\]

Then there exists a sufficiently small $\epsilon_1=\epsilon_1(N,\kappa,A)\in(0,v_*]$ such that the Einstein--scalar field equations admit a solution in
\[
\left\{-1\leq u<0,\quad 0\leq v<\epsilon_1(-u)^{1+\max\{\delta,\kappa/N\}}\right\}.
\]
\end{theorem}

\vspace{2mm}

Note that in our preceding construction \cite{AnWu2026} the self-similar parameter satisfies $\kappa\in(0,\frac13)$ and $\Omega^2(u,0)\lesssim(-u)^\kappa$; hence the present results apply to the initial data constructed there. In this paper the above theorem yields an existence region large enough to carry out the trapped-surface formation argument, whereas the method of \cite{An2025} does not produce such a region in the present nonspherical setting. The proof here performs $4N$ iteration steps. At each step, we freeze the nonlinear coefficients at the preceding approximate double-null geometry and solve the resulting equations in a triangular order, thereby improving the Einstein--scalar field residual by an additional factor of $v/(-u)^{1+\delta}$. We then take the $4N$-th iterate as the reference spacetime and close weighted transport and curvature estimates for its difference from the exact solution. The exponent $\kappa/N$ records the gain from the iteration, while $\delta$ accounts for the small loss needed to close the exact estimates.

}

\vspace{2mm}

The second theorem describes the instability mechanism. Within the established existence region, in the spirit of An--Han \cite{AnHan2021} and An \cite{An2025}, we proceed to prove that an anisotropic perturbation {of} the outgoing data, which is arbitrarily small in {a} scale-critical norm, {triggers} the formation of {a} trapped surface and {a} corresponding marginally outer trapped surface (MOTS). Here and below, the subscript $0$ denotes the corresponding quantity for the unperturbed $\kappa$-self-similar background. {For} a universal positive constant $C$, the anisotropic perturbation {need only be imposed} in an angular sector $B_{S_{-1,v}}(\theta_0,r)\subset S_{-1,v}$ {centered at} $\theta_0\in \mathbb{S}^2$ {and having disk} radius $r>0$. 
\begin{theorem}
Suppose that there exist $\theta_0\in \mathbb{S}^2$ and {a disk} radius $r>0$ such that, for every $v\in(0,\epsilon_1)$, the outgoing data along $H_{-1}$ satisfies {\color{black}
\begin{equation*}
        \left|\Omega\chih\right|^2-\left|(\Omega\chih)_0\right|^2+\left|\Omega e_4\phi\right|^2-\left|(\Omega e_4\phi)_0\right|^2 > \frac{1}{C}v^{2\lambda},\quad \text{for}\  \theta\in B_{S_{-1,v}}(\theta_0,r),
    \end{equation*}
    \begin{equation*}
        \left|\Omega\chih-(\Omega\chih)_0\right|+\left|\Omega e_4\phi-(\Omega e_4\phi)_0\right|\leq Cv^\lambda,\quad \text{for all}\  \theta\in\SS.
\end{equation*}    }
{for a} parameter $\lambda$ satisfying
\[
0<\lambda<\frac{\kappa(2N-3)}{2(2N+3\kappa)}. 
\]
Then, for $\epsilon_1$ sufficiently small, every incoming null hypersurface $\underline{H}_v$ with $v\in(0,\epsilon_1)$ contains a trapped surface and a corresponding MOTS.
\end{theorem}

The anisotropic perturbation first produces strong focusing only in a small angular sector, and a deformation argument of the sphere is then used to turn this localized effect into a closed trapped surface. In Section~5{,} we further {show} that, for $v/(-u)^{1+C\delta}\leq\epsilon$, we have
\begin{equation}
    |\Omega D_3\phi|^2\sim(-u)^{-2}\gg\epsilon(-u)^{-2}
    \gtrsim |\Omega\chib|\,|\Omega\chibh|+|\Omega^2\alphab|,
\end{equation}
where $\alphab=W_{3A3B}=R_{3A3B}-\frac12\operatorname{Ric}_{33}g_{AB}$. {We refer to this as a matter-focusing condition; it states} that the matter focusing dominates the shear and curvature terms. {When this condition holds,} a relatively standard mean-curvature-flow argument guarantees the existence of the corresponding MOTS. This setting differs from {that of} An \cite{An2025}, where the incoming naked-singularity background is spherically symmetric and its angular geometry is therefore better adapted to the deformation argument. Here the background itself is nonspherical and more singular. In this paper{,} the focused angular sector is distorted and shrinks as it is transported toward $u=0$, {while} angularly dependent borderline terms create additional derivative losses. Overcoming these new difficulties requires the iteration estimates developed in this paper.

\vspace{3mm}


The $\kappa$-self-similar ansatz on $\underline{H}_0$ reduces the null constraint equations to equations on the sphere $S_{-1,0}$. More precisely, after prescribing the small free data, including the conformal metric, the lapse, the scalar field, and the curl part of the shift, the remaining part of the shift is written as a gradient component plus the prescribed curl component. The $\nabla_3\operatorname{tr}\underline{\chi}$ constraint then becomes a nonlinear equation on $S_{-1,0}$ for this gradient component, together with a scalar normalization parameter chosen so that the right-hand side has zero mean. {This equation is solved iteratively} using elliptic estimates and the invertibility of model operators of the form
\[
    u+\mathcal{L}_X u+h\cdot u=F
\]
for small vector fields $X$ and coefficients $h$ on the sphere. Once the shift is determined, the remaining torsion variable is recovered from the $\nabla_3\eta$ equation together with the Codazzi relation and the identity $\underline{\eta}=-\eta+2\nabla\log\Omega$. This strategy is close in spirit to the constraint construction used in the self-similar vacuum analysis of Rodnianski and Shlapentokh-Rothman \cite{RodnianskiShlapentokhRothman2019}, and here it is carried out for the Einstein{--}scalar field system.

\begin{remark}[The focusing threshold]
The exponent in the trapped-surface criterion is explicit. We have
\[
    \frac{\kappa(2N-3)}{2(2N+3\kappa)}
    =\frac{\kappa}{2}-\frac{3\kappa(\kappa+1)}{2(2N+3\kappa)},
\]
so the admissible range approaches $\lambda<\frac{\kappa}{2}$ as {\color{black}the iteration parameter $N$ tends to infinity}. After the perturbation has been transported in the incoming direction and the Raychaudhuri equation has been integrated, the leading focusing term along the curve $v_0(u)=\epsilon_1(-u)^{1+\frac{3\kappa}{2N}}$ is
\[
    -\frac{v_0(u)^{1+2\lambda}}{(-u)^{1+\kappa}}.
\]
Substituting $v_0(u)$ into the leading term gives
\[
-\epsilon_1^{1+2\lambda}(-u)^{-q},
\qquad
q=\kappa-2\lambda-\frac{3\kappa}{2N}-\frac{3\kappa\lambda}{N}.
\]
The condition $\lambda<\frac{\kappa(2N-3)}{2(2N+3\kappa)}$ is equivalent to $q>0$, so the estimates for the singular terms {remain} precise. Consequently, the theorem proves instability arising from perturbations of the size $v^{\frac{\kappa}{2}-o(1)}$ while still permitting angular localization.
\end{remark}
}

\subsection{New ingredients}

To achieve the above results, we introduce several novel methodological ingredients that significantly streamline the analysis and yield sharper geometric control compared with the prior literature.

\vspace{2mm}
\paragraph{\textbf{Construction of approximate spacetimes.}}
The most important technical step is the construction of a high-order approximate spacetime. Starting from the data on $\underline{H}_0$, we construct a sequence of Lorentzian metrics and scalar fields $(g^{(i)},\phi^{(i)})$. {\color{black}{At step $i+1$,} the nonlinear coefficients in the null structure equations are frozen {using} the $i$-th approximate geometry, and the equations for the new quantities are solved in a triangular order. Thus, each step produces an actual double-null geometry that satisfies the Einstein--scalar field equations only up to an error that decreases as $i$ increases.}

One representative part of the iteration is the recovery of $g^{(i+1)}$ from the outgoing shear. First, using the $i$-th spacetime as background, we solve the $e_3$-transport equation
\begin{equation}\label{eq:intro-approx-chih}
    \left\{
    \begin{aligned}
        & (\Omega\nabla_3)^{(i)}(\Omega\chih^{(i+1/2)})_{AB}
        +\frac{1}{2}(\Omega\tr\chib)^{(i)}(\Omega\chih^{(i+1/2)})_{AB} \\
        &\qquad =(\Omega^{(i)})^2\left(\left(\nabla^{(i)}{\hat\otimes}^{(i)}\eta^{(i)}
        +\eta^{(i)}{\hat\otimes}^{(i)}\eta^{(i)}
        +\nabla\phi^{(i)}{\hat\otimes}^{(i)}\nabla\phi^{(i)}\right)_{AB}
        -\frac{1}{2}(\tr\chi)^{(i)}{\chibh^{(i)}}_{AB}\right),\\
        &{\Omega\chih^{(i+1/2)}}_{AB}(-1,v)
        =\frac{1}{2}\left(\Omega\chih_{AC}g^{CD}{g^{(i)}}_{DB}
        +\Omega\chih_{BC}g^{CD}{g^{(i)}}_{DA}\right)(-1,v).
    \end{aligned}
    \right.
\end{equation}
The output is then inserted into the $e_4$-evolution equations for the metric and the outgoing expansion:
\begin{equation}\label{eq:intro-approx-metric}
    \left\{
    \begin{aligned}
        &\mathcal{L}_{\partial_v}(g^{(i+1)})_{AB}
        =(\Omega\tr\chi)^{(i+1)}(g^{(i+1)})_{AB}
        +2{(\Omega\chih)^{(i+1)}}_{AB},\\
        &\partial_v\left(\Omega\tr\chi\right)^{(i+1)}
        =-\frac{1}{2}\left((\Omega\tr\chi)^{(i+1)}\right)^2 - (\partial_v\phi^{(i+1)})^2-4(\Omega\omega)^{(i+1)}(\Omega\tr\chi)^{(i+1)}
        \\
        &\qquad\qquad\qquad\qquad-{(\Omega\chih)^{(i+1)}}_{AB}{(\Omega\chih)^{(i+1)}}_{CD}\cdot {g^{(i+1)}}^{AC}{g^{(i+1)}}^{BD},\\
        &(\Omega\chih)^{(i+1)}_{AB}
        =\frac{1}{2}\left((\Omega\chih)^{(i+1/2)}_{AC}{g^{(i+1)}}_{DB}
        +(\Omega\chih)^{(i+1/2)}_{CB}{g^{(i+1)}}_{DA}\right){(g^{(i)})}^{CD},\\
        &(g^{(i+1)})_{AB}(u,0)=g_{AB}(u,0),\qquad
        (\Omega^{-1}\tr\chi)^{(i+1)}(u,0)=\Omega^{-1}\tr\chi(u,0).
    \end{aligned}
    \right.
\end{equation}
The intermediate notation $(\Omega\chih)^{(i+1/2)}$ is essential here. The tensor obtained from the $e_3$-transport equation is defined using the old metric $g^{(i)}$, whereas the final outgoing shear at the next step must be trace-free with respect to the new metric $g^{(i+1)}$. Moreover, along $H_{-1}$ the prescribed shear is measured with respect to the exact induced metric, not with respect to $g^{(i)}$. The algebraic passage from $(\Omega\chih)^{(i+1/2)}$ to $(\Omega\chih)^{(i+1)}$ is therefore designed both to preserve the trace-free condition relative to $g^{(i+1)}$ and to improve the matching of the outgoing data on $H_{-1}$ as $i$ increases.

This {example} shows the basic mechanism of the construction. {\color{black}The $e_3$-equations propagate the quantities that are sensitive to the incoming self-similar geometry, whereas the $e_4$-equations recover the metric quantities along the outgoing direction.} The ordering avoids solving the fully nonlinear system at once and isolates the error made by freezing the coefficients. {Set}
\[
    \wt{\Rc}^{(i)}_{\mu\nu}
    =\Ric(g^{(i)})_{\mu\nu}-\partial_\mu\phi^{(i)}\partial_\nu\phi^{(i)},
    \qquad
    \wt{\square\phi}^{(i)}=\Box_{g^{(i)}}\phi^{(i)}{.}
\]
{The} estimates in the construction section show, schematically, that the scale-invariant norms of $\wt{\Rc}^{(i)}$ and $\wt{\square\phi}^{(i)}$ gain one additional factor of $v/(-u)^{1+\delta}$ at each step. {\color{black}After choosing a sufficiently high iterate as the reference spacetime, we {solve} the exact Einstein--scalar field system by estimating the difference between the exact solution and this approximate geometry. This is how the approximation scheme enters the proof of the existence region.}

{\color{black} We also remark the microlocal property of the approximation. If we assume that the initial data satisfy $g, \phi\in C^\infty(S_{u,v}),$
then along each $u$-slice, we can construct $(g^{(\infty)},\phi^{(\infty)})$ from the spacetime sequence $(g^{(i)},\phi^{(i)})$ such that the differences with genuine solution satisfy
$$|g^{(\infty)}-g^{\text{genuine}}|+|\phi^{(\infty)}-\phi^{\text{genuine}}| \lesssim O(v^\infty).$$
}

\vspace{3mm}
\paragraph{\textbf{Scale-invariant comparison estimates.}}
{\color{black}We next choose a sufficiently high iterate as the reference spacetime and denote it by $(\ot{g},\ot{\phi})$.} For every Ricci coefficient, scalar field derivative, and curvature component $\Phi$, we write
\[
    \wt{\Phi}=\Phi-\ot{\Phi}.
\]
The estimates are then carried out in scale-invariant weighted norms on spheres and null hypersurfaces, with norms denoted by $\LS$, $\LH$, $\LHb$, $\LR$, respectively. Their precise definitions are provided in later sections. With these norms, the important point is the form of the target estimates. For a typical curvature pair $(\Psi_1,\Psi_2)$, the desired estimate is
\begin{equation}\label{eq:intro-target-energy}
    \sum_{i\leq 4}
    \lnm \nabla^i\wt{\Psi_1}\rnm^2_{\LH(u,v;N)}
    +\lnm \nabla^i\wt{\Psi_2}\rnm^2_{\LHb(u,v;N)}
    \lesssim \left(\frac{v}{-u}\right)^{3/2}.
\end{equation}
This is one of the estimates that must be tracked and proved separately in the bootstrap argument.

The role of the approximate spacetime is to remove the large self-similar background from these estimates. Instead of deriving energy estimates directly for the exact solution, we derive equations for the differences $\wt{\Phi}$. The coefficients in these equations are controlled by the approximate geometry, while the source terms are either perturbative products involving the differences or the residual errors
\[
    \Ric(\ot{g})_{\mu\nu}-\partial_\mu\ot{\phi}\partial_\nu\ot{\phi},
    \qquad
    \Box_{\ot{g}}\ot{\phi}.
\]
By the construction of the approximate spacetimes, these residuals are already higher order in the scale-invariant region. This is what allows the transport estimates and the curvature energy estimates to close systematically.

Another new feature of our approach is that{,} in the paired energy estimates{,} we no longer {treat} the Weyl curvature components as the primary energy variables. Instead, we employ the following four pairs{:}
\[
\begin{gathered}
\bigl(\dv(\Omega\chic),(\Omega d\eta,\Omega K)\bigr),\qquad
\bigl((\Omega d\etab,\Omega K),\dv(\Omega\chibc)\bigr),\\
\bigl(\nabla(\Omega e_4\phi),\Omega\nabla^2\phi\bigr),\qquad
\bigl(\Omega\nabla^2\phi,\nabla(\Omega e_3\phi)\bigr).
\end{gathered}
\]
{The paired Hodge structures still present, and these four pairs make the energy estimates controlling the differences after iteration more systematic.}

\vspace{3mm}
\paragraph{\textbf{Focusing estimates and trapped surfaces.}}
The trapped-surface argument begins by propagating the perturbation of the outgoing data into the spacetime region close to the center. In this singular regime, {in the isotropic case,} the perturbation of $\Omega\chih$ and $\Omega e_4\phi$ can be {shown to remain} large enough along the incoming direction. Along the curve
$v_0(u)=\epsilon_1(-u)^{1+\frac{3\kappa}{2N}}$, inserting this lower bound into the Raychaudhuri equation for $\Omega^{-1}\tr\chi$ gives
the focusing estimate
\[
    (-u)\Omega^{-1}\tr\chi(u,v_0(u),\theta)
    \leq C-\frac{1}{C}\cdot\frac{v_0(u)^{1+2\lambda}}{(-u)^{1+\kappa}}
    +\text{lower order terms}.
\]
The lower-order terms are dominated precisely under the condition $\lambda<\frac{\kappa(2N-3)}{2(2N+3\kappa)}$.
This gives $\tr\chi<0$ for $-u$ sufficiently small.

{The anisotropic case presents an additional difficulty absent from the first author's work \cite{An2025}, in which the incoming naked-singularity initial data are spherically symmetric. In the present setting,} let $F_{u,v}:S_{-1,v}\to S_{u,v}$ be the transport map along the incoming null direction. If the perturbation is initially supported in a ball $B_{S_{-1,v}}(x,r)$, then the corresponding transported region only satisfies
\[
    B_{S_{u,v}}\bigl(F_{u,v}(x),(-u)^{1+\delta}r\bigr)\subset F_{u,v}\left(B_{S_{-1,v}}(x,r)\right).
\]
Thus, the angular region where strong focusing is available shrinks as $u\to 0$. This {degeneration} explains why the arguments in \cite{An2025} for the spherical naked singularity cannot simply be applied pointwise on the whole sphere in the setting of this paper. {\color{black}The new point, beyond the anisotropic deformation and barrier constructions in An--Han \cite{AnHan2021} and An \cite{An2025}, is that, in the present nonspherical setting, the barrier must be adapted to an angular scale that degenerates under transport. Indeed, the transport estimates show that the focused region has normalized angular radius of order $(-u)^\delta$, so choosing the barrier at this $u$-dependent scale produces derivative losses of order $(-u)^{-\delta}$ and $(-u)^{-2\delta}$ for its first and second derivatives, respectively.The mechanism used to absorb these losses has two ingredients. The $4N$-step construction underlying Theorem~1.4 yields an existence region containing the curve $v_0(u)=\epsilon_1(-u)^{1+\frac{3\kappa}{2N}}$, together with sharp comparison estimates between the exact solution and the approximate reference spacetime. These estimates keep the transport and Raychaudhuri error terms lower order. Writing
\[
\lambda_*(N):=\frac{\kappa(2N-3)}{2(2N+3\kappa)},
\]
the focusing condition is $\lambda<\lambda_*(N)$. We refer to the resulting positive difference $\lambda_*(N)-\lambda$ as the strict threshold gap.

We now introduce an exponent $\delta_1>0$ chosen strictly within this gap, so that
\[
0<\delta_1<\lambda_*(N)-\lambda,
\qquad\text{or equivalently}\qquad
\lambda<\lambda_*(N)-\delta_1.
\]
Set
\[
q:=\kappa-2\lambda-\frac{3\kappa}{2N}-\frac{3\kappa\lambda}{N}
=\left(2+\frac{3\kappa}{N}\right)\bigl(\lambda_*(N)-\lambda\bigr).
\]
Then $q>2\delta_1$, and along $v_0(u)$ the leading Raychaudhuri term
\[
-\frac{v_0(u)^{1+2\lambda}}{(-u)^{1+\kappa}}
=-\epsilon_1^{1+2\lambda}(-u)^{-q}
\]
is stronger than order $-(-u)^{-2\delta_1}$. Thus, the strict threshold gap supplies the $\delta_1$-gain, while the iteration estimates ensure that this gain survives the passage from the approximate geometry to the exact solution. For the subsequent sphere-deformation argument, we retain the weaker localized bound of order $-(-u)^{-\delta_1}$, leaving sufficient margin to absorb the derivative losses. The trapped surface itself is constructed in the exact spacetime; the approximate spacetime serves only as the reference geometry for the estimates.}

{More precisely,} to pass from localized focusing to a closed trapped surface, we deform the sphere inside the incoming null hypersurface. In geodesic coordinates $(U,V,\Theta)$, for a graph $S_{W,V}=\{U=W(\Theta)\}$, the outgoing expansion has the form
\[
    \ell^+(W)=\tr\chi^{(g)}
    +2\Delta W+4\zeta^{(g)}\cdot\nabla W
    -|\nabla W|^2\bigl(\tr\chib^{(g)}+4\omegab^{(g)}\bigr).
\]
{Here} $(\cdot)^{(g)}$ denotes the corresponding quantity in geodesic coordinates. Writing $W=-Ve^{-\Phi}$, we have
\[
    \frac12(-W)\ell^+(W)
    \leq
    \frac12(-W)\Omega^{-1}\tr\chi
    +o(1)+\Delta_{\gamma}\Phi-\frac{\kappa}{2}|\nabla\Phi|_{\gamma}^{2}.
\]
On the transported support of the perturbation, the first term is already very negative. Away from this support, we choose a barrier function $h$ satisfying
\[
    \Delta_{\gamma(h)}h-\frac{\kappa}{2}|\nabla h|_{\gamma(h)}^2<-\frac12,
    \qquad
    |\nabla h|\lesssim \epsilon_2^{-1},\quad
    |\nabla^2h|\lesssim \epsilon_2^{-2}.
\]
Note that taking $\epsilon_2\sim(-U_1)^\delta$ would produce derivative losses of size $(-U_1)^{-\delta}$ and $(-U_1)^{-2\delta}$. {\color{black}Since $\delta_1>2\delta$, the retained focusing term $-(-U_1)^{-\delta_1}$ dominates both losses. Consequently, the $u$-dependent barrier converts the localized anisotropic focusing into a closed trapped surface.}

\vspace{3mm}
\paragraph{\textbf{MOTS and the null comparison principle.}}
Earlier elliptic approaches to MOTSs and apparent horizons treat existence, uniqueness, and achronality within a unified analysis; see \cite{An2020,AnHan2021,AnHe2024,AnHe2025} by the first author and coauthors. Our aim in {the} present singular setting is limited to the existence of {a MOTS} arising from anisotropic instability. {Using the constructed anisotropic trapped surface as an inner barrier and the concrete estimates obtained above, we achieve this goal} by adapting a relatively standard mean-curvature-flow argument. We refer {the reader} to \cite{RoeschScheuer2021} for more details of the flow {argument}.

\vspace{2mm}

{In Section~5, we go} one step further and derive the second variation formula for the outgoing expansion. Let $(s,\theta)$ be a geodesic foliation on an incoming null hypersurface $\Hb$, so that the null generator is tangent to the $s$-curves. We then consider graph spheres
\[
    S_{w(t)}=\{s=w(t,\theta)\}\subset \Hb .
\]
Here $s$ is an affine parameter along each incoming null generator, so $e_3=\partial_s$ satisfies $D_{e_3}e_3=0$. 

 {We show that} the variation of $\ell^+$ in the incoming null direction takes the form
\begin{equation}\label{second variation formula in introduction}
    \frac{d}{dt}\ell^+(w)
    =-w_t\tr\chib\,\ell^+
    +2\Delta w_t
    -w_t\bigl(2K'-\tr_{S_w}\Ric\bigr)
    +2w_t\operatorname{div}\zeta'
    +4\nabla w_t\cdot\zeta'
    +2w_t|\zeta'|^2 .
\end{equation}
This formula is the key to {proving uniqueness} and to {verifying} the null comparison principle {introduced} by An and He in \cite{AnHe2024}. For instance, if two graphs $S_{w_0}$ and $S_{w_1}$ are MOTSs and {we can prove}
\begin{equation}\label{2nd variation condition 1}
    2K'+\frac12\tr\chib\,\ell^+(w)
    -\tr_{S_w}\Ric-2|\zeta'|^2>0,
\end{equation}

then $w_1-w_0$ cannot have a nonzero sign-changing extremum{,} and {uniqueness} follows.

With \eqref{second variation formula in introduction}, we can also see that if the {following} condition is satisfied{:}
 \begin{equation}\label{2nd variation condition 2}
    \tr\chib\,\ell^+(w(t))
    +2K'_{w(t)}-\tr_{S_{w(t)}}\Ric
    -2\operatorname{div}\zeta'_{w(t)}
    -2|\zeta'_{w(t)}|^2>0,
\end{equation}
then the null comparison principle holds. Indeed, let $S_{w_0}$ be a MOTS and $S_{w_1}$ a graph with $\ell^+(w_1)\leq0$. Set $h=w_1-w_0$ and consider $w(t)=w_0+th$. Integrating \eqref{second variation formula in introduction} gives $0\geq\ell^+(w_1)-\ell^+(w_0)=\mathcal{L}h$. The operator $\mathcal{L}$ is uniformly elliptic, while \eqref{2nd variation condition 2} makes its zeroth-order term $-c(\theta)h$ with $c(\theta)>0$. At a negative minimum of $h$, the maximum principle would give $\mathcal{L}h>0$, contradicting $\mathcal{L}h\leq0$. Hence{,} $h\geq0$, or equivalently{,} $w_1\geq w_0$.

\vspace{2mm}
{As an application of \eqref{second variation formula in introduction}, we further prove in Section~5 that, in the first author's naked-singularity-censoring setting \cite{An2025}, conditions \eqref{2nd variation condition 1} and \eqref{2nd variation condition 2} can be verified. This gives an alternative proof of the uniqueness of the MOTS and the achronality of the apparent horizon.}

\subsection{Outline of the paper}

The remainder of the paper is organized as follows. In Section~2, we collect the double-null formalism used throughout the paper, including the null structure equations, curvature decompositions, basic estimates on the spheres, commutation formulae, and the construction of the self-similar incoming initial data. In Section~3, we construct the sequence of approximating spacetimes and prove the estimates for the corresponding error terms. This is the technical device {that} separates the self-similar background from the perturbative part of the Einstein--scalar field system. In Section~4, we establish the scale-invariant a priori estimates for the exact solution by comparing it with a sufficiently high-order approximate spacetime; these estimates yield the existence region stated in the first theorem. Finally, in Section~5, we prove the trapped-surface formation theorem. We first treat the focusing estimates and the construction of trapped surfaces, and then discuss the associated marginally outer trapped surfaces, the flow argument, and the null comparison principle.

\subsection{Acknowledgments}
XA is supported by MOE Tier 1 grant A-8002933-00-00. SW is supported by the NUS President Graduate Fellowship.

\section{Preliminaries}\label{Section_Preliminaries}
In this section, we first collect the fundamental equations governing a general Lorentzian manifold in the double null foliation \eqref{metric_expression}. Rather than restricting to an Einstein system, we present the equations for a general Lorentzian manifold $\left(\mathcal{M}^{3+1},g\right)$. We then present the initial data construction for the characteristic initial value problem for the Einstein--scalar field equations.

\subsection{Equations}
We introduce the notation
\[\chic_{AB}=\chi_{AB}-\tr\chi g_{AB},\ \chibc_{AB}=\chib_{AB}-\tr\chib g_{AB}.\] 
For the Ricci coefficients \eqref{Ricci_coefficients_definition} and the Gauss curvature $K$, we list the null propagation equations:
\begin{equation}\label{eq:chi}
    \begin{aligned}
        &\Omega\nabla_4 (\Omega^{-1}\tr\chi)=-\operatorname{Ric}_{44}-|\chic|^2, \\
     &\Omega^{-1}\nabla_3 (\Omega\tr\chi)+\tr\chi\tr\chib=2\sdiv\eta+2\left|\eta\right|^2-2K+R+\Rc_{34},\\
     &\Omega^{-1}\nabla_3(\Omega\chic)_{AB}+\frac{1}{2}\tr\chib\chic_{AB}
        =\nabla_A\eta_B+\nabla_B\eta_A-2(\dv\eta)g_{AB}+Kg_{AB}\\
        &\qquad\qquad
        +2\eta_A\eta_B-2|\eta|^2g_{AB}
        -\frac{1}{2}\tr\chi\chibc_{AB}
        -\frac{1}{2}(R+\Rc_{34})g_{AB}
        +\widehat{\operatorname{Ric}}_{AB},  
    \end{aligned}
\end{equation}
\begin{equation}\label{eq:chib}
    \begin{aligned}
        &\Omega\nabla_3(\Omega^{-1}\tr\chib)=-\Rc_{33}-|\chibc|^2,\\
        &\Omega^{-1}\nabla_4(\Omega\tr\chib)+\tr\chi\tr\chib=2\sdiv\etab+2\left|\etab\right|^2-2K+R+\Rc_{34},\\
        &\Omega^{-1}\nabla_4(\Omega\chibc)_{AB}+\frac{1}{2}\tr\chi\chibc_{AB}
        =\nabla_A\etab_B+\nabla_B\etab_A-2(\dv\etab)g_{AB}+Kg_{AB}\\
        &\qquad\qquad
        +2\etab_A\etab_B-2|\etab|^2g_{AB}
        -\frac{1}{2}\tr\chib\chic_{AB}
        -\frac{1}{2}(R+\Rc_{34})g_{AB}
        +\widehat{\operatorname{Ric}}_{AB}, 
    \end{aligned}
\end{equation}
\begin{equation}\label{eq:eta-etab}
    \nabla_4\eta=-2\tr\chi\zeta-\chic\cdot\zeta+\dv\chic-\Rc_4,\quad
    \nabla_3\etab=2\tr\chib\zeta+\chibc\cdot\zeta+\dv\chibc-\Rc_3,
\end{equation}
\begin{equation}\label{eq:K}
    \begin{aligned}
        &\Omega \nabla_4 K+\Omega\tr\chi K=\dv\dv(\Omega\chic),\quad \Omega \nabla_3 K+\Omega\tr\chib K=\dv\dv(\Omega\chibc).
    \end{aligned}
\end{equation}
\begin{equation*}
    \begin{aligned}
        &\Omega^{-1}e_3(\Omega e_4\log\Omega)-\Delta \log\Omega=\frac{1}{2}(K-\dv\eta)+\frac{1}{2}(K-\dv\etab)\\
        &\qquad\qquad-\Rc_{34}-\frac{1}{2}R-\frac{1}{2}\chic\cdot\chibc+\frac{1}{2}\tr\chi\tr\chib-|\etab|^2+2\eta\cdot\etab.
    \end{aligned}
\end{equation*}
For a scalar function $\phi$, we have the wave operator
\begin{equation*}
    \begin{aligned}
        \square \phi=&\Delta\phi-\Omega^{-1}e_3(\Omega e_4\phi)+2\eta\nabla\phi-\frac{1}{2}\tr\chib e_4\phi-\frac{1}{2}\tr\chi e_3\phi\\
        =&\Delta\phi-\Omega^{-1}e_4(\Omega e_3\phi)+2\etab\nabla\phi-\frac{1}{2}\tr\chib e_4\phi-\frac{1}{2}\tr\chi e_3\phi.
    \end{aligned}
\end{equation*}
For a $1$-tensor $f_A$, we write
\[df_{AB}=\nabla_A f_B-\nabla_Bf_A=2\nabla_{[A}f_{B]}.\]
Using the commutators
\begin{equation}\label{eq:commutators}
         \begin{aligned}
        \left[\Omega \nabla_4,  \nabla_A\right] \phi_{B_1 \cdots B_k}
        =&\sum_{i=1}^k\left(-\nabla_{B_i}(\Omega\chi)_A^C+\nabla^C(\Omega\chi)_{AB_i}\right) \phi_{B_1 \cdots \hat{B}_i C \cdots B_k}-\Omega \chi_A{ }^C  \nabla_C \phi_{B_1 \cdots B_k}, \\
        \left[\Omega \nabla_3, \nabla_A\right] \phi_{B_1 \cdots B_k}
        = &\sum_{i=1}^k\left(-\nabla_{B_i}(\Omega\chib)_A^C+\nabla^C(\Omega\chib)_{AB_i}\right) \phi_{B_1 \cdots \hat{B}_i C \cdots B_k}-\Omega \underline{\chi}_A{ }^C \nabla_C\phi_{B_1 \cdots B_k} ,
        \end{aligned}
\end{equation}
we derive equations for first-order derivatives of the Ricci coefficients:
\begin{equation}\label{eq:dv-chic}
    \begin{aligned}
        &\Omega^{-1}\nabla_3\dv(\Omega\chic)_A
        -\nabla^B(d\eta)_{BA}-\nabla_AK = -\frac{1}{2}\nabla_A(R+\Rc_{34})
        +\nabla^B\widehat{\operatorname{Ric}}_{AB}\\
        &\qquad +2K\eta_A
        +\nabla^B\left(
        2\eta_A\eta_B-2|\eta|^2g_{AB}
        -\frac{1}{2}\tr\chi\chibc_{AB}
        -\frac{1}{2}\tr\chib\chic_{AB}\right)\\ 
        &\qquad+\Omega^{-2}\Big\{
        \left[2\nabla_{[D}(\Omega\chibc)^B{}_{A]}\right]\Omega\chic^D{}_{B}-\Omega\chib^{BD}\nabla_D(\Omega\chic_{AB})\\
        &\qquad\qquad\qquad
        +\Omega\chi_A{}^D\nabla_D(\Omega\tr\chib)
        -\nabla^B(\Omega\chibc_B{}^D)(\Omega\chic_{AD}) \Big\}\\
        &\qquad+(\eta^B+\etab^B)\Big[
        \nabla_A\eta_B+\nabla_B\eta_A
        -2(\dv\eta)g_{AB}+Kg_{AB}
        +2\eta_A\eta_B-2|\eta|^2g_{AB}\\
        &\qquad\qquad\qquad
        -\frac{1}{2}\tr\chi\chibc_{AB}
        -\frac{1}{2}\tr\chib\chic_{AB}
        -\frac{1}{2}(R+\Rc_{34})g_{AB}
        +\widehat{\operatorname{Ric}}_{AB}\Big],
    \end{aligned}
\end{equation}
\begin{equation}\label{eq:dv-chibc}
    \begin{aligned}
        &\Omega^{-1}\nabla_4\dv(\Omega\chibc)_A
        -\nabla^B(d\etab)_{BA}-\nabla_AK=-\frac{1}{2}\nabla_A(R+\Rc_{34})
        +\nabla^B\widehat{\operatorname{Ric}}_{AB}\\
        &\qquad+2K\etab_A
        +\nabla^B\left(
        2\etab_A\etab_B-2|\etab|^2g_{AB}
        -\frac{1}{2}\tr\chib\chic_{AB}
        -\frac{1}{2}\tr\chi\chibc_{AB}\right)\\ 
        &\qquad\quad+\Omega^{-2}\Big\{
        2\nabla_{[D}(\Omega\chic)^B{}_{A]}\Omega\chibc_{DB}-\Omega\chi^{BD}\nabla_D(\Omega\chibc_{AB})\\
        &\qquad\quad\qquad\qquad
        +\Omega\chib_A{}^D\nabla_D(\Omega\tr\chi)
        -\nabla^B(\Omega\chic_B{}^D)(\Omega\chibc_{AD})
        \Big\}\\
        &\qquad\quad+(\eta^B+\etab^B)\Big[
        \nabla_A\etab_B+\nabla_B\etab_A
        -2(\dv\etab)g_{AB}+Kg_{AB}
        +2\etab_A\etab_B-2|\etab|^2g_{AB}\\
        &\qquad\quad\qquad\qquad
        -\frac{1}{2}\tr\chib\chic_{AB}
        -\frac{1}{2}\tr\chi\chibc_{AB}
        -\frac{1}{2}(R+\Rc_{34})g_{AB}
        +\widehat{\operatorname{Ric}}_{AB}\Big],
    \end{aligned}
\end{equation}
\begin{equation}\label{eq:d-eta}
    \begin{aligned}
        &\Omega\nabla_3(d\etab)_{AB}
        -\left(d\,\dv(\Omega\chibc)\right)_{AB}={}\Omega\chibc_A{}^C(d\etab)_{BC}
        +\Omega\chibc_B{}^C(d\etab)_{CA}\\
        &\qquad\qquad\qquad\qquad\qquad+\nabla_A(\Omega\tr\chib)\etab_B
        -\nabla_B(\Omega\tr\chib)\etab_A
        -\Omega\tr\chib(d\etab)_{AB}\\
        &\qquad\qquad\qquad\qquad\qquad+2\left(d(\Omega\tr\chib\,\zeta)\right)_{AB}
        -\left(d(\Omega\Rc_3)\right)_{AB},
    \end{aligned}
\end{equation}
\begin{equation}\label{eq:d-etab}
    \begin{aligned}
        &\Omega\nabla_4(d\eta)_{AB}
        -\left(d\,\dv(\Omega\chic)\right)_{AB}
        ={}\Omega\chic_A{}^C(d\eta)_{BC}
        +\Omega\chic_B{}^C(d\eta)_{CA}\\
        &\qquad\qquad\qquad\qquad\qquad+\nabla_A(\Omega\tr\chi)\eta_B
        -\nabla_B(\Omega\tr\chi)\eta_A
        -\Omega\tr\chi(d\eta)_{AB}\\
        &\qquad\qquad\qquad\qquad\qquad-2\left(d(\Omega\tr\chi\,\zeta)\right)_{AB}
        -\left(d(\Omega\Rc_4)\right)_{AB},
    \end{aligned}
\end{equation}
\begin{equation*}
    \begin{aligned}
        &\Omega\nabla_3\left(\dv \etab-K\right)+\Omega\tr\chib\left(\dv \etab-K\right)=4\dv(\zeta\cdot\Omega\chibh)-\dv(\Omega\Rc_3),\\
        &\Omega\nabla_4\left(\dv \eta-K\right)+\Omega\tr\chi\left(\dv \eta-K\right)=-4\dv(\Omega\chih\cdot\zeta)-\dv(\Omega\Rc_4).
    \end{aligned}
\end{equation*}
Expanding $D^\mu\left(\Rc_{\nu\mu}-\frac{1}{2}Rg_{\nu\mu}\right)$, one derives the following identities for the Ricci curvature components $\Rc_{\mu\nu}$:
\begin{equation}\label{eq:dv-Ric4}
    \begin{aligned}
        &D^\mu \left(\Rc_{4\mu}-\frac{1}{2}Rg_{4\mu}\right)=-\frac{1}{2} \nabla_4\left(\operatorname{Ric}_{34}+R\right)+2 \underline{\omega} \operatorname{Ric}_{44}+2 \eta^A \operatorname{Ric}_{4 A}-\frac{1}{2} \nabla_3 \operatorname{Ric}_{44}+\underline{\eta}^A \operatorname{Ric}_{A 4} \\
    &\qquad+\snab^A \operatorname{Ric}_{A 4}-\frac{1}{2} \operatorname{tr} \underline{\chi} \operatorname{Ric}_{44}-\frac{1}{2} \operatorname{tr} \chi\left(\operatorname{Ric}_{34}+R\right) +\zeta^A \operatorname{Ric}_{A 4}-\frac{1}{2} \operatorname{tr} \chi \operatorname{Ric}_{34}-\hat{\chi}^{A B} \widehat{\operatorname{Ric}}_{A B},
    \end{aligned}
\end{equation}
\begin{equation}\label{eq:dv-Ric3}
    \begin{aligned}
        &D^\mu \left(\Rc_{3\mu}-\frac{1}{2}Rg_{3\mu}\right)=-\frac{1}{2} \nabla_4 \operatorname{Ric}_{33}+2 \omega \operatorname{Ric}_{33}+2 \underline{\eta}^A \operatorname{Ric}_{3 A}-\frac{1}{2} \nabla_3\left(\operatorname{Ric}_{34}+R\right)+\eta^A \operatorname{Ric}_{A 3} \\
    &\qquad+\snab^A \operatorname{Ric}_{A 3}-\frac{1}{2} \operatorname{tr} \chi \operatorname{Ric}_{33}-\frac{1}{2} \operatorname{tr} \underline{\chi}\left(\operatorname{Ric}_{34}+R\right)-\zeta^A \operatorname{Ric}_{A 3}-\frac{1}{2} \operatorname{tr} \underline{\chi} \operatorname{Ric}_{34}-\hat{\chi}^{A B} \widehat{\operatorname{Ric}}_{A B}, 
    \end{aligned}
\end{equation}
\begin{equation}\label{eq:dv-RicA}
    \begin{aligned}
        &D^\mu \left(\Rc_{A\mu}-\frac{1}{2}Rg_{A\mu}\right)=-\frac{1}{2} \nabla_4 \operatorname{Ric}_{3 A}+\underline{\omega} \operatorname{Ric}_{4 A}+\eta^B \operatorname{Ric}_{B A}+\frac{1}{2} \eta_A \operatorname{Ric}_{34}-\frac{1}{2} \nabla_3 \operatorname{Ric}_{4 A}+\omega \operatorname{Ric}_{3 A}\\
        &\qquad+\underline{\eta}^B \operatorname{Ric}_{B A}+\frac{1}{2} \underline{\eta}_A \operatorname{Ric}_{34}+\snab^B \widehat{\operatorname{Ric}}_{B A}+\frac{1}{2} \snab_A \operatorname{Ric}_{34} \\
        &\qquad-\frac{1}{2} \operatorname{tr} \underline{\chi} \operatorname{Ric}_{4 A}-\frac{1}{2} \operatorname{tr} \chi \operatorname{Ric}_{3 A}-\frac{1}{2} \underline{\chi}_A{ }^B \operatorname{Ric}_{B 4}-\frac{1}{2} \chi_A{ }^B \operatorname{Ric}_{B 3}.
    \end{aligned}
\end{equation}
These equations hold for any symmetric tensor $T_{\mu\nu}$, and their left-hand sides vanish for $\Rc_{\mu\nu}$ by the contracted Bianchi identity $D^\mu\left(\Rc_{\nu\mu}-\frac{1}{2}Rg_{\nu\mu}\right)\equiv 0.$

\subsection{$\kappa$-self-similarity and initial data on an ingoing cone}

In this part, we prescribe the initial data along $\Hb_0$. 

\begin{definition}\label{self-similar initial value}
    We say that the 5-tuple $(g_{AB},\Omega,\zeta_A,b^A, \phi)$ on $\Hb_0$ is {an} \textbf{$(\epsilon,N)$-small $\k$-self-similar initial data set} if 
    \begin{enumerate}
        \item The data are $\k$-self-similar, meaning that 
        $$g_{AB}(u,0)=(-u)^2g_{AB}(-1,0),\Omega^2(u,0)=(-u)^\k\Omega^2(-1,0),$$
        $$b^A(u,0)=(-u)^{-1}b^A(-1,0),\zeta_A(u,0)=\zeta_A(-1,0),$$
         $$\left((\partial_u+b)\phi\right)(u,0)=(-u)^{-1}\left((\partial_u+b)\phi\right)(-1,0),\left(\partial_{\theta^A}\phi\right)(u,0)=\left(\partial_{\theta^A}\phi\right)(-1,0).$$
    \item The following inequality holds: $\lnm {g}(-1,0)-g^{\SS}, b(-1,0), \log\Omega(-1,0),\phi(-1,0)\rnm_{H^{N+1}(\SS)}\leq \epsilon.$
    \item These quantities satisfy the constraint equations for the Einstein--scalar field equations along $\Hb_0$, namely, the $\nabla_3\zeta$ and $\nabla_3\chib$ equations.
    \end{enumerate}

\end{definition}

To construct such initial data sets, we first study the $\Lie$ propagation equation on the sphere, which governs the solvability of the constraint equations in the self-similar setting. The argument mirrors the construction in \cite{RodnianskiShlapentokhRothman2019}. A proof of Lemma \ref{lemma:Lie_propagation_eq} can be found there, and Proposition \ref{prop:initial_data_construction} generalizes that argument \cite{RodnianskiShlapentokhRothman2019} to the Einstein--scalar field system.
\begin{lemma}\label{lemma:Lie_propagation_eq}
    Let $M\geq 2$ be an integer. Fix $k$, and let $F_{A_1\cdots A_k}\in H^M(\mathcal{T}^{(0,k)}(\SS))$. Consider the equation 
    \begin{equation}\label{Lie_propagation_equation}
        u+\Lie_X u+h\cdot u=F,
    \end{equation}
    where $(h\cdot u)_{A_1\cdots A_k}=\sum_{l=0}^k {h^{(i)}}_{A_1\cdots A_i}^{B_1\cdots B_i}u_{B_1\cdots B_i A_{i+1}\cdots A_k}$.
    {If} the vector field $X$ and the tensors $h^{(i)}$ satisfy 
    \begin{equation*}
        \lnm X\rnm_{H^{M+1}(\SS)}+\sum_{i=0}^k\lnm h^{(i)} \rnm_{H^{M}(\SS)}\ll 1,
    \end{equation*}
    then equation \eqref{Lie_propagation_equation} admits a unique solution $u$ satisfying 
    \begin{equation*}
        \lnm u\rnm_{H^{M}(\SS)}+\lnm \Lie_X u\rnm_{H^M(\SS)}\leq C(k,M)\lnm F\rnm_{H^{M}(\SS)}.
    \end{equation*}
\end{lemma}

\begin{proposition}\label{prop:initial_data_construction}
    Fix constants $0<\k<1$ and $\epsilon\ll \k$, and consider a tuple $(g,\Omega,\phi,\Pi_{\cl}b)$ at $S_{-1,0}$ satisfying 
    $$\lnm {g}(-1,0)-g^{\SS},\log\Omega(-1,0),\phi(-1,0),\Pi_{\cl}b\rnm_{H^{N+1}(\SS)}\leq \epsilon,$$
    where $X=\Pi_{\cl}X+\Pi_{\dv}X$ is a decomposition with $\dv X=\dv\Pi_{\dv}X$ and $\cl X=\cl\Pi_{\cl}X$.
    Then there exists an $(\epsilon,N)$-small $\k$-self-similar initial value set generated by $(g,\Omega,\phi,\Pi_{\cl}b)$.
\end{proposition}
\begin{proof}
    Imposing the self-similarity condition, we can reduce the $\nabla_3\tr\chib$ equation to {the following equation} on $S_{-1,0}$:
    \begin{equation*}
        \begin{aligned}
            \dv b-\Lie_b\dv b-\frac{1}{2}(\dv b)^2=&\frac{1}{4}\left|\nabla\hat\otimes b\right|^2-2\k+4\Lie_b\log\Omega-2\dv b\Lie_b\log\Omega\\
            &+(\partial_u\phi)^2+2\partial_u\phi \Lie_b\phi+(\Lie_b\phi)^2{.}
        \end{aligned}
    \end{equation*}
    If $\phi(u,0)=-\sqrt{2\wt\k+2\k}\log(-u)+\phi(-1,0)$ and $b=\nabla f+\Pi_{\cl}b$, then we can write the equation as 
\begin{equation*}
    \Delta f-\Lie_b\Delta f=\frac{1}{4}|\nabla\hat\otimes \nabla f|^2+\wt\k+O(\epsilon^2)\wt\k+O(\epsilon)\nabla f+O(\epsilon^2)(\nabla f)^2.
\end{equation*}  
We solve this equation via an iteration scheme. Consider sequences of functions $\{D_i\},\{f_i\}$ and constants $\{\wt\k_i\}$, initialized by $D_0=f_0=\wt\k_0=0$.
Writing $b_i=\nabla f_i+\Pi_{\cl}b$, we define $\wt{D}_{i+1}(\wt{\k},\theta)$ as the solution of
\begin{equation}\label{dv_b_constraint_iteration_D_i}
   \begin{aligned}
     \wt{D}_{i+1}-\Lie_{b_i}\wt{D}_{i+1}=&\frac{1}{2}(\dv b_i)^2+\frac{1}{4}\left|\nabla\hat\otimes b_i\right|^2+4\Lie_{b_i}\log\Omega-2\dv b_i\Lie_{b_i}\log\Omega\\
            &+2\wt\k+2\left({2\wt{\k}_i+2\k}\right)^{1/2} \Lie_{b_i}\phi+(\Lie_{b_i}\phi)^2.
   \end{aligned}
\end{equation}
We then choose $\wt{\k}=\wt{\k}_{i+1}$ so that $\int_{S_{-1,0}}\wt{D}_{i+1}(\wt{\k})dV_g=0$ and set $D_{i+1}=\wt{D}_{i+1}\left(\wt{\k}_{i+1}\right)$. The function $f_{i+1}$ is then determined by
\begin{equation*}
    \Delta f_{i+1}=D_{i+1},\ \int_{S_{-1,0}}f dV_g=0.
\end{equation*} 
Assuming $\lnm D_{i}\rnm_{H^{N}}\leq A_i\epsilon$ and $\left|\wt\k_i\right|\leq B_i\epsilon$,
standard elliptic theory gives $\lnm f_i\rnm_{H^{N+2}}\leq C_0 A_i\epsilon$. To estimate $D_{i+1}$ and $\wt\k_{i+1}$, we rewrite \eqref{dv_b_constraint_iteration_D_i} as
\begin{equation}\label{dv_b_constraint_iteration_D_i_11111}
    \left(\wt{D}_{i+1}-2\wt{\k}\right)-\Lie_{b_i}\left(\wt{D}_{i+1}-2\wt{\k}\right)\sim C(A_i,B_i)O_{H^N}(\epsilon^2).
\end{equation}
Since the right-hand side of \eqref{dv_b_constraint_iteration_D_i_11111} is independent of $\wt{\k}$, it follows that
\begin{equation*}
    \wt{D}_{i+1}(\wt{\k})=\wt{D}_{i+1}(0)+2\wt{\k},
\end{equation*}
with $\lnm \wt{D}_{i+1}(0)\rnm_{H^{N}}\leq C(A_i,B_i)\epsilon^2.$ Therefore{,} we have $\wt{\k}_{i+1}=-\frac{1}{2{\rm Area}(S_{-1,0})}\int_{S_{-1,0}}\wt{D}_{i+1}(0)$ with $$\left|\wt{\k}_{i+1}\right|\leq C(A_i,B_i)\epsilon^2,$$
and hence $\lnm D_{i+1}\rnm_{H^N}\leq C(A_i,B_i)\epsilon^2,$ $\lnm f_{i+1}\rnm_{H^{N+2}}\leq C(A_i,B_i)\epsilon^2$. Consequently, for all $i$ we obtain the uniform bounds
\begin{equation}\label{dv_b_constraint_iteration_D_i_uni_bound}
    \lnm f_i\rnm_{H^{N+2}}+\lnm D_{i}\rnm_{H^{N}}+\left|\wt\k_i\right|\leq C_0\epsilon^2.
\end{equation}
To demonstrate convergence as $i\to\infty$, we introduce the error quantity $$E_{i,j}=\wt{D}_{i+j+1}(\wt{\k}^{(i+j+1)})-2\wt{\k}^{(i+j+1)}-\wt{D}_{i+1}(\wt{\k}^{(i+1)})+2\wt{\k}^{(i+1)},$$ which satisfies
\begin{equation}\label{dv_b_constraint_iteration_D_i_22222}
    \begin{aligned}
        E_{i,j}-&\Lie_{b_i}E_{i,j}\sim\Lie_{\nabla f_{i+j}-\nabla f_i} D_{i+j}+\left(\nabla^2 (f_{i+j}-f_i)\right)\left(\nabla^2 (f_{i+j}+f_i)+2\nabla\Pi_\cl b\right)\\
        &+\nabla(f_{i+j}-f_i)\nabla\log\Omega-\Delta (f_{i+j}-f_i)\Lie_{b_i}\log\Omega-\Delta f_{i+1}\Lie_{\nabla(f_{i+j}-f_i)}\log\Omega\\
        &+\left(\sqrt{\wt\k_{i+j}+\k}-\sqrt{\wt\k_i+\k}\right)\Lie_{b_i}\phi+\sqrt{\wt\k_{i+j}+\k}\Lie_{\nabla(f_{i+j}-f_i)}\phi+\Lie_{\nabla(f_{i+j}-f_i)}\phi\Lie_{b_i+b_{i+j}}\phi.
    \end{aligned}
\end{equation}
From the definition of $\wt{\k}_{i+j}$ and the mean-value property, we deduce
\begin{equation*}
    \left|\wt{\k}_{i+j}-\wt\k_{i}\right|\leq \frac{1}{2{\rm Area}(S_{-1,0})}\int_{S_{-1,0}}\left|\wt{D}_{i+j}(0)-\wt{D}_{i}(0)\right|\leq C_0\lnm E_{i-1,j}\rnm_{H^{N}},
\end{equation*}
\begin{equation*}
    \lnm D_{i+j}-D_i\rnm_{H^{N}}\leq \lnm E_{i-1,j}\rnm_{H^{N}}+2\left|\wt\k_{i+j}-\wt\k_i\right|\leq  C_0\lnm E_{i-1,j}\rnm_{H^{N}},
\end{equation*}
\begin{equation*}
        \lnm f_{i+j}-f_i\rnm_{H^{N+2}}\leq C_0\lnm D_{i+j}-D_i\rnm_{H^{N}}\leq C_0\lnm E_{i-1,j}\rnm_{H^{N}}.
\end{equation*}
Combining \eqref{dv_b_constraint_iteration_D_i_22222} with the uniform bound \eqref{dv_b_constraint_iteration_D_i_uni_bound}, we arrive at
\begin{equation*}
    \lnm E_{i,j}\rnm_{H^{N}}\leq C(\k^{-1})\epsilon \lnm E_{i-1,j}\rnm_{H^{N}}\leq \frac{1}{2}\lnm E_{i-1,j}\rnm_{H^{N}}\leq \frac{1}{2^i}C_0\epsilon^2
\end{equation*}
for any $i,j$. Hence all three sequences are Cauchy in their respective norms and converge to the desired limit.

It remains to determine the torsion. Combining the $\nabla_3\etab$ equation, the $\betab$-Codazzi equation, and the relation $\etab=\eta-2\nabla\log\Omega$, we obtain under self-similarity
\begin{equation*}
    \Lie_b \eta_A+\left(-2+\dv b\right)\eta_A=2\nabla_A(\Lie_b\log\Omega)-\dv(\nabla\hat\otimes b)_A+\frac{1}{2}\nabla_A\dv b-\Omega e_3\phi \nabla_A\phi,
\end{equation*}
and therefore $\zeta=\eta-\nabla\log\Omega$ can be determined.

\end{proof}

Next, we prescribe $\Omega\chi$ along $\Hb_0$. To this end, we rewrite the $\nabla_3(\Omega^{-1}\tr\chi)$ equation as
\begin{equation*}
    \Omega\nabla_3(\Omega^{-1}\tr\chi)+(\Omega\tr\chib-4\Omega\omegab)\Omega^{-1}\tr\chi= -2K+2\dv\eta+2|\eta|^2+|\nabla\phi|^2,
\end{equation*}
which, in the context of self-similarity, reduces to the following $\Lie_b$-equation at $S_{-1,0}$:
\begin{equation}\label{eq:Lie_eq_trchi}
    \Lie_b (\Omega^{-1}\tr\chi)-\left(1-\dv b+\k+2\Lie_b\log\Omega\right) (\Omega^{-1}\tr\chi)=\left(-2K+2\dv\eta+2|\eta|^2+|\nabla\phi|^2\right).
\end{equation}
For $\Omega^{-1}\chih$, self-similarity reduces {the} $\nabla_3$ equation to 
\begin{equation}\label{eq:Lie_eq_chih}
    \begin{aligned}
        &\Lie_b(\Omega^{-1}\chih_{AB})-\left(\frac{1}{2}\dv b+\k-2\Lie_b\log\Omega\right)\Omega^{-1}\chih_{AB}\\
        =&\frac{1}{2}(\nabla\hat\otimes b)^C_{\  A}(\Omega^{-1}\chih)_{BC}+ \frac{1}{2}(\nabla\hat\otimes b)^C_{\  B}(\Omega^{-1}\chih)_{AC}  +(\nabla\hat\otimes\eta)_{AB}+(\eta\hat\otimes\eta)_{AB}\\
        &-\frac{1}{2}(\nabla\hat\otimes b)_{AB}(\Omega^{-1}\tr\chi)+D_A\phi D_B\phi-\frac{1}{2}g_{AB}D_C\phi D^C\phi.
    \end{aligned}
\end{equation}
By Lemma \ref{lemma:Lie_propagation_eq}, the equations \eqref{eq:Lie_eq_trchi} and \eqref{eq:Lie_eq_chih} are solvable. {This completes the construction of the incoming initial data.}
\begin{proposition}\label{prop:self-similar-initial-data}
    Fix constants $0<\k<1$ and $\epsilon\ll \k$. Let $(g,\Omega,\zeta,b,\phi)$ be the $(\epsilon, N)$-small $\k$-self-similar initial data along $\Hb_0$ generated by the tuple $(g,\Omega,\phi,\Pi_{\cl}b)$ satisfying
    $$\lnm {g}(-1,0)-g^{\SS},\log\Omega(-1,0),\phi(-1,0),\Pi_{\cl}b\rnm_{H^{N+1}(\SS)}\leq \epsilon.$$ 
    There exists a unique $\Omega^{-1}\chi$ along $\Hb_0$ satisfying the self-similarity condition
    $$\Omega^{-1}\chi_{AB}(u,0)=(-u)\Omega^{-1}\chi_{AB}(-1,0).$$ 
    Specifically, its trace and trace-free parts are obtained from \eqref{eq:Lie_eq_trchi} and \eqref{eq:Lie_eq_chih}, respectively.

    For the Ricci coefficients $\psi\in\{\Omega\chib,\Omega^{-1}\chi,\eta,\etab\}$, we have the estimates
    \[\sum_{1\leq i\leq N-1}\lnm (-u)^{i}\nabla^i\psi\rnm_{L^2(S_{u,0})}\lesssim\epsilon.\]
\end{proposition}

\subsection{Short-pulse initial data on an outgoing cone}
Suppose that $\Omega$ is given on $H_{-1}$ with $\Omega\omega(-1,v)$ uniformly bounded. We need to prescribe $(g,\phi)$ along $H_{-1}$. The shift $b$ then follows from the transport equations:
\begin{equation*}
    \nabla_4\zeta=-\nabla_4\nabla\log\Omega-2\tr\chi\zeta-\chic\cdot\zeta+\dv\chic-\nabla\phi\nabla_4\phi,\quad \Lie_v b^A=-4\Omega^2\zeta^A.
\end{equation*}
To generate a trapped surface, the short pulse should be sufficiently strong. In this article, we require the perturbation to satisfy
\begin{equation}\label{cond:initial-outgoing-data}
    |\nabla_4\phi(-1,v,\theta)|+|\chih(-1,v,\theta)|\geq v^\delta,\ \delta<1,
\end{equation}
for $\theta\in U\subset \mathbb{S}^2$ {and} $v\ll 1${.}
We {choose an arbitrary} function $f(v,\theta)$ and a $g(-1,0)$-trace-free symmetric tensor $X_{AB}(v,\theta)$ such that 
\[|f(v,\theta)|+|X(v,\theta)|\geq v^\delta\]
for $\theta\in U\subset \mathbb{S}^2.$
Set $\phi(-1,v,\theta)=\phi(-1,0,\theta)+\int_0^v\Omega f(-1,v^\prime,\theta)dv^\prime$.
We then define $g$ via ODEs:
\begin{equation*}
    \left\{\begin{aligned}
        &\Lie_v g_{AB}=\Omega\tr\chi g_{AB}+2\Omega\chih_{AB},\\
        &\Lie_v (\Omega\tr\chi)=-\frac{1}{2}\left(\Omega\tr\chi\right)^2-\left|\Omega\chih\right|^2-(\partial_v\phi)^2-4\Omega\omega\Omega\tr\chi,\\
        &\Omega\chih_{AB}=\frac{1}{2}\left(X_{AC}g^{CD}(-1,0)g_{DB}+X_{BC}g(-1,0)^{CD}g_{DA}\right).
    \end{aligned}\right.
\end{equation*}
It follows directly that 
\[\left|g-g|_{v=0}\right|\lesssim v,\]
and thus 
\begin{equation*}
    \left|\Omega\chih-X\right|\leq \left|X\cdot g^{-1}\cdot (g-g|_{v=0})\right|\lesssim v.
\end{equation*}
The requirement \eqref{cond:initial-outgoing-data} is satisfied.

 \section{Construction of Approximating Spacetimes}\label{Section_Approximating_metric}
\subsection{Construction}
Suppose that we have $(\epsilon,N)$-small $\k$-self-similar initial data as in Proposition \ref{prop:self-similar-initial-data} along $\Hb_0$ and $(\Omega,\chi,e_4\phi)$ along $H_{-1}$. Assume $1\ll N_0\ll N$. Our goal is to construct an approximating Lorentzian spacetime $(\mathcal{M},\ot{g},\ot\phi)$ with the assigned initial data and a rapidly decaying curvature residue
\[\Rc\left(\ot{g}\right)-d\ot{\phi}\otimes d\ot{\phi}=o_{v\rightarrow 0}(v^\infty).\]

We construct a sequence of approximating Lorentzian metrics $g^{(i)}$ and scalar fields $\phi^{(i)}$ by an iterative procedure. The zeroth-order data $g^{(0)},b^{(0)}$, $\Omega^{(0)}$, $\chib^{(0)},\zeta^{(0)},(\nabla_A\phi)^{(0)},(\Omega D_3\phi)^{(0)}$ are taken to be their values at $v=0$, while all remaining connection quantities are initialized as $\psi^{(0)}=0$.

In \textbf{STEP i+1}, all tensor contractions are performed with respect to $g^{(i)}$.
The components of $g^{(i+1)}$ are obtained by solving the following sequence of equations.
First, from the $\nabla_3\chih$ equation, we define $(\Omega\chih)^{(i+1/2)}_{AB}$ by
\begin{equation}\label{Construction_Eq_chih}
    \left\{
    \begin{aligned}
        & (\Omega\nabla_3)^{(i)}(\Omega\chih^{(i+1/2)})_{AB}+\frac{1}{2}(\Omega\tr\chib)^{(i)}(\Omega\chih^{(i+1/2)})_{AB}\\
        &\qquad =(\Omega^{(i)})^2\left(\left(\nabla^{(i)}{\hat\otimes}^{(i)}\eta^{(i)}+\eta^{(i)}{\hat\otimes}^{(i)}\eta^{(i)}+\nabla\phi^{(i)}{\hat\otimes}^{(i)}\nabla\phi^{(i)}\right)_{AB}-\frac{1}{2}(\tr\chi)^{(i)}{\chibh^{(i)}}_{AB}\right),\\
        &{\Omega\chih^{(i+1/2)}}_{AB}(-1,v)=\frac{1}{2}\left(\Omega\chih_{AC}g^{CD}{g^{(i)}}_{DB}+\Omega\chih_{BC}g^{CD}{g^{(i)}}_{DA}\right)(-1,v).
    \end{aligned}
    \right.
\end{equation}
Second, the scalar field derivative $\partial_v\phi^{(i+1)}$ is determined from the $\nabla_3\nabla_4\phi$ equation:
\begin{equation}\label{Construction_Eq_D4phi}
    \left\{\begin{aligned}
        &(\Omega e_3)^{(i)}(\partial_v\phi^{(i+1)})+\frac{1}{2}(\Omega\tr\chib)^{(i)}(\partial_v\phi^{(i+1)})\\
        &\qquad =(\Omega^{(i)})^2\Delta^{(i)}\phi^{(i)}-\frac{1}{2}(\Omega\tr\chi)^{(i)}(\Omega e_3)^{(i)}\phi^{(i)}+2(\Omega^{(i)})^2\eta^{(i)}\nabla\phi^{(i)},\\
        &\partial_v\phi^{(i+1)}(-1,v)=\partial_v\phi(-1,v),\ \ \phi^{(i+1)}(u,0)=\phi(u,0).
    \end{aligned}\right.
\end{equation}
Third, from the $\nabla_4\omegab$ equation, we define $\left(\Omega\omegab\right)^{(i+1/2)}$ via
\begin{equation}\label{Construction_Eq_omegab}
    \left\{\begin{aligned}
        &\partial_v(\Omega\omegab)^{(i+1/2)}=(\Omega^{(i)})^2\left(\left|\eta^{(i)}\right|^2-\eta^{(i)}\etab^{(i)}\right)+\frac{1}{4}(\Omega e_3 \phi)^{(i)}\partial_v\phi^{(i+1)}\\
        &\quad -\frac{1}{2}(\Omega^{(i)})^2\left(K^{(i)}-\frac{1}{2}\left|\nabla\phi^{(i)}\right|^2\right)+\frac{1}{4}\left((\Omega\chih)^{(i+1/2)}(\Omega\chibh)^{(i)}-\frac{1}{2}(\Omega\tr\chi)^{(i)}(\Omega\tr\chib)^{(i)}\right)\\
        &(\Omega\omegab)^{(i+1/2)}(u,0)=\Omega\omegab(u,0).
    \end{aligned}\right.
\end{equation}
Since $e_3\log\Omega=-2\omegab$, we can recover $\Omega^{(i+1)}$ from $(\Omega\omegab)^{(i+1/2)}$ as follows:
\begin{equation}\label{Construction_Eq_logOmega}
    \left\{
    \begin{aligned}
        &(b^{(i+1/2)})^A=(b_0)^A-\int_0^v4\left(\Omega^{(i)}\right)^2{\zeta^{(i)}}^A,\\
        &\left(\partial_u+b^{(i+1/2)}\right)\log \Omega^{(i+1)}=-2(\Omega\omegab)^{(i+1/2)},\\
        &\Omega^{(i+1)}(-1,v)=\Omega(-1,v).
    \end{aligned}
    \right.
\end{equation}
With $\Omega^{(i+1)}$, we write 
\begin{equation*}
    (\Omega\omega)^{(i+1)}=-\frac12\partial_v\log\Omega^{(i+1)}.
\end{equation*}
Using the $\nabla_4\eta$ equation together with the relation $\zeta=\eta-\nabla\log\Omega$, we have the equation for $\nabla_4\zeta$:
\begin{equation*}
    \begin{aligned}
        \Omega\nabla_4\zeta+\frac{3}{2}\Omega\tr\chi\,\zeta+\Omega\chih\cdot\zeta=2\nabla(\Omega\omega)+\dv(\Omega\chih)-\frac{1}{2}\nabla(\Omega\tr\chi)+\Omega\tr\chi\nabla\log\Omega-\Omega\Rc_{4\cdot}.
    \end{aligned}
\end{equation*}
We thus define $\left(\zeta^{(i+1)}\right)^A$ as the solution of
\begin{equation}\label{Construction_Eq_zeta}
    \left\{\begin{aligned}
        &\Lie_v(\zeta^{(i+1)})^A=-2(\Omega\tr\chi)^{(i)}(\zeta^{(i)})^A-2(g^{(i)})^{AC}\left((\Omega\chih)^{(i+1/2)}\right)_{CB}(\zeta^{(i)})^B\\
        &\qquad+2(g^{(i)})^{AB}\partial_B(\Omega\omega)^{(i+1)} +(\dv)^{(i)}(\Omega\chih)^{(i+1/2)}_B(g^{(i)})^{BA} \\
        &\qquad-\frac{1}{2}(g^{(i)})^{AB}\partial_B(\Omega\tr\chi)^{(i)}+(\Omega\tr\chi)^{(i)}(g^{(i)})^{AB}\partial_B\log\Omega^{(i+1)}\\
        &\qquad-\partial_v\phi^{(i+1)}\partial_B\phi^{(i)}(g^{(i)})^{AB},\\
        &(\zeta^{(i+1)})^A(u,0)=\zeta^A(u,0).
    \end{aligned}\right.
\end{equation}
We remark that it is necessary to define $\zeta^A$ with an upper index, since the shift equation $\Lie_v b^A=-4\Omega^2\zeta^A$ involves the upper-index quantity.
\begin{equation*}
    \left\{\begin{aligned}
        &\Lie_v (b^{(i+1)})^A=-4(\Omega^{(i+1)})^2(\zeta^{(i+1)})^A,\\
        &(b^{(i+1)})^A(u,0)=b^A(u,0).
    \end{aligned}\right.
\end{equation*}
Finally, the induced metric and the outgoing expansion are determined by the coupled system
\begin{equation}\label{Construction_Eq_g_chi}
    \left\{
    \begin{aligned}
        &\Lie_v(g^{(i+1)})_{AB}=(\Omega^{(i+1)})^2(\Omega^{-1}\tr\chi)^{(i+1)}(g^{(i+1)})_{AB}+2{(\Omega\chih)^{(i+1)}}_{AB},\\
        &\partial_v\left(\Omega^{-1}\tr\chi\right)^{(i+1)}=-\frac{1}{2}(\Omega^{(i+1)})^2\left((\Omega^{-1}\tr\chi)^{(i+1)}\right)^2-(\Omega^{(i+1)})^{-2}\left(\partial_v\phi^{(i+1)}\right)^2\\
        &\qquad\qquad-(\Omega^{(i+1)})^{-2}{(\Omega\chih)^{(i+1)}}_{AB}{(\Omega\chih)^{(i+1)}}_{CD}{g^{(i+1)}}^{AC}{g^{(i+1)}}^{BD},\\
        &(\Omega\chih)^{(i+1)}_{AB}=\frac{1}{2}\left((\Omega\chih)^{(i+1/2)}_{AC}{g^{(i+1)}}_{DB} +(\Omega\chih)^{(i+1)}_{CB}{g^{(i+1)}}_{DA}\right){(g^{(i)})}^{CD},\\
        &(g^{(i+1)})_{AB}(u,0)=g_{AB}(u,0),\ \ (\Omega^{-1}\tr\chi)^{(i+1)}(u,0)=\Omega^{-1}\tr\chi(u,0).
    \end{aligned}
    \right.
\end{equation}
Assembling the above, we obtain the full spacetime metric $g^{(i+1)}$:
\begin{equation*}
    \begin{aligned}
        {g^{(i+1)}}=&-2(\Omega^{(i+1)})^2(du\otimes dv+\dv\otimes du)\\
        &\qquad+{g^{(i+1)}}_{AB}\left(d\theta^A-{(b^{(i+1)})}^Adu\right)\otimes \left(d\theta^B-{(b^{(i+1)})}^Bdu\right),
    \end{aligned}
\end{equation*}
with {the} null frame ${e_3^{(i+1)}}=(\Omega^{(i+1)})^{-1}\left(\partial_u+(b^{(i+1)})^A\partial_{\theta^A}\right)$, $e_4^{(i+1)}=(\Omega^{(i+1)})^{-1}\partial_v$, $e_A=\partial_{\theta^A}$.
The connection and curvature components with respect to $g^{(i+1)}$ are denoted by $\cdot^{(i+1)}$.
We write $$\wt{\Rc}^{(i)}_{\mu\nu}=\Rc^{(i)}_{\mu\nu}-e_\mu^{(i)}\phi^{(i)}\cdot e_{\nu}^{(i)}\phi^{(i)},\quad \wt{\square\phi}^{(i)}=\square_{g^{(i)}}\phi^{(i)}.$$
The solvability of these equations is straightforward, since each equation is an ODE with known quantities on the right-hand side.

To verify that these approximating metrics converge to a solution of the {Einstein--scalar field} system as $i\to\infty$, we require not only estimates for $\wt{\Rc}$ but also control of the initial values along $H_{-1}$.
\begin{proposition}\label{prop:sec-3-H-initial-g-chi-estimate}
    Along $u=-1$, suppose we {are given an} initial data set $(g_{AB},\Omega,\phi)$ satisfying
    \begin{equation*}
        \partial_v\left(\Omega\tr\chi\right)=-\frac{1}{2}(\Omega\tr\chi)^2-\left|\Omega\chih\right|^2-4\Omega\omega\Omega\tr\chi-(\partial_v\phi)^2.
    \end{equation*}
    Let $g^{(0)}=g(-1,0)$, $(\Omega\chih)_{AB}^{(i+1/2)}=\frac{1}{2}\left(\Omega\chih_{AC}g^{CD}{g^{(i)}}_{DB}+\Omega\chih_{BC}g^{CD}{g^{(i)}}_{DA}\right)$, and let $g^{(i+1)}$, $(\Omega\chi)^{(i+1)}$ be defined by \ref{Construction_Eq_g_chi} with $\Omega^{(i+1)}=\Omega$ {and} $\partial_v\phi^{(i)}=\partial_v\phi$. Suppose that the bound $$\sup_v\lnm \Omega\chi,\Omega\omega,\partial_v\phi\rnm_{W^{N_0,\infty}(S_{-1,v})}<\infty,$$ holds. Then
    \begin{equation*}
        \lnm g^{(i)}-g, (\Omega\tr\chi)^{(i)}-\Omega\tr\chi, (\Omega\chih)^{(i+1)}-\Omega\chih\rnm_{W^{N_0,\infty}(S_{-1,v})}\lesssim v^{i+1}.
    \end{equation*}
\end{proposition}
\begin{proof}
    The base case is immediate:
    $$\lnm g^{(0)}-g, (\Omega\tr\chi)^{(0)}-\Omega\tr\chi\rnm_{W^{N_0,\infty}(S_{-1,v})}\lesssim v$$
    follows from the finiteness of $\lnm\Omega\chi,\Omega\omega,\Omega\chih\rnm_{L^\infty_vH^{N_0}}$. Proceeding by induction, suppose that we already have
    $$\lnm g^{(i)}-g, (\Omega\tr\chi)^{(i)}-\Omega\tr\chi\rnm_{W^{N_0,\infty}(S_{-1,v})}\lesssim v^{i+1}.$$ 
    We then estimate $g^{(i+1)}-g$ and $(\Omega\chih)^{(i+1)}-\Omega\chih$.
    By the definition of the iteration scheme, we observe that
    \begin{equation}\label{3_11}
        \Omega\chih^{(i+1)}-\Omega\chih^{(i+1/2)}\sim \Omega\chih^{(i+1/2)}\left(g^{(i+1)}-g^{(i)}\right),\ \ \Omega\chih^{(i+1/2)}-\Omega\chih\sim \Omega\chih\left(g^{(i)}-{g}\right).
    \end{equation}
    It follows that
    \[
        \lnm \Omega\chih^{(i+1/2)}-\Omega\chih \rnm_{W^{N_0,\infty}(S_{-1,v})}\lesssim v^{i+1},
    \]
    \[\lnm \Omega\chih^{(i+1)}-\Omega\chih^{(i+1/2)} \rnm_{W^{N_0,\infty}(S_{-1,v})}\lesssim v^{i+1}+\lnm g^{(i+1)}-g^{(i)}\rnm_{W^{N_0,\infty}(S_{-1,v})}.\]
    Taking differences in the defining equations, we obtain
    \begin{equation*}
        \begin{aligned}
            &\partial_v\left(\Omega\tr\chi^{(i+1)}-\Omega\tr\chi\right)\\
            &\quad\sim (\Omega\tr\chi,\Omega\omega)\left(\Omega\tr\chi^{(i+1)}-\Omega\tr\chi\right)+\Omega\chih\left(\Omega\chih^{(i+1)}-\Omega\chih\right)+(\Omega\chih)^2\left(g^{(i+1)}-g\right),\\
            & \Lie_v \left(g^{(i+1)}_{AB}-g_{AB}\right)\sim\Omega\chi\left(g^{(i+1)}-g^{(i)}\right)+g\left(\Omega\chi^{(i+1)}-\Omega\chi\right).
        \end{aligned}
    \end{equation*}
    Applying the $e_4$-transport estimate 
    \[\lnm f(v)-f(0)\rnm_{L^\infty(\mathbb{S})}\leq \int_0^v\lnm \Lie_v f(v^\prime)\rnm_{L^\infty(\mathbb{S})}dv^\prime\]
    yields the desired bounds for $g^{(i+1)}-g$ and $\Omega\tr\chi^{(i+1)}-\Omega\tr\chi$, and the estimate for $\Omega\chih^{(i+1)}-\Omega\chih$ then follows from \ref{3_11}.
\end{proof}

\subsection{First-step estimates} In this section, we establish useful estimates for the first approximate spacetime $\{(\mathcal{M},g^{(1)},\phi^{(1)})\}$.

\vspace{2mm}
Instead of estimating the $L^2(S)$ norm, we use the $L^\infty$ norm directly for convenience.
\begin{lemma}\label{lemma:sec3-transport-estimate}
    For a Lorentzian metric $g$ of the form \eqref{metric_expression}, consider the transport equation $\Omega \nabla_3 f-\frac{\lambda}{(-u)}f=F$ for tensors $f$ and $F$. Then we have the estimate
    \begin{equation*}
        \lnm f(u,v)\rnm_{L^\infty(S_{u,v})}\leq \frac{1}{(-u)^\lambda}\lnm f\rnm_{L^\infty(S_{-1,v})}+\frac{1}{(-u)^\lambda}\int_{-1}^u (-u^\prime)^\lambda\lnm F\rnm_{L^\infty(S_{u^\prime,v})}du^\prime.
    \end{equation*}
\end{lemma}
\begin{proof}
Along $\Hb_v$, $\Omega^{-1}e_3$ is the tangent vector field {to} affinely parametrized geodesics, which means that $D_{\Omega^{-1}e_3}(\Omega^{-1}e_3)$ vanishes. It is natural to introduce the geodesic coordinates $(u,v,\theta)\rightarrow (S(u,v,\theta),V,\Theta)$ defined by
\begin{equation*}
    (u,v,\theta)=\exp_{(-1,V,\Theta)}\left(S\cdot (\Omega^{-1}e_3)|_{(-1,V,\Theta)}\right),
\end{equation*}
where $(-1,V,\Theta)$ is the point on $H_{-1}$ in double-null coordinates. Therefore, we have $\partial_S=\Omega^{-1}e_3.$ Since $u(0,V,\Theta)=-1$ and $\partial_Su=\Omega^{-2}$, we then have
$$u(S,V,\Theta)=\int_0^S\Omega^{-2}(S^\prime,V,\Theta)dS^\prime.$$
We further write $U=u(S,V,\Theta)$ and consider the coordinates $(U,V,\Theta)$. It then follows that $\partial_U=\left(\frac{\partial U}{\partial S}\right)^{-1}\partial_S=\Omega e_3.$

Because $u=U$ and $v=V$ in both coordinate systems, reparametrization and the identity $\partial_U |f|^2=\langle f, F \rangle +\frac{2\lambda}{-U}|f|^2$ give
\begin{equation*}
    \begin{aligned}
        \lnm f\rnm_{L^\infty(S_{u,v})}\leq & (-U)^{-\lambda} \sup_{\Theta}\left(|f(-1,V,\Theta)|+\int_{-1}^U(-U^\prime)^\lambda|F(U^\prime,V,\Theta)|dU^\prime\right)\\
        \leq & \frac{1}{(-u)^\lambda}\lnm f\rnm_{L^\infty(S_{-1,v})}+\frac{1}{(-u)^\lambda}\int_{-1}^u (-u^\prime)^\lambda\lnm F\rnm_{L^\infty(S_{u^\prime,v})}du^\prime.
    \end{aligned}
\end{equation*}
\end{proof}

Throughout this section, we fix a constant $\delta$ satisfying $\epsilon\ll \frac{1}{N_0}\delta\ll \frac{\k}{N_0^2}$. For notational convenience, we adopt the convention $\delta+\delta\sim \delta$, so that $\frac{1}{(-u)^{\delta}}\times \frac{1}{(-u)^{\delta}}\sim \frac{1}{(-u)^{\delta}}$; equivalently, one may replace $\delta$ by $o(\k/N_0)$. This convention is harmless since replacing $\delta$ by $\delta/2$ leaves all arguments unchanged.

We start with the quantities defined by the $\nabla_3$ equations.
\begin{proposition}\label{prop:sec-3-first-estimate-chih-D4phi}
    For $j,l\leq N_0-1$, the following estimate holds:
    \begin{equation*}
        \lnm (-u)^{j+l+1}\left((\Omega \nabla_3)^{(0)}\right)^l\left(\nabla^{(0)}\right)^j\left((\Omega\chih)^{(1/2)},(\partial_v\phi)^{(1)} \right)\rnm_{L^\infty(S_{u,v})}\lesssim \frac{1}{(-u)^{\delta}}.
    \end{equation*}
\end{proposition}
\begin{proof}
    In the construction equation \eqref{Construction_Eq_chih} for $\Omega^{-1}\chih^{(i+1/2)}$, we denote the right-hand side by $\R_i(\chih)$. We derive
    \begin{equation*}
        \begin{aligned}
            &(\Omega\nabla_3)^{(0)}(\nabla^{(0)})^j(\Omega\chih)^{(1/2)} + \frac{1+j}{2}(\Omega\tr\chib)^{(0)}(\nabla^{(0)})^j(\Omega\chih)^{(1/2)} + (\Omega\chibh)^{(0)}(\nabla^{(0)})^j(\Omega\chih)^{(1/2)}\\
            \sim & \sum_{i_1+i_2=j-1}(\nabla^{(0)})^{i_1+1}(\Omega\chib)^{(0)}(\nabla^{(0)})^{i_2}(\Omega\chih)^{(1/2)}+(\nabla^{(0)})^j\R_0(\chih).
        \end{aligned}
    \end{equation*}
    We have
    \[\sum_{i,j\leq N_0}\lnm (-u)^{i+j}\left((\Omega\nabla_3)^{(0)}\right)^j(\nabla^{(0)})^i\R_0(\chih)\rnm_{L^\infty(S_{u,v})}\lesssim (-u)^{\k-2}\leq (-u)^{-2-\delta}.\]
    Together with $\left|(-u)(\Omega\tr\chib)^{(0)}+2\right|+\left|(-u)(\Omega\chibh)^{(0)}\right|=o(\delta)$, this gives
    \begin{equation*}
        \begin{aligned}
            &(\Omega e_3)^{(0)}\left|(\nabla^{(0)})^j(\Omega\chih)^{(1/2)}\right|-\frac{1+j+o(\delta)}{-u}\left|(\nabla^{(0)})^j(\Omega\chih)^{(1/2)}\right|\\
            \lesssim & \sum_{i\leq j-1} \frac{1}{(-u)^{j-i+1}}\left|(\nabla^{(0)})^i(\Omega\chih)^{(1/2)}\right|+\frac{1}{(-u)^{j+2+\delta}}.
        \end{aligned}
    \end{equation*}
    Using Lemma \ref{lemma:sec3-transport-estimate}, we conclude inductively that 
    \begin{equation*}
        \lnm (\nabla^{(0)})^j(\Omega\chih)^{(1/2)}\rnm_{L^\infty(S_{u,v})}\lesssim \frac{1}{(-u)^{j+2+\delta}},\,\text{for}\, j\leq N_0.
    \end{equation*}
    Moreover, for $t_l=\left((-u)(\Omega\nabla_3)^{(0)}\right)^{l}(\Omega\chih)^{(1/2)}$, we have 
    \begin{equation*}
        \begin{aligned}
            &(\Omega\nabla_3)^{(0)}t_l + \frac{1}{2}(\Omega\tr\chib)^{(0)}t_l\\
            \sim & \sum_{i_1+i_2=l-1}\left((-u)(\Omega\nabla_3)^{(0)}\right)^{i_1+1}(\Omega\tr\chib)t_{i_2}+\left((-u)(\Omega\nabla_3)^{(0)}\right)^{l}\R_0(\chih).
        \end{aligned}
    \end{equation*}
    By induction, we can estimate $t_l$ for $l=0,1,\cdots,N_0$, and then $(\nabla^{(0)})^jt_l$ for $j=0,1,\cdots,N_0$.
    
    For $\partial_v\phi^{(1)}$, the same argument applies because \eqref{Construction_Eq_D4phi} has the same structure. This completes the proof.
\end{proof}

Before estimating the remaining quantities, we make a bootstrap assumption. For $\psi$ equal to any of $\zeta$, $\Omega\tr\chi$, $\Omega D_3\phi$, and $\nabla\phi$, and $\gamma$ equal to any of $g, b,\log\Omega$, we assume
\begin{equation*}
    \begin{aligned}
        &\sum_{j,l\leq N_0-2}\lnm (-u)^{j+l+1}\left((\Omega\nabla_3)^{(0)}\right)^{j}(\nabla^{(0)})^l\left(\psi^{(1)}-\psi^{(0)}\right)\rnm_{L^\infty(S_{u,v})}\leq B\frac{v}{(-u)^{1+\delta}},\\
        &\sum_{j,l\leq N_0-2}\lnm (-u)^{j+l}\left((\Omega\nabla_3)^{(0)}\right)^{j}(\nabla^{(0)})^l\left(\gamma^{(1)}-
        \gamma^{(0)}\right)\rnm_{L^\infty(S_{u,v})}\leq B\frac{v}{(-u)^{1+\delta}}.
    \end{aligned}
\end{equation*}
Assume further that
\begin{equation*}
    \begin{aligned}
        &\sum_{j,l\leq N_0-2}\lnm (-u)^{j+l}\left((\Omega\nabla_3)^{(0)}\right)^j(\nabla^{(0)})^l\left((\Omega\chih)^{(1)}-(\Omega\chih)^{(1/2)} \right)\rnm_{L^\infty(S_{u,v})} \leq B\frac{v}{(-u)^{1+\delta}}.
    \end{aligned}
\end{equation*}
Because we work in the region $v/(-u)^{1+\delta}\ll 1$, the assumption implies
\[\lnm (-u)^{j+l+1}\left((\Omega\nabla_3)^{(0)}\right)^{j}(\nabla^{(0)})^l \psi^{(1)},(-u)^{-1}\gamma^{(1)}\rnm_{L^\infty(S_{u,v})}\lesssim 1,\]
and
\[\lnm (-u)^{j+l+1}\left((\Omega\nabla_3)^{(0)}\right)^{j}(\nabla^{(0)})^l (\Omega\chih)^{(1)} \rnm_{L^\infty(S_{u,v})}\lesssim \frac{1}{(-u)^\delta}.\]
We can improve the bootstrap assumption via the $\Lie_v$ equations.
\begin{proposition}\label{Estimate_for_first_step}
    Suppose $v/(-u)^{1+\delta}\ll 1$. For $\psi$ equal to any of $\zeta$, $\Omega\tr\chi$, $\Omega D_3\phi$, and $\nabla\phi$, we have the estimate
    \begin{equation}\label{eq:Estimate_for_first_step_eq1}
    \sum_{j,l\leq N_0-2}\lnm (-u)^{j+l}\left((\Omega\nabla_3)^{(0)}\right)^j(\nabla^{(0)})^l\left(\psi^{(1)}-\psi^{(0)}\right)\rnm_{L^\infty(S_{u,v})}\lesssim \frac{v}{(-u)^{2+\delta}}.
\end{equation}
For $\Omega\chih$, $\Omega\chib$, and $\Omega\omegab$, we similarly have
    \begin{equation}\label{eq:Estimate_for_first_step_eq2}
       \begin{aligned}
        &\sum_{j,l\leq N_0-2} \lnm (-u)^{j+l}\left((\Omega\nabla_3)^{(0)}\right)^j(\nabla^{(0)})^l\left((\Omega\chih)^{(1)}-{(\Omega\chih)}^{(1/2)},(\Omega\omegab)^{(1/2)}-(\Omega\omegab)^{(0)}\right)\rnm_{L^\infty(S_{u,0})}\\
        &\qquad\qquad\lesssim \frac{v}{(-u)^{2+\delta}},\\
        &\sum_{j,l\leq N_0-3} \lnm (-u)^{j+l}\left((\Omega\nabla_3)^{(0)}\right)^j(\nabla^{(0)})^l\left( (\Omega\omegab)^{(1 )}-(\Omega\omegab)^{(1/2)},(\Omega\chib)^{(1 )}-(\Omega\chib)^{(0)}\right)\rnm_{L^\infty(S_{u,0})} \\
        &\qquad\qquad\lesssim \frac{v}{(-u)^{2+\delta}}.
       \end{aligned}
    \end{equation}
    For the metric components $\gamma=g,b,\log\Omega$, we also have
    \begin{equation}\label{eq:Estimate_for_first_step_eq3}
        \sum_{j,l\leq N_0-2}\lnm (-u)^{j+l}\left((\Omega\nabla_3)^{(0)}\right)^j(\nabla^{(0)})^l\left(\gamma^{(1)}-\gamma^{(0)}\right)\rnm_{L^\infty(S_{u,0})}\lesssim\frac{v}{(-u)^{1+\delta}}.
    \end{equation}
\end{proposition}
\begin{proof}
    First, for $\zeta$ {and} $\Omega\tr\chi$, we recall equations \eqref{Construction_Eq_zeta} and \eqref{Construction_Eq_g_chi}. The right-hand sides of the equations, denoted by $\R_i(\zeta),\R_i(\tr\chi)$, satisfy
    \begin{equation*}
        \sum_{j,l\leq N_0-2}\lnm (-u)^{j+l }\left((\Omega\nabla_3)^{(0)}\right)^j(\nabla^{(0)})^l\left(\R_0(\zeta),\R_0(\tr\chi)\right)\rnm_{L^\infty(S_{u,v})}\lesssim \frac{1}{(-u)^{2+\delta}}.
    \end{equation*}
    Integrating in $v$ yields \eqref{eq:Estimate_for_first_step_eq1} for $\zeta,\Omega\tr\chi$. The same argument applies to $(\Omega\omegab)^{(1/2)}-(\Omega\omegab)^{(0)}$ using equation \eqref{Construction_Eq_omegab}.

    We remark that, since $(\Omega\chih)^{(1/2)}$ is in $W^{N_0-1,\infty}$ and $\dv^{(0)}(\Omega\chih)^{(1/2)}$ appears in $\R_0(\zeta)$, we obtain only $\zeta^{(1)}\in W^{N_0-2,\infty}$. Such {a} derivative loss is allowed because we are solving a system of ODEs.

    For $(\nabla\phi)^{(1)}$, we note that 
    \[\Lie_v(\nabla_A\phi)^{(1)}=\nabla_A\partial_v\phi^{(1)},\]
    and hence the estimate follows from Proposition \ref{prop:sec-3-first-estimate-chih-D4phi}.

    Because $\Lie_v\left(b^{(1)}-b^{(0)}\right)=-\frac{1}{4}\left(\Omega^{(1)}\right)^2\zeta^{(1)}$ and the bootstrap assumption provides bounds for $\Omega^{(1)}$ and $\zeta^{(1)}$, direct integration yields \eqref{eq:Estimate_for_first_step_eq3} for $b$. The construction of $g^{(1)}$ from \eqref{Construction_Eq_g_chi} also implies \eqref{eq:Estimate_for_first_step_eq3} for $g$.

    For $\log\Omega^{(1)}$, the $e_3$-propagation equation reads
    \begin{equation*}
        \left(\partial_u+b^{(1/2)}\right)\left(\log\Omega^{(1)}-\log\Omega^{(0)}\right)=-\left(b^{(1/2)}-b^{(0)}\right)\log\Omega^{(0)}-2\left((\Omega\omegab)^{(1/2)}-(\Omega\omegab)^{(0)}\right).
    \end{equation*}
    By the definition of $b^{(1/2)}$ in \eqref{Construction_Eq_logOmega} and the estimate {for} $(\Omega\omegab)^{(1/2)}-(\Omega\omegab)^{(0)}$, we can use Lemma \ref{lemma:sec3-transport-estimate} to obtain \eqref{eq:Estimate_for_first_step_eq3} for $\log\Omega$.

    For $\Omega\omegab$, we observe the identity
    \begin{equation*}
        -2\left(\left(\Omega\omegab\right)^{(1)}-\left(\Omega\omegab\right)^{(1/2)}\right)=\left(b^{(1)}-b^{(1/2)}\right)\nabla\log\Omega^{(1)},
    \end{equation*}
    which yields the desired estimate \eqref{eq:Estimate_for_first_step_eq2} for $\left(\Omega\omegab\right)^{(1)}-\left(\Omega\omegab\right)^{(1/2)}$. Because $\log\Omega^{(1)}$ is in $W^{N_0-2,\infty}$, {the resulting regularity is weaker}.
    
    From the identity
    \[\partial_v(\Omega D_3\phi)^{(1)}=(\Omega e_3)^{(0)}\partial_v\phi^{(1)}+\left(b^{(1)}-b^{(0)}\right)\nabla(\partial_v\phi)^{(1)}-4(\Omega^{(1)})^2\zeta^{(1)}\cdot\nabla\phi^{(1)},\]
    Proposition \ref{prop:sec-3-first-estimate-chih-D4phi} and the preceding estimates give the desired estimate \eqref{eq:Estimate_for_first_step_eq1} for $(\Omega D_3\phi)$.
    
    Similarly, for $\Omega\chib$, a direct computation yields
    \begin{equation*}
        \begin{aligned}
            &\Lie_v\left(\left(\Omega\chib\right)^{(1)}-\left(\Omega\chib\right)^{(0)}\right)\\
            =& \frac{1}{2}\Lie_v\left(\Lie_{b^{(1)}-b^{(0)}}g^{(1)}+(\Omega e_3)^{(0)}\left(g^{(1)}-g^{(0)}\right)\right)\\
            =& -2(\Omega^{(1)})^2\Lie_{\zeta^{(1)}} g^{(1)}+\Lie_{b^{(1)}-b^{(0)}}(\Omega\chi)^{(1)}+(\Omega e_3)^{(0)}\left((\Omega\chi)^{(1)}-(\Omega\chi)^{(0)}\right)\\
            &+\frac{1}{2}\Lie_{\Lie_v b^{(0)}}\left(g^{(1)}-g^{(0)}\right),
        \end{aligned}
    \end{equation*}
    and hence {the desired estimate follows from the transport estimate}. We also remark that $\Lie_{\zeta^{(1)}} g^{(1)}$ contributes to $\nabla\zeta^{(1)}$, and thus $(\Omega\chib)^{(1)}$ has only $W^{N_0-3,\infty}$ regularity.
\end{proof}

With the preceding estimates, we can bound $(\Omega\omega)^{(1)}$.
\begin{corollary}
    The quantity $\Omega\omega$ satisfies the bound
    \begin{equation*}
        \sum_{j,l\leq N_0-3}\lnm (-u)^{j+l}\left((\Omega e_3)^{(0)}\right)^l\left(\nabla^{(0)}\right)^j(\Omega\omega)^{(1)}\rnm_{L^2(S_{u,0})}\lesssim \frac{1}{(-u)^{\delta}}.
    \end{equation*}
\end{corollary}
\begin{proof}
    From the definition of $\Omega\omega$, we compute
    \begin{equation*}
        \begin{aligned}
            (\Omega e_3)^{(1/2)}(\Omega\omega)^{(1)}=& \Lie_v(\Omega\omegab)^{(1/2)}+\frac{1}{2}{\Lie_v b^{(1/2)}}\cdot\nabla\log\Omega^{(1)}.
        \end{aligned}
    \end{equation*}
    This gives $\left|\left(\partial_u+b^{(1/2)}\right)(\Omega\omega)^{(1)}\right|\lesssim \frac{1}{(-u)^{2+\delta}}$, and the claimed estimate follows by integrating along $e_3$.
\end{proof}

\subsection{Iteration estimates}
Having established the base case, we now prove the following general estimates by induction. We first state the theorem.
\begin{theorem}\label{thm:sec-3}
    Fix $N_0\in\mathbb{N}$ with $1\ll N_0$. Let $0\ll \epsilon \ll \frac{1}{N_0}\delta\ll \k$. Assume $v\ll {(-u)^{1+\delta}}$. For the constructed metric $g^{(i)}_{\mu\nu}$, the following estimates hold for the Einstein residues:

\begin{equation}\label{Construction_Estimate_Rc}
    \begin{aligned}
        &\sum_{j,l\leq N_0-4i }\lnm (-u)^{l+j }\left((\Omega \nabla_3)^{(i)}\right)^{l}\left(\nabla^{(i)}\right)^j\left(\wt{\Omega^2\widehat{\Rc}}^{(i)}_{AB},\wt{\Omega^2\Rc}^{(i)}_{34},\wt{\Omega^2R}^{(i)},\wt{\Omega\Rc}^{(i)}_{4A}\right)\rnm_{L^\infty(S_{u,v})}\\
        &\qquad\qquad\qquad\lesssim \frac{v^{i}}{(-u)^{2+i+\delta}},
    \end{aligned}
\end{equation}
\begin{equation}\label{Construction_Estimate_Rc_3}
    \begin{aligned}
        &\sum_{j,l\leq N_0-4i }\lnm (-u)^{l+j }\left((\Omega \nabla_3)^{(i)}\right)^{l}\left(\nabla^{(i)}\right)^j\left(\left(\wt{\Omega^2\Rc}^{(i)}_{34}+\wt{\Omega^2R}^{(i)}\right),\wt{\Omega\Rc}^{(i)}_{3A},\wt{\Omega^2\Rc}^{(i)}_{33}\right)\rnm_{L^\infty(S_{u,v})}\\
        &\qquad+\sum_{j,l\leq N_0-4i-1 }\lnm (-u)^{l+j }\left((\Omega \nabla_3)^{(i)}\right)^{l}\left(\nabla^{(i)}\right)^j\wt{\Omega^2\Rc}^{(i)}_{33}\rnm_{L^\infty(S_{u,v})} \lesssim \frac{v^{i+1}}{(-u)^{i+3+\delta}}.
    \end{aligned}
\end{equation}

The wave operator applied to the scalar field satisfies
    
\begin{equation}\label{Construction_Estimate_square_phi}
        \sum_{j,l\leq N_0-4i }\lnm (-u)^{l+j }\left((\Omega \nabla_3)^{(i)}\right)^{l}\left(\nabla^{(i)}\right)^j\left({\left(\Omega^2\square\phi\right)}^{(i)}\right)\rnm_{L^\infty(S_{u,v})}\lesssim \frac{v^{i }}{(-u)^{i+2+\delta}}.
    \end{equation}
For the outgoing Ricci coefficients $\psib=\Omega\chih,\partial_v\phi,\Omega\omega$, the differences satisfy

\begin{equation}\label{Construction_Estimate_psi_e4}
    \sum_{j,l\leq N_0-4i-2}\lnm (-u)^{l+j}\left((\Omega \nabla_3)^{(i)}\right)^{l}\left(\nabla^{(i)}\right)^j\left(\psib^{(i+1)}-\psib^{(i)}\right)\rnm_{L^\infty(S_{u,v})}\lesssim \frac{v^{i}}{(-u)^{1+i+\delta}},
\end{equation}

and for $\Omega\chib$ and $\psi=\Omega\tr\chi,\Omega\chib,\Omega \nabla_3\phi,\nabla\phi,\eta,\etab,\zeta$,

\begin{equation}\label{Construction_Estimate_psi_e3}
    \begin{aligned}
        &\sum_{j,l\leq N_0-4i-3}\lnm (-u)^{l+j}\left((\Omega \nabla_3)^{(i)}\right)^{l}\left(\nabla^{(i)}\right)^j\left((\Omega\chib)^{(i+1)}-(\Omega\chib)^{(i)}\right)\rnm_{L^\infty(S_{u,v})}\\
        &+\sum_{j,l\leq N_0-4i-2}\lnm (-u)^{l+j}\left((\Omega \nabla_3)^{(i)}\right)^{l}\left(\nabla^{(i)}\right)^j\left(\psi^{(i+1)}-\psi^{(i)}\right)\rnm_{L^\infty(S_{u,v})}\lesssim \frac{v^{i+1}}{(-u)^{i+2+\delta}}.
    \end{aligned}
\end{equation}

The metric components $\gamma=g,b,\log\Omega$ satisfy 
\begin{equation}\label{Construction_Estimate_metric}
    \sum_{j,l\leq N_0-4i-2}\lnm (-u)^{l+j}\left((\Omega \nabla_3)^{(i)}\right)^{l}\left(\nabla^{(i)}\right)^j\left(\gamma^{(i+1)}-\gamma^{(i)}\right)\rnm_{L^\infty(S_{u,v})}\lesssim \frac{v^{i+1}}{(-u)^{i+1+\delta}}.
\end{equation}

Furthermore, we have the following auxiliary estimates:
\begin{equation}\label{Construction_Estimate_chih}
    \begin{aligned}
        &\sum_{j,l\leq N_0-4i-1 }\lnm (-u)^{l+j}\left((\Omega \nabla_3)^{(i)}\right)^{l}\left(\nabla^{(i)}\right)^j\left((\Omega\chih)^{(i+1)}-(\Omega\chih)^{(i+1/2)}\right)\rnm_{L^\infty(S_{u,v})} \lesssim \frac{v^{i+1}}{(-u)^{2+i+\delta}},
    \end{aligned}
\end{equation} 
\begin{equation}\label{Construction_Estimate_omegab}
    \begin{aligned}
        &\sum_{j,l\leq N_0-4i }\lnm (-u)^{l+j}\left((\Omega \nabla_3)^{(i)}\right)^{l}\left(\nabla^{(i)}\right)^j\left((\Omega\omegab)^{(i+1/2)}-(\Omega\omegab)^{(i)}\right)\rnm_{L^\infty(S_{u,v})} \lesssim \frac{v^{i+1}}{(-u)^{2+i+\delta}}.
    \end{aligned}
\end{equation} 

\end{theorem}
The derivative count $N_0-4i$ is not optimal, but these estimates suffice to solve the system of ODEs. In general, the initial data on $S_{-1,v}$ and $S_{u,0}$ can be taken sufficiently smooth as a perturbation of a spherically symmetric solution.

\vspace{2mm}

The proof proceeds by induction: assuming \ref{Construction_Estimate_psi_e4} and \ref{Construction_Estimate_psi_e3} for $\psi^{(i)}-\psi^{(i-1)}$, together with \ref{Construction_Estimate_Rc} and \ref{Construction_Estimate_Rc_3} for $\wt{\Rc}^{(i)}$, we establish \ref{Construction_Estimate_psi_e4} and \ref{Construction_Estimate_psi_e3} for $\psi^{(i+1)}-\psi^{(i)}$ and $\psib^{(i+1)}-\psib^{(i)}$. We then derive \ref{Construction_Estimate_Rc}, \ref{Construction_Estimate_Rc_3}, and \ref{Construction_Estimate_square_phi} for the {Einstein--scalar field} residues from these improved connection estimates. The base case is provided by Proposition \ref{Estimate_for_first_step}. Because we derive $\wt{\Rc}^{(i+1)}$ and $\wt{\square\phi}^{(i+1)}$ from the estimates for $\psi^{(i+1)}$ and $\psib^{(i+1)}$, we do not need to compute $\wt{\Rc}^{(1)}$ and $\wt{\square\phi}^{(1)}$ explicitly.

In the rest of this section, we give the proof in detail.

\subsubsection{Estimates for $\gamma, \psi,\psib$}
Assuming the estimates hold at step $i$, we proceed to estimate $(\cdot)^{(i+1)}-(\cdot)^{(i)}$.
Recall the equation

\begin{equation*}
    \left((\Omega\nabla_3)^{(i)}+\frac{1}{2}(\Omega\tr\chib)^{(i)}\right)\left(\Omega\chih^{(i+1/2)}-\Omega\chih^{(i)}\right)=- \wt{\Omega^2\widehat{\Rc}}^{(i)}_{AB}{{.}}
\end{equation*} 

{A}rguments similar to those in Proposition \ref{prop:sec-3-first-estimate-chih-D4phi} yield
\begin{equation}\label{eq:sec-3-3-chih-est-1}
    \begin{aligned}
        &\sum_{j,l\leq N_0-4i }\lnm (-u)^{l+j}\left((\Omega \nabla_3)^{(i)}\right)^{l}\left(\nabla^{(i)}\right)^j\left((\Omega\chih)^{(i+1/2)}-(\Omega\chih)^{(i)}\right)\rnm_{L^\infty(S_{u,v})}\\
        \lesssim & \frac{1}{(-u)^{1+o(\delta)}}\left(v^{i}+\int_{-1}^u(-u^\prime)^{1+o(\delta)}\frac{v^{i}}{(-u^\prime)^{2+i+\delta}}\right)\lesssim\frac{v^i}{(-u)^{1+i+\delta}}.
    \end{aligned}
\end{equation}
Here we have used the initial data estimates in Proposition \ref{prop:sec-3-H-initial-g-chi-estimate} to obtain
\[|(\Omega\chih)^{(i+1/2)}-(\Omega\chih)^{(i)}|\leq |(\Omega\chih)^{(i+1/2)}-\Omega\chih|+|\Omega\chih-(\Omega\chih)^{(i)}|\lesssim v^i.\]
For $\partial_v\phi$, by construction we have

\begin{equation*}
    \begin{aligned}
        &\left((\Omega\nabla_3)^{(i)}+\frac{1}{2}(\Omega\tr\chib)^{(i)}\right)\left(\partial_v\phi^{(i+1)}-\partial_v\phi^{(i)}\right)
        =-(\Omega^2\square\phi)^{(i)}{{.}}
    \end{aligned}
\end{equation*}
    {The same argument as for $\Omega\chih$ yields the inequality}
\begin{equation}\label{eq:sec-3-3-D4phi-est}
    \begin{aligned}
        &\sum_{j,l\leq N_0-4i }\lnm (-u)^{l+j}\left((\Omega \nabla_3)^{(i)}\right)^{l}\left(\nabla^{(i)}\right)^j\left(\partial_v\phi^{(i+1)}-\partial_v\phi^{(i)}\right)\rnm_{L^\infty(S_{u,v})} \lesssim\frac{v^i}{(-u)^{1+i+\delta}}.
    \end{aligned}
\end{equation}
Because of the identity
\[\Lie_v\left(\nabla_A\phi^{(i+1)}-\nabla_A\phi^{(i)}\right)=\nabla_A\left(\partial_v\phi^{(i+1)}-\partial_v\phi^{(i)}\right),\]
we further obtain 
\begin{equation}\label{eq:sec-3-3-nablaphi-est}
    \begin{aligned}
        &\sum_{j,l\leq N_0-4i-1 }\lnm (-u)^{l+j}\left((\Omega \nabla_3)^{(i)}\right)^{l}\left(\nabla^{(i)}\right)^j\left(\nabla\phi^{(i+1)}-\nabla\phi^{(i)}\right)\rnm_{L^\infty(S_{u,v})}\lesssim\frac{v^{i+1}}{(-u)^{2+i+\delta}}.
    \end{aligned}
\end{equation}

For $\Omega\omegab$, we write the equation 
\begin{equation*}
    \begin{aligned}
        &\partial_v\left((\Omega\omegab)^{(i+1/2)}-(\Omega\omegab)^{(i)}\right)\\
        =&\frac{1}{4}\left((\Omega\chih)^{(i+1/2)}-(\Omega\chih)^{(i)}\right)(\Omega\chibh)^{(i)}-\frac{1}{4}\wt{\Omega^2\Rc}^{(i)}_{AB}(g^{(i)})^{AB}\\
        &+\frac{1}{4}(\Omega e_3\phi)^{(i)}\left(\partial_v\phi^{(i+1)}-\partial_v\phi^{(i)}\right)-\frac{1}{4}\wt{\Omega^2\Rc_{34}}^{(i)}.
    \end{aligned}
\end{equation*}
The preceding estimates for $(\Omega\chih)^{(i+1/2)}-(\Omega\chih)^{(i)}$ and $\partial_v\phi^{(i+1)}-\partial_v\phi^{(i)}$, together with the step-$i$ estimate \eqref{Construction_Estimate_Rc}, give

\begin{equation}\label{eq:sec-3-3-omegab-est-1}
    \begin{aligned}
        &\sum_{j,l\leq N_0-4i }\lnm (-u)^{l+j}\left((\Omega \nabla_3)^{(i)}\right)^{l}\left(\nabla^{(i)}\right)^j\left((\Omega\omegab)^{(i+1/2)}-(\Omega\omegab)^{(i)}\right)\rnm_{L^\infty(S_{u,v})}\\
        &\qquad\qquad\qquad\lesssim \frac{v^{i+1}}{(-u)^{2+i+\delta}}.
    \end{aligned}
\end{equation}  

For $i\geq 1$, we have $b^{(i+1/2)}=b^{(i)}$ by \eqref{Construction_Eq_logOmega}. This intermediate quantity is defined because we assigned $b^{(0)}=b_0$. We therefore do not distinguish between $b^{(i+1/2)}$ and $b^{(i)}$ in the subsequent proof.
The construction of $\log\Omega^{(i+1)}$ gives
\[(\Omega e_3)^{(i)}\left(\log\Omega^{(i+1)}-\log\Omega^{(i)}\right)=-2\left((\Omega\omegab)^{(i+1/2)}-(\Omega\omegab)^{(i)}\right).\]
The preceding result \eqref{eq:sec-3-3-omegab-est-1} yields
 
\begin{equation}\label{eq:sec-3-3-logOmega-est}
    \begin{aligned}
        &\sum_{j,l\leq N_0-4i }\lnm (-u)^{l+j}\left((\Omega \nabla_3)^{(i)}\right)^{l}\left(\nabla^{(i)}\right)^j\left(\log\Omega^{(i+1)}-\log\Omega^{(i)}\right)\rnm_{L^\infty(S_{u,v})}\lesssim \frac{v^{i+1}}{(-u)^{2+i+\delta}}.
    \end{aligned}
\end{equation}

For $\Omega\omega$, we have the identity 
 \begin{equation*}
    \begin{aligned}
        &(\Omega e_3)^{(i)}\left((\Omega\omega)^{(i+1)}-(\Omega\omega)^{(i)}\right)\\
        =&\partial_v\left((\Omega\omegab)^{(i+1/2)}-(\Omega\omegab)^{(i)}\right)+\frac{1}{2}\Lie_v b^{(i)}\cdot\nabla(\log\Omega^{(i+1)}-\log\Omega^{(i)})\\
        =&\frac{1}{4}\left((\Omega\chih)^{(i+1/2)}-(\Omega\chih)^{(i)}\right)(\Omega\chibh)^{(i)}-\frac{1}{4}\wt{\Omega^2\Rc}^{(i)}_{AB}(g^{(i)})^{AB}-\wt{\Omega^2\Rc_{34}}^{(i)}\\
        &+\frac{1}{4}(\Omega e_3\phi)^{(i)}\left(\partial_v\phi^{(i+1)}-\partial_v\phi^{(i)}\right)-\frac{1}{8}(\Omega^{(i)})^2\zeta^{(i)}\cdot\nabla\left(\log\Omega^{(i+1)}-\log\Omega^{(i)}\right).
    \end{aligned}
\end{equation*}
We note that the power of $v$ on the right-hand side is $v^{i}$, even though the last term is of $O(v^{i+1})$. We also note that there is a derivative loss in $\nabla\log\Omega^{(i+1)}$, and therefore arrive at

\begin{equation}\label{eq:sec-3-3-omega-est}
    \begin{aligned}
        &\sum_{j,l\leq N_0-4i-1 }\lnm (-u)^{l+j}\left((\Omega \nabla_3)^{(i)}\right)^{l}\left(\nabla^{(i)}\right)^j\left((\Omega\omega)^{(i+1)}-(\Omega\omega)^{(i)}\right)\rnm_{L^\infty(S_{u,v})} \lesssim \frac{v^{i}}{(-u)^{1+i+\delta}}.
    \end{aligned}
\end{equation}

We next estimate $\zeta$. From \eqref{Construction_Eq_zeta}, we have 

\begin{equation*}
    \begin{aligned}
        &\Lie_v\left(\zeta^{(i+1)}-\zeta^{(i)}\right)=-3\left((\Omega\chih)^{(i+1/2)}-(\Omega\chih)^{(i)}\right)\zeta^{(i)}+\frac{1}{2}\nabla\left((\Omega\omega)^{(i+1)}-(\Omega\omega)^{(i)}\right)\\
        &\qquad-\frac{1}{2}\left((\Omega\omega)^{(i+1)}\nabla\log\Omega^{(i+1)}-(\Omega\omega)^{(i)}\nabla\log\Omega^{(i)}\right)+(\dv)^{(i)}\left((\Omega\chih)^{(i+1/2)}-(\Omega\chih)^{(i)}\right)\\
        &\qquad-(\Omega\chih)^{(i+1/2)}\nabla\log\Omega^{(i+1)}+(\Omega\chih)^{(i)}\nabla\log\Omega^{(i)}+\frac{1}{2}(\Omega\tr\chi)^{(i)}\nabla(\log\Omega^{(i+1)}-\log\Omega^{(i)})\\
        &\qquad-(\partial_v\phi^{(i+1)}-\partial_v\phi^{(i)})\nabla\phi^{(i)}+\wt{\Omega\Rc_{4\cdot}}^{(i)}.
    \end{aligned}
\end{equation*} 
The induction hypothesis and the preceding estimates, {together with} integration in $v$, imply

\begin{equation}\label{eq:sec-3-3-zeta-est}
    \begin{aligned}
        &\sum_{j,l\leq N_0-4i-2 }\lnm (-u)^{l+j}\left((\Omega \nabla_3)^{(i)}\right)^{l}\left(\nabla^{(i)}\right)^j\left(\zeta^{(i+1)}-\zeta^{(i)}\right)\rnm_{L^\infty(S_{u,v})} \lesssim \frac{v^{i+1}}{(-u)^{2+i+\delta}}.
    \end{aligned}
\end{equation}  
Since 
$$\Lie_v\left(b^{(i+1)}-b^{(i)}\right)=-4\left((\Omega^{(i+1)})^2\zeta^{(i+1)}-(\Omega^{(i)})^2\zeta^{(i)}\right),$$
we conclude that 

\begin{equation}\label{eq:sec-3-3-b-est}
    \begin{aligned}
        &\sum_{j,l\leq N_0-4i-2 }\lnm (-u)^{l+j}\left((\Omega \nabla_3)^{(i)}\right)^{l}\left(\nabla^{(i)}\right)^j\left(b^{(i+1)}-b^{(i)}\right)\rnm_{L^\infty(S_{u,v})} \lesssim \frac{v^{i+2}}{(-u)^{2+i+\delta}},
    \end{aligned}
\end{equation}  
which is stronger than \eqref{Construction_Estimate_metric}.

Together with the preceding inequality, the identity 
$$(\Omega\omegab)^{(i+1)}-(\Omega\omegab)^{(i+1/2)}=-\frac{1}{2}\left(b^{(i+1)}-b^{(i)}\right)\cdot\nabla\log\Omega^{(i)}$$ 
gives the final estimate for $\Omega\omegab$:

\begin{equation}\label{eq:sec-3-3-omegab-est-2}
    \begin{aligned}
        &\sum_{j,l\leq N_0-4i-2 }\lnm (-u)^{l+j}\left((\Omega \nabla_3)^{(i)}\right)^{l}\left(\nabla^{(i)}\right)^j\left((\Omega\omegab)^{(i+1)}-(\Omega\omegab)^{(i+1/2)}\right)\rnm_{L^\infty(S_{u,v})}\lesssim \frac{v^{i+2}}{(-u)^{3+i+\delta}}.
    \end{aligned}
\end{equation}  
{This completes the proof of \eqref{Construction_Estimate_psi_e4} for $\Omega\omegab$.}

We proceed to estimate $g,\chi$ using their defining equations \eqref{Construction_Eq_g_chi}. Schematically, we write
\begin{equation*}
    \begin{aligned}
        &(\Omega\chih)^{(i+1)}-(\Omega\chih)^{(i+1/2)}\sim (\Omega\chih)^{(i+1)}\left(g^{(i+1)}-g^{(i)}\right),\\
        &\Lie_v \left(g^{(i+1)}-g^{(i)}\right)\sim (\Omega\chih)^{(i+1)}-(\Omega\chih)^{(i)}\\
        &\qquad+(\Omega\tr\chi)^{(i+1)}g^{(i+1)}-(\Omega\tr\chi)^{(i)}g^{(i)},\\
        &\partial_v\left((\Omega\tr\chi)^{(i+1)}-(\Omega\tr\chi)^{(i)}\right)\sim \left|(\Omega\chih)^{(i+1)}\right|^2-\left|(\Omega\chih)^{(i)}\right|^2+(\partial_v\phi^{(i+1)})^2-(\partial_v\phi^{(i)})^2\\
        &\qquad+(\Omega\tr\chi)^{(i+1)}(\Omega\omega)^{(i+1)}-(\Omega\tr\chi)^{(i)}(\Omega\omega)^{(i)} + \left((\Omega\tr\chi)^{(i+1)}\right)^2-\left((\Omega\tr\chi)^{(i)}\right)^2.
    \end{aligned}
\end{equation*}
Using \eqref{eq:sec-3-3-chih-est-1}, \eqref{eq:sec-3-3-D4phi-est}, and \eqref{eq:sec-3-3-omega-est}, the standard transport estimate implies

\begin{equation}\label{eq:sec-3-3-trchi-est}
    \begin{aligned}
        &\sum_{j,l\leq N_0-4i-1 }\lnm (-u)^{l+j}\left((\Omega \nabla_3)^{(i)}\right)^{l}\left(\nabla^{(i)}\right)^j\left((\Omega\tr\chi)^{(i+1)}-(\Omega\tr\chi)^{(i)}\right)\rnm_{L^\infty(S_{u,v})} \lesssim \frac{v^{i+1}}{(-u)^{2+i+\delta}},
    \end{aligned}
\end{equation} 
\begin{equation}\label{eq:sec-3-3-g-est}
    \begin{aligned}
        &\sum_{j,l\leq N_0-4i-1 }\lnm (-u)^{l+j}\left((\Omega \nabla_3)^{(i)}\right)^{l}\left(\nabla^{(i)}\right)^j\left(g^{(i+1)}-g^{(i)}\right)\rnm_{L^\infty(S_{u,v})} \lesssim \frac{v^{i+1}}{(-u)^{1+i+\delta}},
    \end{aligned}
\end{equation} 
\begin{equation}\label{eq:sec-3-3-chih-est-2}
    \begin{aligned}
        &\sum_{j,l\leq N_0-4i-1 }\lnm (-u)^{l+j}\left((\Omega \nabla_3)^{(i)}\right)^{l}\left(\nabla^{(i)}\right)^j\left((\Omega\chih)^{(i+1)}-(\Omega\chih)^{(i+1/2)}\right)\rnm_{L^\infty(S_{u,v})} \lesssim \frac{v^{i+1}}{(-u)^{2+i+\delta}}.
    \end{aligned}
\end{equation} 

Together with \eqref{eq:sec-3-3-chih-est-1}, this proves \eqref{Construction_Estimate_psi_e3} for $\Omega\chih$.

\vspace{2mm}
Finally, we estimate $\Omega D_3\phi$ and $\Omega\chib$.
Because 

\begin{equation*}
    \begin{aligned}
        &\partial_v\left((\Omega D_3\phi)^{(i+1)}-(\Omega D_3\phi)^{(i)}\right)=(\Omega e_3)^{(i)}\left(\partial_v\phi^{(i+1)}-\partial_v\phi^{(i)}\right)+\left(b^{(i+1)}-b^{(i)}\right)\cdot\nabla\partial_v\phi^{(i+1)}\\
        &\qquad\qquad-4\left(\Omega^{(i+1)}\right)^2\zeta^{(i+1)}\cdot\nabla\phi^{(i+1)}+4\left(\Omega^{(i)}\right)^2\zeta^{(i)}\cdot\nabla\phi^{(i)},
    \end{aligned}
\end{equation*}

the preceding estimates \eqref{eq:sec-3-3-D4phi-est}, \eqref{eq:sec-3-3-nablaphi-est}, \eqref{eq:sec-3-3-zeta-est}, \eqref{eq:sec-3-3-logOmega-est}, and \eqref{eq:sec-3-3-b-est} together give

\begin{equation}\label{eq:sec-3-3-D3phi-est}
    \begin{aligned}
        &\sum_{j,l\leq N_0-4i-2 }\lnm (-u)^{l+j}\left((\Omega \nabla_3)^{(i)}\right)^{l}\left(\nabla^{(i)}\right)^j\left((\Omega D_3\phi)^{(i+1)}-(\Omega D_3\phi)^{(i)}\right)\rnm_{L^\infty(S_{u,v})} \lesssim \frac{v^{i+1}}{(-u)^{2+i+\delta}}.
    \end{aligned}
\end{equation} 

For $\Omega\chib$, we similarly have

\begin{equation*}
    \begin{aligned}
        &\Lie_v\left((\Omega\chib)^{(i+1)}-(\Omega\chib)^{(i)}\right) = \Lie_{(\Omega e_3)^{(i)}}\left((\Omega\chi)^{(i+1)}-(\Omega\chi)^{(i)}\right)\\
        &\qquad +\Lie_{(\Omega e_3)^{(i+1)}-(\Omega e_3)^{(i)}}(\Omega\chi)^{(i+1)}+\frac{1}{2}\Lie_{\left(\Lie_v b^{(i+1)}-\Lie_v b^{(i)}\right)}g^{(i)}+\frac{1}{2}\Lie_{\Lie_v b^{(i+1)}}\left(g^{(i+1)}-g^{(i)}\right)\\
        \sim &(\Omega\nabla_3)^{(i)}\left((\Omega\chi)^{(i+1)}-(\Omega\chi)^{(i)}\right)+(\Omega\chib)^{(i)}\left((\Omega\chi)^{(i+1)}-(\Omega\chi)^{(i)}\right)+\left({b^{(i+1)}-b^{(i)}}\right)\nabla(\Omega\chi)^{(i+1)}\\
        &\qquad +\Omega\chi^{(i+1)}\left((\Omega\chib)^{(i+1)}-(\Omega\chib)^{(i)}\right)+\nabla^{(i)}\otimes \left((\Omega^{(i+1)})^2\zeta^{(i+1)}-(\Omega^{(i)})^2\zeta^{(i)}\right)\\
        &\qquad +(\Omega^{(i+1)})^2\left(\nabla^{(i+1)}\otimes\zeta^{(i+1)}-\nabla^{(i)}\otimes\zeta^{(i)}\right).
    \end{aligned}
\end{equation*}

Here we use the notation $(\nabla\otimes X)_{AB}=\nabla_AX_B+\nabla_BX_A$ and {write} $\sim$ to suppress constant coefficients.
Employing \eqref{eq:sec-3-3-b-est}, \eqref{eq:sec-3-3-chih-est-1}, \eqref{eq:sec-3-3-chih-est-2}, \eqref{eq:sec-3-3-zeta-est}, and \eqref{eq:sec-3-3-logOmega-est}, we obtain

\begin{equation}\label{eq:sec-3-3-chib-est}
    \begin{aligned}
        &\sum_{j,l\leq N_0-4i-3 }\lnm (-u)^{l+j}\left((\Omega \nabla_3)^{(i)}\right)^{l}\left(\nabla^{(i)}\right)^j\left((\Omega \chib)^{(i+1)}-(\Omega \chib)^{(i)}\right)\rnm_{L^\infty(S_{u,v})} \lesssim \frac{v^{i+1}}{(-u)^{2+i+\delta}}.
    \end{aligned}
\end{equation} 

We thus finish proving \eqref{Construction_Estimate_psi_e3}, \eqref{Construction_Estimate_psi_e4}, \eqref{Construction_Estimate_metric}, \eqref{Construction_Estimate_chih}, and \eqref{Construction_Estimate_omegab}.

\subsubsection{Estimates for $\wt{\Rc}$ and $\square\phi$} In this {subsection}, we complete the proof of Theorem \ref{thm:sec-3}.
Using the $\nabla_3\chih$ equation and the construction of $(\Omega\chih)^{(i+1/2)}$, we have

\begin{equation}\label{eq:sec-3-3-hatRc}
    \begin{aligned}
        &\wt{\Omega^2\widehat{\Rc}}_{AB}^{(i+1)}
        = (\Omega\nabla_3)^{(i+1)}(\Omega\chih)^{(i+1)}+\frac{1}{2}\left(\Omega\tr\chib\right)^{(i+1)}\left(\Omega\chih\right)^{(i+1)}-\left(\Omega^2\nabla\phi\hat\otimes\nabla\phi\right)^{(i+1)}\\
        &-(\Omega\nabla_3)^{(i)}(\Omega\chih)^{(i+1/2)}-\frac{1}{2}\left(\Omega\tr\chib\right)^{(i)}\left(\Omega\chih\right)^{(i+1/2)}+\left(\Omega^2\nabla\phi\hat\otimes\nabla\phi\right)^{(i)}\\
        &-\left(\Omega^2(\nabla\hat\otimes\eta+\eta\hat\otimes\eta-\frac{1}{2}\tr\chi\chibh)\right)^{(i+1)}+\left(\Omega^2(\nabla\hat\otimes\eta+\eta\hat\otimes\eta-\frac{1}{2}\tr\chi\chibh)\right)^{(i)}.
    \end{aligned}
\end{equation}
The induction hypotheses \eqref{Construction_Estimate_psi_e3}, \eqref{Construction_Estimate_psi_e4}, \eqref{Construction_Estimate_metric}, and \eqref{Construction_Estimate_chih} imply 

\begin{equation}\label{eq:sec-3-3-hatRc-est}
    \begin{aligned}
        \sum_{j,l\leq N_0-4i-3 }\lnm (-u)^{l+j }\left((\Omega \nabla_3)^{(i)}\right)^{l}\left(\nabla^{(i)}\right)^j\left(\wt{\Omega^2\widehat{\Rc}}^{(i+1)}_{AB}\right)\rnm_{L^\infty(S_{u,v})}\lesssim \frac{v^{i+1}}{(-u)^{3+i+\delta}}.
    \end{aligned}
\end{equation}

Recall the formula
\begin{equation*}
    \begin{aligned}
        \frac{1}{4}\Omega^2\Rc_{34}&=\partial_v\left(\Omega\omegab\right)-\frac{1}{4}\Omega\chibh\cdot\Omega\chih-\frac{1}{2}\Omega^2\eta^2+\Omega^2\eta\etab \\
        &-\frac{1}{4}\left(\Omega e_3(\Omega\tr\chi)+\frac{1}{2}\Omega\tr\chi\Omega\tr\chib-2\Omega^2\dv\eta+2\Omega^2\eta^2\right),
    \end{aligned}
\end{equation*}
for a Lorentzian manifold. We derive
\begin{equation*}
    \begin{aligned}
        &\frac{1}{4}\left((\Omega^2\Rc_{34})^{(i+1)}-(\Omega^2\Rc_{34})^{(i)}\right)\\
        =&\partial_v\left((\Omega\omegab)^{(i+1)}-(\Omega\omegab)^{(i)}\right)-\frac{1}{4}\left((\Omega\chibh)^{(i+1)}(\Omega\chih)^{(i+1)}-(\Omega\chibh)^{(i)}(\Omega\chih)^{(i)}\right)\\
        &\qquad+\left(-\frac{1}{2}\Omega^2\eta^2+\Omega^2\eta\etab -\frac{1}{4}\Omega e_3(\Omega\tr\chi)-\frac{1}{8}\Omega\tr\chi\Omega\tr\chib+\frac{1}{2}\Omega^2\dv\eta-\frac{1}{2}\Omega^2\eta^2\right)^{(i+1)}\\
        &\qquad-\left(-\frac{1}{2}\Omega^2\eta^2+\Omega^2\eta\etab -\frac{1}{4}\Omega e_3(\Omega\tr\chi)-\frac{1}{8}\Omega\tr\chi\Omega\tr\chib+\frac{1}{2}\Omega^2\dv\eta-\frac{1}{2}\Omega^2\eta^2\right)^{(i)}.
    \end{aligned}
\end{equation*}
Moreover, the construction of $(\Omega\omegab)^{(i+1/2)}$ implies
\begin{equation*}
    \begin{aligned}
        &\frac{1}{4}{(\Omega^2\Rc_{34})}^{(i)}-\frac{1}{4}(\Omega e_3\phi)^{(i+1)}\partial_v\phi^{(i+1)}=\frac{1}{4}\left((\Omega\chih)^{(i+1/2)}-(\Omega\chih)^{(i)}\right)(\Omega\chibh)^{(i)}-\frac{1}{4}\wt{\Omega^2\Rc}^{(i)}_{AB}(g^{(i)})^{AB}\\
        &\qquad+\frac{1}{4}\left((\Omega e_3\phi)^{(i)}-(\Omega e_3\phi)^{(i+1)}\right)\partial_v\phi^{(i+1)}-\partial_v\left((\Omega\omegab)^{(i+1/2)}-(\Omega\omegab)^{(i)}\right).
    \end{aligned}
\end{equation*}
Adding {these identities} gives the expression for $\wt{\Omega^2\Rc_{34}}^{(i+1)}${:}
\begin{equation}\label{eq:residue-Rc34}
    \begin{aligned}
        &\frac{1}{4}\wt{\Omega^2\Rc_{34}}^{(i+1)}=\partial_v\left((\Omega\omegab)^{(i+1)}-(\Omega\omegab)^{(i+1/2)}\right)-\frac{1}{4}\left((\Omega\chibh)^{(i+1)}(\Omega\chih)^{(i+1)}-(\Omega\chibh)^{(i)}(\Omega\chih)^{(i+1/2)}\right)\\
        &\qquad-\frac{1}{4}\left(\wt{\Omega^2\Rc_{34}}^{(i)}+\wt{\Omega^2 R}^{(i)}\right)+\frac{1}{4}\left((\Omega e_3\phi)^{(i)}-(\Omega e_3\phi)^{(i+1)}\right)\partial_v\phi^{(i+1)}\\
        &\qquad+\left(-\frac{1}{2}\Omega^2\eta^2+\Omega^2\eta\etab -\frac{1}{4}\Omega e_3(\Omega\tr\chi)-\frac{1}{8}\Omega\tr\chi\Omega\tr\chib+\frac{1}{2}\Omega^2\dv\eta-\frac{1}{2}\Omega^2\eta^2\right)^{(i+1)}\\
        &\qquad-\left(-\frac{1}{2}\Omega^2\eta^2+\Omega^2\eta\etab -\frac{1}{4}\Omega e_3(\Omega\tr\chi)-\frac{1}{8}\Omega\tr\chi\Omega\tr\chib+\frac{1}{2}\Omega^2\dv\eta-\frac{1}{2}\Omega^2\eta^2\right)^{(i)}.
    \end{aligned}
\end{equation}
From the identity
  
$$(\Omega\omegab)^{(i+1)}-(\Omega\omegab)^{(i+1/2)}=-\frac{1}{2}\left(b^{(i+1)}-b^{(i)}\right)\cdot\nabla\log\Omega^{(i)},$$
we have 
\begin{equation*}
    \begin{aligned}
        &\partial_v\left((\Omega\omegab)^{(i+1)}-(\Omega\omegab)^{(i+1/2)}\right)\\
        &\qquad=\frac{1}{8}\left((\Omega^{(i+1)})^2\zeta^{(i+1)}-(\Omega^{(i)})^2\zeta^{(i)}\right)\cdot\nabla\log\Omega^{(i)}+\left(b^{(i+1)}-b^{(i)}\right)\cdot\nabla(\Omega\omega)^{(i)}.
    \end{aligned}
\end{equation*}
Hence, we obtain 
\begin{equation*}
    \sum_{j+l\leq N_0-4i-3} \lnm (-u)^{j+l}\left((\Omega\nabla_3)^{(i)}\right)^l(\nabla^{(i)})^j\wt{\Omega^2\Rc_{34}}^{(i+1)}\rnm_{L^\infty(S_{u,v})}\lesssim \frac{v^{i+1}}{(-u)^{i+3+\delta}}.
\end{equation*}

Turning to $\Omega\Rc_{4A}$, for a Lorentzian manifold{,} we recall the identity
\begin{equation*}
    \Omega\Rc_{4}^{\ A}=-\Lie_v\zeta^A-2\Omega\tr\chi\zeta^A-2\Omega\chih^{A}_{\ B}\zeta^B 2\nabla^A(\Omega\omega)+\dv(\Omega\chih)^{A}-\frac{1}{2}\nabla^A(\Omega\tr\chi)+\Omega\tr\chi\nabla^A\log\Omega,
\end{equation*}
from which we deduce that 
\begin{equation*}
    \begin{aligned}
        &(\Omega\Rc_{4\cdot})^{(i+1)}-\partial_v\phi^{(i+1)}\nabla\phi^{(i+1)}\\
        =&-\partial_v\phi^{(i+1)}\left(\nabla\phi^{(i+1)}-\nabla\phi^{(i)}\right)-\left(\zeta^{(i+1)}(\Omega\chih)^{(i+1)}-\zeta^{(i)}(\Omega\chih)^{(i+1/2)}\right)\\
        &-\frac{1}{2}(\Omega\omega)^{(i+1)}\nabla\log\Omega^{(i+1)}\left(g^{(i+1)}-g^{(i)}\right)+\frac{1}{2}(g^{(i+1)}-g^{(i)})\nabla(\Omega\omega)^{(i+1)}\\
        &+\left(\dv^{(i+1)}(\Omega\chih)^{(i+1)}-\dv^{(i)}(\Omega\chih)^{(i+1/2)}\right)-\left((\Omega\chih)^{(i+1)}-(\Omega\chih)^{(i+1/2)}\right)\nabla\log\Omega^{(i+1)}\\
        &-\frac{1}{2}\left(\nabla(\Omega\tr\chi)^{(i+1)}-\nabla(\Omega\tr\chi)^{(i)}\right)+\frac{1}{2}((\Omega\tr\chi)^{(i+1)}-(\Omega\tr\chi)^{(i)})\nabla\log\Omega^{(i+1)},
    \end{aligned}
\end{equation*}
and therefore obtain
\begin{equation*}
    \sum_{j+l\leq N_0-4i-3} \lnm (-u)^{j+l}\left((\Omega\nabla_3)^{(i)}\right)^l(\nabla^{(i)})^j\wt{\Omega\Rc_{4\cdot}}^{(i+1)}\rnm_{L^\infty(S_{u,v})}\lesssim \frac{v^{i+1}}{(-u)^{i+3+\delta}}.
\end{equation*}

For $\square\phi$, we observe
\begin{equation*}
        \begin{aligned}
            &(\Omega^2\square\phi)^{(i+1)}\\
            =&\left(\Omega^2\Delta_g\phi-\frac{1}{2}\Omega\tr\chi\Omega e_3\phi+2\Omega^2\eta\cdot\nabla\phi\right)^{(i+1)} -\left(\Omega^2\Delta_g\phi-\frac{1}{2}\Omega\tr\chi\Omega e_3\phi+2\Omega^2\eta\cdot\nabla\phi\right)^{(i)}\\
            &+\frac{1}{2}\left(\left(\Omega\tr\chib\right)^{(i)}-\left(\Omega\tr\chib\right)^{(i+1)}\right)\left(\Omega e_4\phi\right)^{(i+1)}-\left(b^{(i+1)}-b^{(i)}\right)\nabla\partial_v\phi^{(i+1)}.
        \end{aligned}
    \end{equation*}
    It then follows that
    \begin{equation*}
    \sum_{j+l\leq N_0-4i-3} \lnm (-u)^{j+l}\left((\Omega\nabla_3)^{(i)}\right)^l(\nabla^{(i)})^j\left(\Omega^2\square\phi\right)^{(i+1)}\rnm_{L^\infty(S_{u,v})}\lesssim \frac{v^{i+1}}{(-u)^{i+3+\delta}}.
\end{equation*}

Let $T_{\mu\nu}=D_\mu\phi D_\nu\phi$. Then $D^\mu \left(T_{\mu\nu}-\frac{1}{2}(g^{\rho\lambda}T_{\rho\lambda})g_{\mu\nu}\right)=D_\lambda\phi\square\phi$. With $\wt{\Rc}=\Rc-T$ and $\wt{R}=g^{\mu\nu}\wt{\Rc}$, equations \eqref{eq:dv-Ric4}, \eqref{eq:dv-Ric3}, and \eqref{eq:dv-RicA} imply
\begin{equation}\label{Equation_Lorentz_wt_Rc_1}
    \begin{aligned}
        &\frac{1}{2} \nabla_4\left(\wt{\Rc}_{34}+\wt{R}\right)=2 \underline{\omega} \wt{\Rc}_{44}+2 \eta^A \wt{\Rc}_{4 A}-\frac{1}{2} \nabla_3 \wt{\Rc}_{44}+\underline{\eta}^A \wt{\Rc}_{A 4} +\nabla^A \wt{\Rc}_{A 4}-\frac{1}{2} \operatorname{tr} \underline{\chi} \wt{\Rc}_{44}\\
    &\qquad\qquad-\frac{1}{2} \operatorname{tr} \chi\left(\wt{\Rc}_{34}+\wt{R}\right) +\zeta^A \wt{\Rc}_{A 4}-\frac{1}{2} \operatorname{tr} \chi \wt{\Rc}_{34}-\hat{\chi}^{A B} \wt{\widehat{\Rc}}_{A B}+D_4\phi\square\phi,
    \end{aligned}
\end{equation}
\begin{equation}\label{Equation_Lorentz_wt_Rc_2}
    \begin{aligned}
        &\frac{1}{2} \nabla_4 \wt{\Rc}_{33}=2 \omega \wt{\Rc}_{33}+2 \underline{\eta}^A \wt{\Rc}_{3 A}-\frac{1}{2} \nabla_3\left(\wt{\Rc}_{34}+\wt{R}\right)+\eta^A \wt{\Rc}_{A 3} +\nabla^A \wt{\Rc}_{A 3}-\frac{1}{2} \operatorname{tr} \chi \wt{\Rc}_{33}\\
    &\qquad\qquad-\frac{1}{2} \operatorname{tr} \underline{\chi}\left(\wt{\Rc}_{34}+R\right) -\zeta^A \wt{\Rc}_{A 3}-\frac{1}{2} \operatorname{tr} \underline{\chi} \wt{\Rc}_{34}-\hat{\chi}^{A B} \wt{\widehat{\Rc}}_{A B}+D_3\phi\square\phi, 
    \end{aligned}
\end{equation}
\begin{equation}\label{Equation_Lorentz_wt_Rc_3}
    \begin{aligned}
        &\frac{1}{2} \nabla_4 \wt{\Rc}_{3 A}= \underline{\omega} \wt{\Rc}_{4 A}+\eta^B \wt{\Rc}_{B A}+\frac{1}{2} \eta_A \wt{\Rc}_{34}-\frac{1}{2} \nabla_3 \wt{\Rc}_{4 A}+\omega \wt{\Rc}_{3 A} +\underline{\eta}^B \wt{\Rc}_{B A}+\frac{1}{2} \underline{\eta}_A \wt{\Rc}_{34}\\
    &\qquad\qquad+\nabla^B \wt{\widehat{\Rc}}_{B A}+\frac{1}{2} \nabla_B \wt{\Rc}_{34} -\frac{1}{2} \operatorname{tr} \underline{\chi} \wt{\Rc}_{4 A}-\frac{1}{2} \operatorname{tr} \chi \wt{\Rc}_{3 A}\\
    &\qquad\qquad-\frac{1}{2} \underline{\chi}_A{ }^B \wt{\Rc}_{B 4}-\frac{1}{2} \chi_A{ }^B \wt{\Rc}_{B 3}+\nabla_A\phi\square\phi.
    \end{aligned}
\end{equation}

We now have estimates for $\widehat{\Rc}_{AB},\Rc_{34},\Rc_{4A},\square\phi$, while $\wt{\Omega^2\Rc_{44}}\equiv 0$ by construction \eqref{Construction_Eq_g_chi}. Using equations \eqref{Equation_Lorentz_wt_Rc_1} and \eqref{Equation_Lorentz_wt_Rc_3}, we obtain
\begin{equation*}
    \begin{aligned}
        &\sum_{j+l\leq N_0-4i-4} \lnm (-u)^{j+l}\left((\Omega\nabla_3)^{(i)}\right)^l(\nabla^{(i)})^j\left(\wt{\Omega\Rc_{3\cdot}}^{(i+1)},\left(\wt{\Omega^2\Rc_{34}}+\wt{\Omega^2 R}\right)^{(i+1)}\right)\rnm_{L^\infty(S_{u,v})}\\
        &\qquad\qquad\qquad\lesssim \frac{v^{i+2}}{(-u)^{i+4+\delta}},
    \end{aligned}
\end{equation*}
and \eqref{eq:dv-Ric3} implies 
\begin{equation*}
    \sum_{j+l\leq N_0-4i-5} \lnm (-u)^{j+l}\left((\Omega\nabla_3)^{(i)}\right)^l(\nabla^{(i)})^j\wt{\Omega^2\Rc_{33}}^{(i+1)}\rnm_{L^\infty(S_{u,v})}\lesssim \frac{v^{i+2}}{(-u)^{i+4+\delta}}.
\end{equation*}

This completes the proof of Theorem {\ref{thm:sec-3}}.
 \section{The A Priori Estimates for Einstein Scalar-Field Equations}\label{Section_Estimates}

\subsection{Main theorem and bootstrap assumption}
In this section, we turn to the a priori estimates for the characteristic initial value problem (CIVP) {for} the {Einstein--scalar field} equations
    \begin{equation*}
        \Rc_{\mu\nu}=D_\mu\phi D_\nu\phi,\quad \square_g \phi=0.
    \end{equation*}
\begin{theorem}
    Fix $N\in\mathbb{N}$ {with} $1\ll N$.
    Consider the CIVP 
    with $(\epsilon,N^2)$-small $\k$-self-similar initial data on $\Hb_0$ and {with} $\chi,\omega,D_4\phi$ along $H_{-1}${.}
    {Assume that}
    \begin{equation*}
        \sup_{v}\lnm \Omega\chi,\Omega\omega,\partial_v\phi\rnm_{H^{N^2}(S_{-1,v})}\leq A{{.}}
    \end{equation*}
    {Then} $\epsilon=\epsilon(N,\k)$ can be chosen sufficiently small so that the CIVP has a solution in $\{-1\leq u<0,\ 0\leq v\leq \epsilon_1(-u)^{1+\frac{\k}{N}}\}$ for a suitably small $\epsilon_1=\epsilon_1(N,\k,A)$.
\end{theorem}

Choosing $\epsilon\ll \frac{\delta}{N^2}\ll \frac{\k}{N^3}$ and setting $\left(\ot{g},\ot{\phi}\right)=\left(g^{(4N)},\phi^{(4N)}\right)$, we find that
\begin{equation*}
    \sum_{j\leq 6}\lnm (-u)^{j+2}\nabla^k\left(\Rc\left(\ot{g}\right)_{\mu\nu}-\partial_\mu \ot\phi\partial_\nu\ot{\phi}\right)\rnm_{L^2(S_{u,v})}\lesssim \frac{v^{4N}}{(-u)^{4N+\delta+\k}}\ll \frac{v^{2N}}{(-u)^{2N-\k}}.
\end{equation*}
This holds for $v\ll (-u)^{1+\frac{\delta+2\k}{N}}\leq (-u)^{1+\delta}$.
Henceforth, we use ${\ot{\cdot}}$ to denote quantities associated with the approximating solution $\left(\ot{g},\ot{\phi}\right)$. The preceding analysis provides the following bounds for the relevant $\ot{\cdot}$ quantities.
\begin{lemma}
    For $\frac{v}{(-u)^{1+\delta}}\leq\epsilon_1\ll 1$ and $j\leq N/2$, the following bounds hold:
    \begin{equation*}
        \lnm (-u)^{j-1}\nabla^j\left(\ot{g}-g|_{v=0},\ot{b}-b|_{v=0},{\log\left(\ot{\Omega}/\Omega|_{v=0}\right)}\right)\rnm_{L^2(S_{u,v})}\lesssim \epsilon_1\ll 1,
    \end{equation*}
    \begin{equation*}
        \lnm (-u)^j\nabla^j\left(\ot{\Omega\chih},\ot{\Omega\omega},\ot{\Omega e_4\phi}\right)\rnm_{L^2(S_{u,v})}\lesssim \frac{1}{(-u)^\delta},
    \end{equation*}
    and, for $\psi=\Omega\tr\chi,\Omega\omegab,\Omega e_3\phi,\nabla\phi,\Omega\chib,\eta,\etab$ and the Gauss curvature $K$,
    \begin{equation*}
        \lnm (-u)^j\nabla^j\left(\ot\psi-\psi|_{v=0}\right),(-u)^{j+1}\nabla^j\left(\ot{K}-K|_{v=0}\right)\rnm_{L^2(S_{u,v})}\lesssim \epsilon_1\ll 1.
    \end{equation*}
\end{lemma}
\begin{proof}
    The estimates for $\ot\psi$ are dominated by the first iterate $(\cdot)^{(1)}$, since the subsequent differences $(\cdot)^{(i+1)}-(\cdot)^{(i)}$ contribute only higher powers of $\frac{v}{(-u)^{1+\delta}}$. The bound for $K$ follows from the estimates for $g$, because it can be treated in terms of second-order derivatives of the metric.
\end{proof}

For all Ricci coefficients, scalar field derivatives, and Weyl curvature components, we denote the difference between a quantity and its approximating counterpart by $\wt{\Phi}_{A_1\cdots A_r}=\Phi_{A_1\cdots A_r}-\ot{\Phi}_{A_1\cdots A_r}$. The residual terms $\wt{\Rc}_{\mu\nu}=\ot{\Rc}_{\mu\nu}-\partial_\mu\ot{\phi}\partial_\nu\ot{\phi}$ and $\wt{\square\phi}=\square_{\ot{g}}\ot{\phi}$ measure the failure of the approximating solution to satisfy the {Einstein--scalar field} equations.
To streamline the analysis, we introduce the following scale-invariant weighted norms.
\begin{definition}
    We define signatures for metric, connection, and curvature components. For a metric component $\gamma\in\{g_{AB},b^A,\log\Omega\}$, we define $s(\Omega^n\gamma)=1$ for any $n\in\mathbb{Z}$. We also define $s(\Omega^n\psi)=0$ for a connection component and $s(\Omega^n\Psi)=-1$ for a curvature component. The signature of a difference is defined by $s(\wt{\cdot})=s(\cdot)$. Finally, we define $s(D\cdot)=s(\cdot)-1$ for $D=\nabla_A,D_3,D_4$.
    We define the spherical norm 
    \begin{equation*}
        \lnm \Phi\rnm_{\LS(u,v;a)}=\left(\int_{S_{u,v}}(-u)^{-2s(\Phi)}\left(\frac{-u}{v}\right)^{2a}\left|\Phi\right|^2dV_g\right)^{\frac{1}{2}},
    \end{equation*}
    the energy norms on $H$ and $\Hb$,
    \begin{equation*}
        \lnm \Phi\rnm_{\LH(u,v;a)}=\left(\int_{0}^v\int_{S_{u,v^\prime}}(-u)^{-1-2s(\Phi)}\left(\frac{-u}{v^\prime}\right)^{2a}\left|\Phi\right|^2dV_g dv^\prime\right)^{\frac{1}{2}},
    \end{equation*}
    \begin{equation*}
        \lnm \Phi\rnm_{\LHb(u,v;a)}=\left(\int_{-1}^u\int_{S_{u^\prime,v}}(-u^\prime)^{-1-2s(\Phi)}\left(\frac{-u^\prime}{v}\right)^{2a}\left|\Phi\right|^2dV_g du^\prime\right)^{\frac{1}{2}},
    \end{equation*}
    and the energy on $(-1,u)\times(0,v)\times \SS$
    \begin{equation*}
        \lnm \Phi\rnm_{\LR(u,v;a)}=\left(\int_0^v\int_{-1}^u\int_{S_{u^\prime,v^\prime}}\frac{1}{(-u^\prime) v^\prime}(-u^\prime)^{-2s(\Phi)}\left(\frac{-u^\prime}{v^\prime}\right)^{2a}\left|\Phi\right|^2dV_g du^\prime dv^\prime\right)^{\frac{1}{2}}.
    \end{equation*}
\end{definition}

\noindent\textbf{Bootstrap Assumption.} With the above norms in hand, we formulate the following energy bootstrap assumptions for $i\leq 4, j\leq 5$:
\begin{equation}\label{eq:Bootstrap_Assumption}
    \begin{aligned}
        &\lnm \nabla^j\left(\wt{\chic},\wt\eta,\wt{\etab},\wt{\nabla\phi},\wt{e_4\phi}\right),\nabla^i\wt{K}\rnm^2_{\LH(u,v;N)}+\lnm \nabla^j\left(\wt{\eta},\wt\etab,\wt{\chibc},\wt{e_3\phi},\wt{\nabla\phi}\right),\nabla^i\wt{K}\rnm^2_{\LHb(u,v;N)}\\
        &\quad+\lnm \nabla^{j}\left(\wt{\tr\chi},\wt{\tr\chib}\right)\rnm^2_{\LS(u,v;N)}+\lnm \nabla^j\left(\wt{\eta},\wt\etab,\wt{\chibc},\wt{e_3\phi},\wt{\nabla\phi}\right),\nabla^i\wt{K}\rnm^2_{\LR(u,v;N)} \leq B\left(\frac{v}{-u}\right)^{3/2}.
    \end{aligned}
\end{equation}
We remark that the particular bound $(-v/u)^{3/2}$ is not important. The proof also closes with $(-v/u)^\lambda$ for any $0\leq \lambda\ll N$.
The remainder of this section is devoted to improving these bootstrap assumptions. By choosing $\frac{v}{(-u)^{1+\k/N}}\leq \epsilon_1$ sufficiently small, we show that the bootstrap {constant} can be improved to $ B \lesssim \frac{v^{1/2}}{(-u)^{1/2}}\ll 1$.

\subsection{Calculation of differences and useful estimates}
We first introduce several useful lemmas {for} the norm $L^2(S_{u,v})$ and then generalize them to $\LS(u,v;a)$.
To control the $L^\infty$ norms of the quantities of interest, we use the Sobolev inequalities. 
\begin{proposition}[Sobolev Inequalities]
    Consider a tensor field $f_{i_1\cdots i_r}$ on $S_{u,v}$. Suppose that 
    $$\sup_{0<v}\left|(-u)^{-2}g(u,v)-{g^{\SS}}\right|_{g^{\SS}}\ll 1,$$
    where $g^{\SS}$ denotes the standard metric on the unit sphere. Then the following Sobolev inequalities hold:
    \begin{equation}\label{lemma:Sobolev_ineq}
        \begin{aligned}
            \lnm (-u)f\rnm_{L^\infty(S_{u,v})}
            \lesssim & \sum_{i\leq 2}\lnm (-u)^i\nabla^i f\rnm_{L^2(S_{u,v})}.
        \end{aligned}
    \end{equation} 
\end{proposition}
\begin{proof}
    Setting $h_{i_1\cdots i_r}(u,v)=(-u)^{r+1}f_{i_1\cdots i_r}(u,v)$, the assumption ensures that $$\left|\nabla^ih\right|_{g^\SS}\sim (-u)^{i+1}\left|\nabla^if\right|_{g}.$$ The desired Sobolev inequalities then follow from the standard interpolation estimate on $\SS$:
    \begin{equation*}
        \lnm h\rnm_{L^\infty(\SS)}\lesssim\lnm h\rnm_{L^4(\SS)}^{1/2}\lnm  h\rnm_{W^{1,4}(\SS)}^{1/2}\lesssim \lnm h\rnm_{L^2(\SS)}^{1/4}\lnm h\rnm^{1/2}_{H^1(\SS)}\lnm  h\rnm_{H^2(\SS)}^{1/4}.
    \end{equation*} 
\end{proof}
The following proposition provides estimates for products of tensors in scale-invariant norms, which are {frequently employed} when handling lower-order nonlinear terms. The proof relies on applying the Sobolev inequalities above to control the $L^\infty$ factors.
\begin{proposition}
    Consider tensor fields $f_i$ on $S_{u,v}$. Suppose that
    \begin{equation}\label{cond:prop-sobolev}
        \sup_{0<v}\left|(-u)^{-2}g(u,v)-{g^{\SS}}\right|_{g^{\SS}}\ll 1,
    \end{equation}
    where $g^{\SS}$ denotes the standard metric on the unit sphere. Let $n\geq 3$. Then the following estimates hold:
\begin{equation}\label{eq:prop-sobolev-1}
    \begin{aligned}
        &\sum_{i\leq n}\lnm (-u)^{l-1+i}\nabla^i(f_{1} \cdots f_l)\rnm_{L^2(S_{u,v})}\\
        &\qquad\lesssim \sum_{s=1}^l\prod_{j\neq s}\left(\sum_{i_j\leq \lfloor(n+1)/2\rfloor}\lnm (-u)^{i_j}\nabla^{i_j}f_{j}\rnm_{L^2(S_{u,v})}\right)\cdot \sum_{i_s\leq n}\lnm(-u)^{i_s}\nabla^{i_s}f_s\rnm_{L^2(H_u^v)},
    \end{aligned}
\end{equation}
and
\begin{equation}\label{eq:prop-sobolev-2}
    \begin{aligned}
        \sum_{i\leq n}&\lnm (-u)^{l-1+i}\nabla^i(f_{1} \cdots f_l)\rnm_{L^2(H_u^v)}\\
        &\qquad\lesssim  \sum_{s=1}^l\prod_{j\neq s}\left(\sum_{i_j\leq \lfloor(n+1)/2\rfloor}\sup_v\lnm (-u)^{i_j}\nabla^{i_j}f_{j}\rnm_{L^2(S_{u,v})}\right)\cdot \sum_{i_s\leq n}\lnm(-u)^{i_s}\nabla^{i_s}f_s\rnm_{L^2(H_u^v)},
    \end{aligned}
\end{equation}
\begin{equation}\label{eq:prop-sobolev-3}
    \begin{aligned}
        \sum_{i\leq n}&\lnm (-u)^{l-1+i}\nabla^i(f_{1} \cdots f_l)\rnm_{L^2(\Hb_v^u)}\\
        &\qquad\lesssim  \sum_{s=1}^l\prod_{j\neq s}\left(\sum_{i_j\leq \lfloor(n+1)/2\rfloor}\sup_u\lnm (-u)^{i_j}\nabla^{i_j}f_{j}\rnm_{L^2(S_{u,v})}\right)\cdot \sum_{i_s\leq n}\lnm(-u)^{i_s}\nabla^{i_s}f_s\rnm_{L^2(\Hb_v^u)}.
    \end{aligned}
\end{equation}
\end{proposition}
\begin{proof}
    We recall the corresponding inequalities for the $L^2(\SS)$ norm. The normalized inequalities then follow from condition \eqref{cond:prop-sobolev}.
    For $n=2k$ and $k\geq 2$, we have 
    \begin{equation*}
        \begin{aligned}
            &\|f_1\cdots f_l\|_{H^n(\SS)}
            \lesssim
            \sum_{p\leq l}\left(\| f_p\|_{W^{n-2,\infty}(\SS)}\prod_{m\neq p}\| f_m\|_{H^k(\SS)}\right.\\
            &\qquad\qquad\qquad\left.+ \|f_p\|_{H^n(\SS)}\prod_{m\neq p}\|f_m\|_{L^\infty(\SS)}+\|f_p\|_{W^{n-1,4}(\SS)}\prod_{m\neq p}\|f_m\|_{W^{1,4}(\SS)}\right)\\
            \lesssim &
            \sum_{p\leq l}\left(\| f_p\|_{H^n(\SS)}\prod_{m\neq p}\| f_m\|_{H^k(\SS)}+ \|f_p\|_{H^n(\SS)}\prod_{m\neq p}\|f_m\|_{H^2(\SS)}+\|f_p\|_{H^n(\SS)}\prod_{m\neq p}\|f_m\|_{H^2(\SS)}\right),
        \end{aligned}
    \end{equation*}
    and for $n=2k-1$, $k\geq 2$,
    \begin{equation*}
        \begin{aligned}
            &\|f_1\cdots f_l\|_{H^n(\SS)}
            \lesssim 
            \sum_{p\leq l} \left(\|f_p\|_{H^n(\SS)}\prod_{m\neq p}\|f_m\|_{L^\infty(\SS)}+ \|f_p\|_{W^{n-1,4}(\SS)}\prod_{m\neq p}\|f_m\|_{W^{k-1,4}(\SS)} \right)\\
            \lesssim &
            \sum_{p\leq l}\left(\|f_p\|_{H^n(\SS)}\prod_{m\neq p}\|f_m\|_{H^k(\SS)}\right).
        \end{aligned}
    \end{equation*}
    We thus obtain \eqref{eq:prop-sobolev-1}, and \eqref{eq:prop-sobolev-2}, \eqref{eq:prop-sobolev-3} follow after integration.
\end{proof}
To deal with $\nabla^{i_1}(\eta+\etab)^{i_2}\nabla^{i_3}(\chi,\chib)\nabla^{i_4}\psi$ arising from the commutation formulae, we state a corollary of the preceding proposition.
\begin{corollary}\label{key_corollary}
    Under the conditions of the preceding proposition, the following multilinear estimate holds:
    \begin{equation*}
        \begin{aligned}
            &\sum_{i\leq 3}\sum_{i_1+i_2+i_3+i_4= i}\lnm (-u)^i\nabla^{i_1}\psi_1\nabla^{i_2}\psi_2^{i_3}\nabla^{i_4}\psi_4\rnm^2_{L^2(S_{u,v})}\\
            \lesssim & \frac{1}{(-u)^2}\sum_{i_1,i_2,i_3,i_4\leq 3}\lnm (-u)^{i_1}\nabla^{i_1}\psi_1\rnm^2_{L^2(S_{u,v})}\lnm (-u)^{i_2}\nabla^{i_2}\psi_2\rnm^{2i_3}_{L^2(S_{u,v})}\lnm (-u)^{i_4}\nabla^{i_4}\psi_4\rnm^2_{L^2(S_{u,v})}.
        \end{aligned}
    \end{equation*}
\end{corollary}
The next lemma provides a systematic way to estimate differences of products, which is particularly useful for handling scale-critical quantities in the iteration scheme.
\begin{lemma}\label{Estimate_Lemma_Calculation_of_Differences}
    Let $n\geq 3$ be a positive integer, and consider tensors $\psi_1,\cdots,\psi_m $ and $\ot{\psi}_1,\cdots,\ot\psi_m$ with $\wt{\psi}_i=\psi_i-\ot{\psi}_i$.
    Then the following estimates hold:
    \begin{equation}\label{eq:prop-diff-sobolev-1}
        \begin{aligned}
            &\lnm {{\psi_{1}\cdots{\psi_{m-1}}}\wt\psi_m}\rnm_{H^n(\SS)}+
            \lnm \wt{{\psi_{1}\cdots{\psi_m}}}\rnm_{H^n(\SS)}\\
            &\qquad\lesssim \sum_{i=1}^{m}\left(\lnm \wt\psi_i\rnm_{H^n(\SS)}\cdot\prod_{j\neq i}\left(\lnm\wt{\psi}_i\rnm_{H^{\lfloor (n+1)/2\rfloor}(\SS)}+\lnm \ot\psi_j\rnm_{H^n(\SS)}\right)\right).
        \end{aligned}
    \end{equation}
    If we further assume
    \begin{equation*} 
        \sup_{0<v}\left|(-u)^{-2}g(u,v)-{g^{\SS}}\right|_{g^{\SS}}\ll 1,
    \end{equation*}
    the inequalities above can be rewritten in the $\LS$ norm as
    \begin{equation*}
        \begin{aligned}
            &\sum_{j\leq n}\lnm \nabla^j\left(\wt\psi_1\psi_2\cdots\psi_m\right)\rnm_{\LS(u,v;a)}+
            \sum_{j\leq n}\lnm \nabla^j\left(\wt{\psi_1\cdots\psi_m}\right)\rnm_{\LS(u,v;a)}\\
            &\qquad\lesssim  \sum_{l=1}^m\prod_{i\neq l}\left(\sum_{j_i\leq n}\lnm \nabla^{j_i}\ot\psi_i\rnm_{\LS(u,v;a_i)}+1\right)\cdot\sum_{j_l\leq n}\lnm \nabla^{j_l}\wt\psi_l\rnm_{\LS(u,v;a_l)}.
        \end{aligned}
    \end{equation*}
\end{lemma}
\begin{proof}
    {Both} $\wt{[\psi_1\cdots\psi_m]}$ and $\psi_1\cdots\psi_{m-1}\wt{\psi}_m$ {can be written} schematically as
    \[\sum_{k=1}^m\sum_{\sigma\in S_m}\wt\psi_{\sigma(1)}\cdots\wt\psi_{\sigma(k)}\ot{\psi}_{\sigma(k+1)}\cdots\ot\psi_{\sigma(m)}.\]
    For example, we estimate $\wt\psi_1\cdots\wt\psi_k\ot\psi_{k+1}\cdots\ot\psi_{m}$. Applying \eqref{eq:prop-sobolev-1}, we find 
    \begin{equation*}
        \begin{aligned}
            &\lnm \wt\psi_1\cdots\wt\psi_k\ot\psi_{k+1}\cdots\ot\psi_{m}\rnm_{H^n(\SS)}\lesssim \sum_{s=1}^k\lnm \wt\psi_s\rnm_{H^{n}(\SS)}\prod_{j\leq k, j\neq s}\lnm \ot\psi_j\rnm_{H^{\lfloor (n+1)/2\rfloor}(\SS)}\prod_{j\geq k+1}\lnm \wt\psi_j\rnm_{H^{\lfloor (n+1)/2\rfloor}(\SS)}\\
            &\qquad\qquad+\sum_{s=k+1}^m\lnm \ot\psi_s\rnm_{H^{n}(\SS)}\prod_{j\leq k}\lnm \ot\psi_j\rnm_{H^{\lfloor (n+1)/2\rfloor}(\SS)}\prod_{j\geq k+1,j\neq s}\lnm \wt\psi_j\rnm_{H^{\lfloor (n+1)/2\rfloor}(\SS)}\\
            &\qquad \lesssim \sum_{i=1}^{m}\left(\lnm \wt\psi_i\rnm_{H^n(\SS)}\cdot\prod_{j\neq i}\left(\lnm\wt{\psi}_i\rnm_{H^{\lfloor (n+1)/2\rfloor}(\SS)}+\lnm \ot\psi_j\rnm_{H^n(\SS)}\right)\right),
        \end{aligned}
    \end{equation*}
    which {proves} \eqref{eq:prop-diff-sobolev-1}.
\end{proof}

To streamline the subsequent analysis, we recast the transport and energy estimates in the weighted norm framework.
First consider the transport equations
\begin{equation}\label{Transport_Est_Sample_Eq}
    \begin{aligned}
        &\Omega \nabla_3\psi_1+\left(\lambda_1\Omega\tr\chib+\mu_1\Omega\omegab\right)\psi_1=\psi_{11}\psi_{12}+\Psi_{1},\\ 
        &\Omega \nabla_4\psi_2+\left(\lambda_2\Omega\tr\chi+\mu_2\Omega\omega\right)\psi_2=\psi_{21}\psi_{22}+\Psi_{2}.
    \end{aligned}
\end{equation}
We assume analogous equations for $\ot{\psi}_i,\ot{\Psi}_i$ and require all quantities to coincide with their corresponding $\ot{\cdot}$ quantities along $v=0$; that is, $\wt{\cdot}|_{v=0}=0$. 
\begin{lemma}[Transport Estimate]\label{Transport_Estimate_LS} 
    Consider equations \eqref{Transport_Est_Sample_Eq} and assume $\lambda_1+\mu_1+\lambda_2+\mu_2+n\ll N$, $n\geq 3$. Suppose {that, for $j\leq n+1$, we have}
    $$\sum_{j\leq n+1}\lnm (-u)^j\nabla^j\left(\ot{\psi}_1,\ot{\psi}_{1i}\right)\rnm_{L^2(S_{u,v})}+(-u)^\delta\lnm (-u)^j\nabla^j\left(\ot{\psi}_2,\ot{\psi}_{2i}\right)\rnm_{L^2(S_{u,v})}\lesssim 1,$$
    and
     $$\sum_{j\leq n}\lnm (-u)^j\nabla^j\left(\Omega\chib\right)\rnm_{L^2(S_{u,v})}+ (-u)^\delta\lnm (-u)^j\nabla^j\left(\Omega\chi\right)\rnm_{L^2(S_{u,v})}\lesssim 1.$$ 
    Let $g(u,v)$ and $g(u,0)$ be close to each other, in the sense that
    
    $$\sup_{0<v^\prime<v}\left|\sqrt{\det\left(g(u,v^\prime)g^{-1}(u,v)\right)}-1\right| \ll 1.
    $$
    Moreover, assume that all $\wt{\cdot}$ quantities are small, meaning that their norms of interest are much less than $1$. Then we have the estimates
    \begin{equation*}
        \begin{aligned}
            &\sum_{j\leq n}\left(\lnm \nabla^j\wt{\psi}_1\rnm_{\LS(u,v;N)}^2-\lnm \nabla^j\wt{\psi}_1\rnm_{\LS(-1,v;N)}^2+N\lnm \nabla^j\wt{\psi_1}\rnm^2_{\LHb(u,v;N)}\right) \\
            &\qquad\lesssim \sum_{i\leq n+1,j\leq n}\lnm \nabla^i\wt{b},\nabla^j\left(\wt{\psi}_{11},\wt{\psi}_{12},\wt{\Omega\chib},\wt{\Psi}_1\right)\rnm_{\LHb(u,v;N)}^2,\\
            &\sum_{j\leq n}\lnm \nabla^j\wt{\psi}_2 \rnm_{\LS(u,v;N)}^2\lesssim \sum_{j\leq n}\frac{v}{(-u)^{1+C(n)\delta}}\lnm \nabla^j\left(\wt{\psi}_{21},\wt{\psi}_{22},\wt{\Omega\chi}\right)\rnm_{\LH(u,v;N)}^2\\
        &\qquad\qquad\qquad\qquad+\frac{v}{-u}\sum_{j\leq n}\lnm \nabla^j\wt{\Psi}_1\rnm_{\LH(u,v;N)}^2.
        \end{aligned}
    \end{equation*}

\end{lemma}
\begin{proof}
    We first take the difference of the equations in \eqref{Transport_Est_Sample_Eq} and obtain
    \begin{equation*}
        \begin{aligned}
            \Omega\nabla_3\wt{\psi}_1=&\wt{\psi_{11}\psi_{12}}+\wt{\Psi_1}-\lambda_1\wt{\Omega\tr\chib\psi_1}-\mu_1\wt{\Omega\omegab\psi_1}-\Lie_{\wt{b}}\ot{\psi}_1+\wt{\Omega\chib}\ot{\psi}_1,\\
            \Omega\nabla_4\wt{\psi}_2=& \wt{\psi_{21}\psi_{22}}+\wt{\Psi_2}-\lambda_2\wt{\Omega\tr\chi\psi_2}-\mu_2\wt{\Omega\omega\psi_2}+\wt{\Omega\chi}\ot{\psi}_2.
        \end{aligned}
    \end{equation*}
    We note that {\small$$\Lie_{\wt{b}}\left(\ot{\psi}_1\right)_{A_1\cdots A_r}=\wt{b}\cdot\nabla\left(\ot{\psi}_1\right)_{A_1\cdots A_r}+\frac{1}{2}\sum_{i=1}^r\left(\nabla_{A_i}\wt{b}^B+\nabla^{B}\wt{b}_{A_i} \right)\left(\ot\psi_1\right)_{A_1\cdots B\cdots A_r}, $$}\\
    so the main contribution of $\nabla^j\left(\Lie_{\wt{b}}\ot{\psi_1}\right)$ {in} the scale-invariant norm is $\nabla^{j+1}\wt{b}$.
    Using the commutation formula, we have 
    \begin{equation*}
        \begin{aligned}
            \Omega\nabla_3\nabla^j\wt{\psi}_1\sim &\nabla^i(\Omega\chib\wt{\psi}_1)+\nabla^j\left(\Omega\nabla_3\wt{\psi}_1\right),\ 
            \Omega\nabla_4\nabla^j\wt{\psi}_2\sim&\nabla^j\left(\Omega\chi\wt{\psi}_2\right)+\nabla^j\left(\Omega \nabla_4\wt\psi_2\right).
        \end{aligned}
    \end{equation*}
    We first estimate $\wt{\psi}_1$. For $j\leq n$, we compute 
    \begin{equation*}
        \begin{aligned}
            &\int_{S_{u,v}}(-u)^{2j}\left(\frac{-u}{v}\right)^{2N}\left|\nabla^j\wt{\psi}_1\right|^2-\int_{S_{-1,v}}\left(\frac{1}{v}\right)^{2N}\left|\nabla^j\wt{\psi}_1\right|^2\\
            =&\int_{-1}^u\int_{S_{u^\prime,v}}\left(-\frac{2N+2j}{-u^\prime}+\Omega\tr\chib\right)(-u^\prime)^{2N+2j}v^{-2N}\left|\nabla^j\wt{\psi}_1\right|^2\\
        &\qquad\qquad\qquad+2(-u^\prime)^{2N+2j}v^{-2N}\nabla^j\wt{\psi}_1\cdot\Omega \nabla_3\left(\nabla^j\wt{\psi}_1\right)\\
            \lesssim & \int_{-1}^u\int_{S_{u^\prime,v}}-\frac{N}{-u^\prime}(-u^\prime)^{2N+2j}v^{-2N}\left|\nabla^j\wt{\psi}_1\right|^2+2(-u^\prime)^{2N+2j+1}v^{-2N}\left|\Omega \nabla_3\left(\nabla^j\wt{\psi}_1\right)\right|^2.
        \end{aligned}
    \end{equation*}
    Lemma \ref{Estimate_Lemma_Calculation_of_Differences} yields
    \begin{equation*}
        \begin{aligned}
            &\sum_{j\leq n}\int_{S_{u,v}}(-u)^{2j+1}\left|\Omega\nabla_3\left(\nabla^j\wt{\psi}_1\right)\right|^2\\
            \lesssim & \sum_{j\leq n,i\leq n+1}\int_{S_{u,v}}C(\lambda_1,\mu_1,n)(-u)\left(\left|\Omega\tr\chib\right|^2+\left|\Omega\omegab\right|^2\right)(-u)^{2j}\left|\nabla^j\wt{\psi}_1\right|^2+(-u)^{2j+1}\left|\nabla^j\wt{\Psi}_1\right|\\
            &\qquad\qquad+(-u)^{2j-1}\left|\nabla^j\left(\wt{\Omega\chib},\wt{\Omega\omegab},\wt\psi_{11},\wt\psi_{12}\right)\right|^2+(-u)^{2i-3}\left|\nabla^{i}\wt{b}\right|^2.
        \end{aligned}
    \end{equation*}
    Using $C(\lambda_1,\mu_1,n)\ll N$ and $|\Omega\tr\chib|+|\Omega\omegab|\lesssim\frac{1}{-u}$, we can then cancel $\nabla^j\wt{\psi}_1$ terms and obtain the desired estimate.
    {Next}, for $\nabla^j\wt{\psi}_2$, the transport estimate
    \[\|f(v)-f(0)\|_{L^2(\SS)}^2=\lnm \int_0^v\Lie_v f\rnm^2_{L^2(\SS)}\leq v\int_0^v\|\Lie_v f\|^2_{L^2(\SS)}\]
    {gives}
    \begin{equation*}
        \begin{aligned}
            &\int_{S_{u,v}}(-u)^{2j}\left(\frac{-u}{v}\right)^{2N}\left|\nabla^j\wt{\psi}_2\right|^2\\
            \lesssim & \frac{v}{-u}\left(\frac{-u}{v}\right)^{2N}\int_{0}^v\int_{S_{u,v^\prime}} \left((-u)^{2j+1}\left|\Omega\nabla_4\left(\nabla^j\wt\psi_2\right)\right|^2+(-u)^{2j+1}\left|\Omega\chi\right|^2\left|\nabla^j\wt{\psi}_2\right|^2\right).
        \end{aligned}
    \end{equation*}
    Employing Lemma \ref{Estimate_Lemma_Calculation_of_Differences} and 
    $$\sum_{j\leq n}\lnm \nabla^j\left(\ot{\psi}_2,\ot{\psi}_{2i},{\Omega\chi},{\Omega\omega},{\etab}\right)\rnm_{\LS(u,v;0)}+\sum_{j\leq n-1}\lnm\nabla^{j}{\Omega\beta}\rnm_{\LS(u,v;0)}\lesssim \frac{1}{(-u)^\delta},$$ 
    we obtain 
    \begin{equation*}
        \begin{aligned}
            &\sum_{j\leq n}\int_{S_{u,v}} (-u)^{2j+1}\left|\Omega\nabla_4\left(\nabla^j\wt\psi_2\right)\right|^2\\
            \lesssim & \frac{1}{-u}\sum_{j\leq n}\frac{1}{(-u)^{C(n)\delta}}\int_{S_{u,v}}(-u)^{2j}\left|\nabla^j\left(\wt{\psi}_2,\wt{\psi}_{2i},\wt{\Omega\chi},\wt{\Omega\omega}\right)\right|^2+\frac{1}{-u}\int_{S_{u,v}}(-u)^{2j+2}\left|\nabla^j\wt{\Psi}_2\right|^2.
        \end{aligned}
    \end{equation*}
    Thus, 
        
        \begin{equation*}
            \begin{aligned}
            &\sum_{j\leq n}\lnm \nabla^j\wt{\psi}_2 \rnm_{\LS(u,v;N)}^2\lesssim \frac{v}{(-u)^{1+C(n)\delta}}\lnm \nabla^j\left(\wt{\psi}_{21},\wt{\psi}_{22},\wt{\Omega\chi},\wt{\Omega\omega}\right)\rnm_{\LH(u,v;N)}^2\\
            &\qquad\qquad\qquad+\frac{v}{-u}\lnm \nabla^j\wt{\Psi}_1\rnm_{\LH(u,v;N)}^2+\sum_{j\leq n}\frac{v}{(-u)^{1+C(n)\delta}}\int_0^v \frac{1}{-u}\lnm \nabla^j\wt{\psi}_2\rnm_{\LS(u,v^\prime;N)}^2,
        \end{aligned}
        \end{equation*}
    and {if} $v\ll (-u)^{1+C(n)\delta}$, {the} desired estimate holds for $\wt{\psi}_2$.

\end{proof}

We proceed to establish the energy estimate. 
\begin{lemma}[Energy Estimate]\label{Energy_Estimate_LH_LHb}
    We now consider the equations 
    \begin{equation*}
        \begin{aligned}
\Omega\nabla_3\Psi_1-\Omega\mathcal{D}\Psi_2=&\psi\Psi+\psi^3+\psi_1\nabla(\Omega\chic)+\psi_1\nabla(\Omega e_4\phi),\\
        \Omega\nabla_4\Psi_2+{}^*\mathcal{D}(\Omega\Psi_1) =&\psi\Psi+\psi^3+\psi\nabla(\Omega\chic)+\psi\nabla(\Omega e_4\phi),
    \end{aligned}
    \end{equation*}
    for the pair $\Psi_1,\Psi_2$. Here ${}^*\mathcal{D}$ stands for the Hodge dual of the differential operator $\mathcal{D}$.
    Suppose that
    \begin{equation*}
        \sum_{j\leq n+1}\lnm \nabla^j\ot\psi_1\rnm_{\LS(u,v;0)}  +\sum_{j\leq n+1}(-u)^\delta\lnm \nabla^j\left(\ot{\Psi}_1,\ot{\Psi}_2,\ot{\psi}\right)\rnm_{\LS(u,v;0)}\lesssim 1,
    \end{equation*}
    \begin{equation*}
        \sum_{i\leq n-1}\lnm \nabla^{i}K\rnm_{\LS(u,v;0)}+\sum_{j\leq n}\lnm \nabla^j\left(\Omega\chib\right)\rnm_{\LS(u,v;0)}+\sum_{j\leq n}(-u)^\delta
        \lnm \nabla^j\left(\Omega\chi\right)\rnm_{\LS(u,v;0)}\lesssim 1,
    \end{equation*}
    and the $\wt{\cdot}$ quantities are small, that is{,} $\sum_{i=1,2}\sum_{j\leq n}\lnm \nabla^j\left(\wt{\Psi}_i,\wt{\psi}_i\right)\rnm_{\LH(u,v;0)}\ll 1.$
    Then, for $1\ll N$ and $\frac{v}{(-u)^{1+C(n)\delta}}\ll 1$, we have the estimate
    \begin{equation}\label{eq:Energy_Estimate_LH_LHb}
        \begin{aligned}
            &\sum_{j\leq n}\left(\lnm\nabla^j\wt{\Psi}_1 \rnm_{\LH(u,v;N)}^2-\lnm\nabla^j\wt{\Psi}_1 \rnm_{\LH(-1,v;N)}^2\right)+\lnm\nabla^j\wt{\Psi}_2 \rnm_{\LHb(u,v;N)}^2\\
            &\qquad\qquad+\lnm\nabla^j\wt{\Psi}_2 \rnm_{\LR(u,v;N)}^2+N\int_{-1}^u(-u^\prime)^{-1}\lnm \nabla^{n}\wt{\Psi_1}\rnm^2_{\LH(u^\prime,v;N)}\\
            \lesssim & \frac{1}{(-u)^{C(n)\delta}}\int_{0}^v\int_{-1}^u\frac{1}{(-u^\prime)^2}\left(\sum_{j\leq n+1,\ i\leq n}\lnm \nabla^j\left(\wt{g},\wt{b}\right),\nabla^i\wt{\Psi}, \nabla^i\wt{\psi}\rnm^2_{\LS(u^\prime,v^\prime;N)}\right)\\
            &+\int_{0}^v\int_{-1}^u\frac{1}{(-u^\prime)^2}\lnm \nabla^{n+1}\left(\wt{\Omega\chic},\wt{\Omega e_4\phi}\right)\rnm^2_{\LS(u^\prime,v^\prime;N)}.
        \end{aligned}
    \end{equation}

\end{lemma}
\begin{remark}
    In the remainder of this section, we use $\Psi$ to represent curvature components and first-order derivatives of Ricci coefficients other than $\nabla\chic$, as well as second-order derivatives of $\phi$ other than $\nabla(e_4\phi)$; we use $\psi$ to represent the Ricci coefficients. More specifically, we require 
\[\Psi\in\{\Omega K, \nabla(\Omega\tr\chi),\nabla(\Omega\tr\chib),\nabla(\Omega\chibc),\Omega\nabla\eta,\Omega\nabla\etab,\nabla(\Omega e_3\phi),\nabla^2\phi\},\]
\[\psi\in\{\Omega\omega,\Omega\omegab,\eta,\etab,\zeta,\Omega\tr\chi,\Omega\chic,\Omega\tr\chib,\Omega\chibc,\Omega e_3\phi,\nabla\phi,\Omega e_4\phi\}.\]

We also highlight the role of 
$$N\int_{-1}^u(-u^\prime)^{-1}\lnm \nabla^{n}\wt{\Psi_1}\rnm^2_{\LH(u^\prime,v;N)}$$ on the left-hand side of the estimate. The top-order terms of $\chic$ and $e_4\phi$, specifically $\nabla^4\dv(\Omega\chic)$ and $\nabla^5 (\Omega e_4\phi)$, are controlled through the energy estimates and act as $\Psi_1$ in the model equations. Therefore, this integral can absorb the last term on the right-hand side.
\end{remark}
\begin{proof}
    We have the equations for the differences:
    \begin{equation*}
        \begin{aligned}
            &\Omega\nabla_3\wt\Psi_1-\Omega\mathcal{D}\left(\wt\Psi_2\right)=\wt{\psi\Psi}+\wt{\psi^3}+\wt{\psi_1\nabla(\Omega\chic)}+\wt{\psi_1\nabla(\Omega e_4\phi)}-\Lie_{\wt{b}}\ot\Psi_1+\wt{\Omega\chib}\ot\Psi_1+\wt{\Omega\mathcal{D}}\ot\Psi_2,\\
            &\Omega\nabla_4\wt\Psi_2+{}^*\mathcal{D}\left(\Omega\wt\Psi_1\right)=\wt{\psi\Psi}+\wt{\psi^3}+\wt{\psi\nabla(\Omega\chic)}+\wt{\psi\nabla(\Omega e_4\phi)}+\wt{\Omega\chi}\ot\Psi_2+\wt{{}^*\mathcal{D}\Omega}\ot{\Psi}_2.
        \end{aligned}
    \end{equation*}
    We can write the schematic identities
    $\wt{\mathcal{D}\Omega}\ot\Psi_2\sim \wt{\Omega}\ot{\nabla\Psi_2}+\Omega\left(\wt{g}\ot{\nabla\Psi}_2+\ot{\Psi}_2\nabla\wt{g}\right)$
    and $$ \wt{\mathcal{D}^*\Omega}\ot{\Psi}_2\sim \nabla\wt{\Omega}\ot{\Psi_2}+\wt{\Omega}\ot{\nabla\Psi_2}+\left(\wt{g}\nabla\ot\Omega+\ot\Omega\nabla\wt{g}\right)\ot{\Psi_2}+\Omega\left(\wt{g}\ot{\nabla\Psi}_2+\ot{\Psi}_2\nabla\wt{g}\right).$$
    Commutation formulae give
    \begin{equation}\label{Sample_Energy_est_commute}
        \begin{aligned}
            &\Omega\nabla_3\nabla^j\wt\Psi_1-\Omega\mathcal{D}\nabla^j\wt\Psi_2=\sum_{i_1+i_2+i_3=j}\Omega\nabla^{i_1}(\eta+\etab)^{i_2}\nabla^{i_3}\wt{\Psi_2}\\
            &\qquad\qquad+\sum_{i_1+\cdots+i_5=j-1}\Omega\nabla^{i_1}(\eta+\etab)^{i_2} \nabla^{i_3}K^{i_4+1}\nabla^{i_5}\wt{\Psi_2} +\nabla^j\left(\Omega\nabla_3\wt{\Psi}_1\right)+\nabla^{j}(\Omega\chib\wt{\Psi}_1),\\
            &\Omega\nabla_4\wt\Psi_2+{}^*\mathcal{D}\left(\Omega\nabla^j\wt\Psi_1\right)=\sum_{i_1+i_2+i_3=j}\Omega\nabla^{i_1}(\eta+\etab)^{i_2}\nabla^{i_3}\wt{\Psi_1}\\
            &\qquad\qquad
            + \sum_{i_1+\cdots+i_5=j-1}\Omega\nabla^{i_1}(\eta+\etab)^{i_2} \nabla^{i_3}K^{i_4+1}\nabla^{i_5}\wt{\Psi_1}+\nabla^j\left(\Omega\nabla_4\wt{\Psi}_2\right)+\nabla^j(\Omega\chi\wt{\Psi}_2).
        \end{aligned}
    \end{equation}
    We compute
    \begin{equation*}
        \begin{aligned}
            &\left(\lnm\nabla^j\wt{\Psi}_1 \rnm_{\LH(u,v;N)}^2-\lnm\nabla^j\wt{\Psi}_1 \rnm_{\LH(-1,v;N)}^2\right)+\lnm\nabla^j\wt{\Psi}_2 \rnm_{\LHb(u,v;N)}^2+\lnm\nabla^j\wt{\Psi}_2 \rnm_{\LR(u,v;N)}^2\\
            \leq & \int_0^v\int_{-1}^u -\frac{2N}{(-u^\prime)^2}\lnm \nabla^j\wt\Psi_1\rnm_{\LS(u^\prime,v^\prime;N)}^2-\frac{2N}{(-u)^\prime v^\prime}\lnm \nabla^j\wt\Psi_2\rnm_{\LS(u^\prime,v^\prime;N)}^2\\
            &\qquad+\frac{1}{(-u^\prime)^2}\lnm \nabla^j\left(\Omega \nabla_3\wt\Psi_1-\Omega \mathcal{D}\nabla^j\wt\Psi_2\right)\rnm_{\LS(u^\prime,v^\prime;N)}^2\\
            &\qquad+\frac{v^\prime}{(-u^\prime)^3}\lnm \nabla^j\left(\Omega \nabla_4\wt\Psi_2+ {}^*\mathcal{D}\left(\Omega\nabla^j\wt\Psi_2\right)\right)\rnm_{\LS(u^\prime,v^\prime;N)}^2\\
            &\qquad +2\frac{(-u^\prime)^{2j+2N}}{(v^\prime)^{2N}}\int_{S_{u^\prime,v^\prime}}\left(\left\langle \nabla^j\wt\Psi_1,\Omega\mathcal{D}\nabla^j\wt\Psi_2 \right\rangle-\left\langle \nabla^j\wt\Psi_2,{}^*\mathcal{D}\left(\Omega\nabla^j\wt\Psi_1\right) \right\rangle\right).
        \end{aligned}
    \end{equation*}
    Since the last line vanishes after integration by parts, it suffices to estimate the terms in \eqref{Sample_Energy_est_commute}.
Using Lemma \ref{Estimate_Lemma_Calculation_of_Differences}, we find 
\begin{equation*}
    \begin{aligned}
        &\sum_{j\leq n}\lnm \nabla^j\left(\Omega \nabla_3\wt\Psi_1-\Omega \mathcal{D}\nabla^j\wt\Psi_2\right)\rnm_{\LS(u^\prime,v^\prime;N)}^2\lesssim \lnm \nabla^{n+1}\left(\wt{\Omega\chic},\wt{\Omega e_4\phi}\right)\rnm_{\LS(u^\prime,v^\prime;N)}^2\\
        &\qquad\qquad
        +  \frac{1}{(-u)^{C(n)\delta}} \sum_{j\leq n}\lnm \nabla^j\left(\wt\Psi,\wt\psi,\wt{b},\wt{\log\Omega},\wt{g}\right),\nabla^{n+1}\left(\wt{g},\wt{b}\right)\rnm_{\LS(u^\prime,v^\prime;N)}^2,
    \end{aligned}
\end{equation*}
and 
\begin{equation*}
    \begin{aligned}
        & \frac{v}{-u}\sum_{j\leq n}\lnm \nabla^j\left(\Omega \nabla_4\wt\Psi_2+ {}^*\mathcal{D}\left(\Omega\nabla^j\wt\Psi_2\right)\right)\rnm_{\LS(u^\prime,v^\prime;N)}^2\\
        \lesssim & \frac{v}{(-u)^{1+C(n)\delta}} \sum_{i\leq n,j\leq n+1}\lnm \nabla^i(\wt\Psi,\wt{\psi}), \nabla^j\left(\wt{\log\Omega},\wt{g}\right),\nabla^{n+1}\left(\wt{\Omega\chic},\wt{\Omega e_4\phi}\right)\rnm_{\LS(u^\prime,v^\prime;N)}^2.
    \end{aligned}
\end{equation*}
After absorbing the $\wt\Psi_i$ terms using $-\frac{N}{-u}\lnm\nabla^j\wt\Psi_1\rnm^2-\frac{N}{v}\lnm\nabla^j\wt\Psi_2\rnm^2$, we obtain the desired estimate. Note that we retain $-\frac{N}{-u}\lnm\nabla^j\wt\Psi_1\rnm^2$ on the right-hand side for the last term in \eqref{Energy_Estimate_LH_LHb}.

\end{proof}
\begin{remark}
    The essential observation for the energy estimates is that $\chic$ and $e_4\phi$ are the only quantities without a top-order $\LHb$ estimate. It is important that the coefficients of their top-order terms contain neither $\chic$ nor $e_4\phi$. For example, $\chic\cdot\nabla^4(\dv\chic,\nabla \partial_v\phi)$ does not appear in the $\Omega \nabla_3\nabla^4\Psi_1$ equation, so {terms such as}
    \[\int_{0}^v\int_{-1}^u\frac{1}{(-u^\prime)^{2+C\delta}}\lnm \nabla^{n+1}\left(\wt{\Omega\chic},\wt{\Omega e_4\phi}\right)\rnm^2_{\LS(u^\prime,v^\prime;N)}\]
    {do not appear} on the right-hand side of the inequalities.
\end{remark}

We also need the following elliptic estimate on the spheres $S_{u,v}$.
\begin{proposition}[Perturbative Elliptic Estimate]\label{prop:perturbative-elliptic-estimate}
    Let $n\geq 3$, and let $g,\ot{g}$ be two metrics on a sphere. Let $K,\ot{K}$ be the respective Gauss curvatures of $g$ and $\ot{g}$, with
    \[
        \sum_{a=0}^{n-1}\|\nabla^a\ot K\|_{L^2(g)}\leq M,
    \]
    and set $\wt K=K-\ot K$. Suppose that a totally symmetric $(r+1)$-covariant tensor $\phi$ satisfies
    \[
        \dv\phi=f,\qquad \cl\phi=l,\qquad \tr\phi=h.
    \]
    Then
    \begin{equation}\label{Elliptic_Estimate}
        \begin{aligned}
            \|\nabla^n\phi\|_{L^2(g)}
            \lesssim_{n,r,M}{}& \sum_{j=0}^{n-1}\|\nabla^j(\phi,f,l,h)\|_{L^2(g)}
            +\mathcal E_n(\wt K;\phi,f,l,h),\\
            \mathcal E_n(\wt K;\phi,f,l,h)
            :={}&\sum_{\substack{a,c\geq0,\ b\geq1\\a+c+2b=n,\ c\leq n-1}}
            \|\nabla^a\wt{K^b}*\nabla^c(\phi,h)\|_{L^2(g)}\\
            &+\sum_{\substack{a,c\geq0,\ b\geq1\\a+c+2b=n-1}}
            \|\nabla^a\wt{K^b}*\nabla^c(f,l)\|_{L^2(g)}.
        \end{aligned}
    \end{equation} 
\end{proposition}

\begin{proof}
    The basic div--curl--trace identity and the Hodge systems satisfied by the symmetrized covariant derivatives of $\phi$ are {given} in \cite[Propositions 15--17]{Luk2012}. Luk's computation shows that the curvature terms at order $n$ are precisely of the schematic forms
    \[
        \nabla^aK^b*\nabla^c(\phi,h),\qquad
        a+c+2b=n,\qquad c\leq n-1,
    \]
    and
    \[
        \nabla^aK^b*\nabla^c(f,l),\qquad
        a+c+2b=n-1.
    \]
    The case $b=0$ in the second family gives the principal source derivatives already displayed in \eqref{Elliptic_Estimate}. For $b\geq1$, substitute $K=\ot K+\wt K$. The terms containing only $\ot K$ are controlled by the assumed bound and the Sobolev inequalities for $g$, whereas the terms containing at least one $\wt K$ are exactly $\mathcal E_n$. Induction {on} the number of commutations, followed by the algebraic recovery of the full covariant derivatives from their symmetrized parts as in Luk's proof, proves \eqref{Elliptic_Estimate}.
\end{proof}
{Recasting the proposition in} the $\LS$ norm, we obtain the corollary below.
\begin{corollary}
    Assume that 
    \[\sup_{u,v}\lnm (-u)^{-2}g_{AB}(u,v)-{g^{\SS}}_{AB}, (-u)^{2}\ot{K}-1\rnm_{H^{n-1}(\SS)}\ll 1.\]
    Consider the Hodge system
    \[\dv\phi=f,\cl\phi=l,\tr\phi=h,\]
    with signatures $s(f)=s(l)=s(h)-1=s(\phi)-1$.
    Then 
    \begin{equation*}
        \|\nabla^n\phi\|_{\LS(u,v;a)}
        \lesssim 
        \sum_{j=0}^{n-1}\|\nabla^j(f,l,h,\phi,\wt{K})\|_{\LS(u,v;a)}.
    \end{equation*}
\end{corollary}

\subsection{Non-top-order estimates}
We first estimate the top-order derivative of $\wt{b}$, which plays an important role in $\Omega\nabla_3\wt{\psi}$-type equations.
\begin{proposition}
    For the metric components $g,b$, we have
    \begin{equation*}
        \sum_{i\leq 5}\lnm \nabla^i(\wt{b},\wt{g})\rnm^2_{\LS(u,v;N)}\lesssim B\left(\frac{v}{-u}\right)^{5/2}.
    \end{equation*}
\end{proposition}
\begin{proof}
    We have the equations
    \begin{equation*}
        \Omega\Lie_4\nabla^j\wt{b}\sim\nabla^j\left(\Omega\chi\wt{b}\right)+\nabla^j\left(\wt{\Omega^2\etab}-\wt{\Omega^2\eta}\right),\ \Omega\nabla_4\nabla^j\wt{g} \sim \nabla^j\left(\Omega\chi\wt{g}\right)+\nabla^j(\Omega\chi).
    \end{equation*}
    The transport estimate yields
    \begin{equation*}
     \begin{aligned}
           &\sum_{j\leq 5}\lnm \nabla^j(\wt{b},\wt{g})\rnm^2_{\LS(u,v;N)}\lesssim \frac{v}{-u}\sum_{i\leq 5}\int_0^v\frac{1}{-u}\lnm \nabla^i\left(\wt{\Omega\chi},\wt{\Omega\eta},\wt{\Omega\etab}\right)\rnm^2_{\LS(u^\prime,v^\prime;N)}\\
           &\qquad\qquad+\frac{v}{-u}\sum_{i\leq 5}\int_0^v\frac{1}{(-u)^{1+C\delta}}\lnm \nabla^i(\wt{b},\wt{g})\rnm^2_{\LS(u,v^\prime;N)}
           \lesssim B\left(\frac{v}{-u}\right)^{5/2}.
     \end{aligned}
    \end{equation*}
\end{proof}
By the Theorema Egregium, we can write $K\sim g^{-2}\partial^2(g^2)$, which leads to the estimate for $\wt{K}$.
\begin{corollary}
    We have the following non-top-order estimates for the curvature components:
    \begin{equation}\label{Estimate_non_top_order_curvature}
        \sum_{i\leq 3}\lnm \nabla^{i}\wt{K}\rnm^2_{\LS(u,v;N)}\lesssim B\frac{v^{5/2}}{(-u)^{5/2}}\ll \frac{v^2}{(-u)^2}.
    \end{equation}
\end{corollary}

We then {establish the non-top-order estimates for the connection components and the scalar field derivatives}.
\begin{proposition}
    We have the following estimates for the Ricci coefficients and $D\phi$:
    \begin{equation}\label{Estimate_non_top_order_Ricci_coefficients}
        \begin{aligned}
            &\sum_{i\leq 4}\lnm \nabla^i\left(\wt{\Omega\chib},\wt{\Omega\tr\chi},\wt{\eta},\wt{\Omega\omegab},\wt{\nabla\phi},\wt{\Omega e_3\phi},\wt{\nabla\phi}\right)\rnm^2_{\LS(u,v;N)}\lesssim B\frac{v^{5/2}}{(-u)^{5/2+C\delta}}\ll\frac{v^2}{(-u)^2},\\
            &\sum_{i\leq 4}\lnm \nabla^i\left(\wt{\Omega\chic},\wt{\Omega\omega},\wt{\Omega e_4\phi}\right)\rnm^2_{\LS(u,v;N)}\lesssim B\frac{v^{3/2}}{(-u)^{3/2+C\delta}}{{,}}\\
            &\sum_{i\leq 4}\lnm \nabla^i\wt{\etab}\rnm^2_{\LS(u,v;N)}\lesssim  B\frac{v^{5/2}}{(-u)^{5/2+C\delta}}+\lnm \nabla^4\wt{\dv(\Omega\chibc)}\rnm^2_{\LHb(u,v;N)}.
        \end{aligned}
    \end{equation}
\end{proposition}
\begin{proof}

Using Lemma {\ref{Transport_Estimate_LS}}, we can estimate the quantities governed by $\nabla_4$ equations. First, for $\Omega\nabla_4\nabla^i(\Omega\chih)$, we note that the top-order term is $\nabla^{i}(\Omega^2\nabla\eta)$. We then compute
\begin{equation*}
    \begin{aligned}
        &\sum_{i\leq 4}\lnm \nabla^i\wt{\Omega\chibh}\rnm^2_{\LS(u,v;N)}
        \lesssim  \frac{v}{(-u)^{1+C\delta}}\sum_{i\leq 4}\lnm \nabla^i\left(\wt{\Omega\chib},\wt{\Omega\chi},\wt{\Omega\etab},\wt{\Omega\nabla\phi}\right)\rnm^2_{\LH(u,v;N)}\\
        &\qquad\qquad\qquad+\frac{v}{-u}\lnm \nabla^5\wt{\Omega^2\nabla\etab}\rnm^2_{\LH(u,v;N)}+\sum_{i\leq 4}\frac{v}{-u}\lnm \wt{\Omega^2\widehat{\Rc}}\rnm^2_{\LH(u,v;N)}
        \lesssim  B\frac{v^{5/2}}{(-u)^{5/2+C\delta}}.
    \end{aligned}
\end{equation*}
We will frequently use $\frac{v^{1/2}}{(-u)^{1/2+C\delta}}\ll 1$.
For $\Omega\tr\chi$, the estimate follows from its $\nabla_4$ equation
\begin{equation*}
    \begin{aligned}
        \sum_{i\leq 4}\lnm \nabla^i\wt{\Omega\tr\chi}\rnm^2_{\LS(u,v;N)}\lesssim &  \frac{v}{(-u)^{1+C\delta}}\sum_{i\leq 4}\lnm \nabla^i\left(\wt{\Omega\omega},\wt{\Omega\chi},\wt{\Omega\nabla\phi}\right)\rnm^2_{\LH(u,v;N)}
        \lesssim  B\frac{v^{5/2}}{(-u)^{5/2+C\delta}}.
    \end{aligned}
\end{equation*}
The top-order term in $\Omega\nabla_4\nabla^i\eta$ is $\nabla^i\dv(\Omega\chic)$. We then have
\begin{equation*}
    \begin{aligned}
        &\sum_{i\leq 4}\lnm \nabla^i\wt{\eta}\rnm^2_{\LS(u,v;N)}\lesssim   \frac{v}{(-u)^{1+C\delta}}\sum_{i\leq 4}\lnm \nabla^i\left(\wt{\Omega\chi},\wt{\etab},\wt{\nabla\phi},\wt{\Omega e_4\phi}\right)\rnm^2_{\LH(u,v;N)}\\
        &\qquad\qquad\qquad\qquad+\frac{v}{-u}\sum_{i\leq 4}\lnm \nabla^i\left(\wt{\dv(\Omega\chic)},\wt{\Omega\Rc}_{4A}\right)\rnm^2_{\LH(u,v;N)} 
        \lesssim  B\frac{v^{5/2}}{(-u)^{5/2+C\delta}}.
    \end{aligned}
\end{equation*}
For $\Omega\omegab,\nabla\phi,\Omega e_3\phi$, we compute
\begin{equation*}
    \begin{aligned}
        &\sum_{i\leq 4}\lnm \nabla^i\wt{\Omega\omegab}\rnm^2_{\LS(u,v;N)}
        \lesssim  \frac{v}{(-u)^{1+C\delta}}\sum_{i\leq 4}\lnm \nabla^i\left(\wt{\Omega\chib},\wt{\Omega\chi},\wt{\Omega\eta},\wt{\Omega\etab},\wt{\Omega e_3\phi},\wt{\Omega e_4\phi},\wt{\Omega\nabla\phi}\right)\rnm^2_{\LH(u,v;N)}\\
        &\qquad\qquad\qquad\qquad+\sum_{i\leq 4}\frac{v}{-u}\lnm \nabla^i\left(\wt{\Omega^2K},\wt{\Omega^2{\Rc}}_{34},\wt{\Omega^2 R}\right)\rnm^2_{\LH(u,v;N)} 
        \lesssim B\frac{v^{5/2}}{(-u)^{5/2+C\delta }},
    \end{aligned}
\end{equation*}

\begin{equation*}
    \begin{aligned}
        \sum_{i\leq 4}\lnm \nabla^i\wt{\nabla\phi}\rnm^2_{\LS(u,v;N)}
        \lesssim & \frac{v}{(-u)^{1+C\delta}}\sum_{i\leq 4}\lnm \nabla^i\left(\wt{\Omega\chi},\wt{\Omega e_4\phi}\right)\rnm^2_{\LH(u,v;N)} +\frac{v}{-u}\lnm \nabla^5\wt{\Omega e_4\phi}\rnm^2_{\LH(u,v;N)}\\
        \lesssim & B\frac{v^{5/2}}{(-u)^{5/2+C\delta}},
    \end{aligned}
\end{equation*}
\begin{equation*}
    \begin{aligned}
        & \sum_{i\leq 4}\lnm \nabla^i\wt{\Omega e_3\phi}\rnm^2_{\LS(u,v;N)}
        \lesssim \frac{v}{(-u)^{1+C\delta}}\sum_{i\leq 4}\lnm \nabla^i\left(\wt{\Omega\chi},\wt{\Omega\chib},\wt{\Omega e_4\phi},\wt{\Omega\eta},\wt{\Omega\nabla\phi}\right)\rnm^2_{\LH(u,v;N)} \\
        &\qquad\qquad\qquad\qquad+\frac{v}{-u}\lnm \nabla^5\wt{\Omega^2 \nabla\phi}\rnm^2_{\LH(u,v;N)}+\sum_{i\leq 4}\frac{v}{-u}\lnm \nabla^i\wt{\Omega^2\square\phi}\rnm^2_{\LH(u,v;N)} 
        \lesssim   B\frac{v^{5/2}}{(-u)^{5/2+C\delta }}.
    \end{aligned}
\end{equation*}
We now estimate the quantities governed by $\nabla_3$ equations. 
Since none of $\Omega\chih,\Omega\omega,\Omega e_4\phi$ appears on the right-hand sides of $\Omega\nabla_3(\Omega\tr\chib)$ and $\Omega\nabla_3\etab$, we obtain
\begin{equation*}
    \begin{aligned}
        & \sum_{i\leq 4}\lnm \nabla^i\wt{\Omega\tr\chib}\rnm^2_{\LS(u,v;N)}
        \lesssim \sum_{i\leq 4}\lnm \nabla^i\wt{\Omega\tr\chib}\rnm^2_{\LS(-1,v;N)}+\sum_{i\leq 4}\lnm \nabla^i\left(\wt{b},\wt{\Omega\omegab},\wt{\Omega\chibc},\wt{\Omega e_3\phi}\right)\rnm^2_{\LHb(u,v;N)}\\
        &\qquad\qquad\qquad+\sum_{i\leq 4}\lnm \nabla^i\wt{\Omega^2\Rc}_{33}\rnm^2_{\LHb(u,v;N)}+\lnm \nabla^5\wt{b}\rnm^2_{\LHb(u,v;N)}\lesssim B\frac{v^{5/2}}{(-u)^{5/2+C\delta}},
    \end{aligned}
\end{equation*}
\begin{equation*}
    \begin{aligned}
        &\sum_{i\leq 4}\lnm \nabla^i\wt{\etab}\rnm^2_{\LS(u,v;N)}
        \lesssim  \sum_{i\leq 4}\lnm \nabla^i\wt{\etab}\rnm^2_{\LS(-1,v;N)} +\sum_{i\leq 4}\lnm \nabla^i\left(\wt{b},\wt{\Omega\chib},\wt{\eta},\wt{\Omega e_3\phi},\wt{\nabla\phi}\right)\rnm^2_{\LHb(u,v;N)}\\
        &\qquad\qquad\qquad+\sum_{i\leq 4}\lnm \nabla^i\left(\wt{\dv(\Omega\chibc)},\wt{{\Rc}}_{3A}\right)\rnm^2_{\LHb(u,v;N)}+\lnm \nabla^5\wt{b}\rnm^2_{\LHb(u,v;N)} \\
        &\qquad\lesssim   B\frac{v^{5/2}}{(-u)^{5/2+C\delta}} + \lnm \nabla^4\wt{\dv(\Omega\chibc)}\rnm^2_{\LHb(u,v;N)}.
    \end{aligned}
\end{equation*}
This completes the estimates for $\psi_1$. We next estimate $\Omega e_4\phi,\Omega\omega,\Omega\chih$. Because of terms containing $\ot{\psi}_2$, such as $\wt{\Omega\chib}\ot\psi_2$ and $\nabla\wt{b}\ot{\psi_2}$, the non-top-order terms acquire an additional factor of $\frac{1}{(-u)^{C\delta}}$. The estimates therefore differ slightly from those in Lemma \ref{Transport_Estimate_LS}. We compute
\begin{equation*}
    \begin{aligned}
        &\sum_{i\leq 4}\lnm \nabla^i\wt{\Omega\chih}\rnm^2_{\LS(u,v;N)} 
        \lesssim  \sum_{i\leq 4}\lnm \nabla^i\wt{\Omega\chih}\rnm^2_{\LS(-1,v;N)} +\sum_{i\leq 4}\lnm \nabla^i\wt{\Omega^2\widehat{\Rc}}\rnm^2_{\LHb(u,v;N)}+\lnm \nabla^5\wt{\Omega^2\nabla\eta}\rnm^2_{\LHb(u,v;N)}\\
        &\qquad\qquad+\frac{1}{(-u)^{C\delta}}\lnm \nabla^5\wt{b}\rnm^2_{\LHb(u,v;N)}+\sum_{i\leq 4}\frac{1}{(-u)^{C\delta}}\lnm \nabla^i\left(\wt{b},\wt{\Omega\chib},\wt{\Omega \tr\chi},\wt{\Omega\eta},\wt{\Omega\nabla\phi}\right)\rnm^2_{\LHb(u,v;N)}\\
        &\qquad\lesssim  B\frac{v^{3/2}}{(-u)^{3/2+C\delta}},
    \end{aligned}
\end{equation*}

\begin{equation*}
    \begin{aligned}
        &\sum_{i\leq 4}\lnm \nabla^i\wt{\Omega\omega}\rnm^2_{\LS(u,v;N)}
        \lesssim \sum_{i\leq 4}\lnm \nabla^i\wt{\Omega\omega}\rnm^2_{\LS(-1,v;N)} +\sum_{i\leq 4}\lnm \nabla^i\left(\wt{\Omega^2K},\wt{\Omega^2{\Rc}}_{34},\wt{\Omega^2 R}\right)\rnm^2_{\LHb(u,v;N)}\\
        &\qquad\qquad+\sum_{i\leq 4}\frac{1}{(-u)^{C\delta}}\lnm \nabla^i\left(\wt{b},\wt{\Omega\chib},\wt{\Omega\chi},\wt{\Omega\eta},\wt{\Omega\etab},\wt{\Omega e_3\phi},\wt{\Omega e_4\phi},\wt{\Omega\nabla\phi}\right),\nabla^5\wt{b}\rnm^2_{\LHb(u,v;N)}\\
        &\qquad\lesssim  B\frac{v^{3/2}}{(-u)^{3/2+C\delta}},
    \end{aligned}
\end{equation*}
and 
\begin{equation*}
    \begin{aligned}
        &\sum_{i\leq 4}\lnm \nabla^i\wt{\Omega e_4\phi}\rnm^2_{\LS(u,v;N)}\lesssim \sum_{i\leq 4}\lnm \nabla^i\wt{\Omega e_4\phi}\rnm^2_{\LS(-1,v;N)}+\sum_{i\leq 4}\lnm \nabla^i\wt{\Omega^2\square\phi}\rnm^2_{\LHb(u,v;N)} \\
        &\qquad+\lnm \nabla^5\wt{\Omega^2\nabla\phi}\rnm^2_{\LHb(u,v;N)}+\sum_{i\leq 4}\frac{1}{(-u)^{C\delta}}\lnm \nabla^i\left(\wt{b},\wt{\Omega\chib},\wt{\Omega \tr\chi},\wt{\Omega\eta},\wt{\Omega\nabla\phi}\right)\rnm^2_{\LHb(u,v;N)}\\
        &\qquad+\frac{1}{(-u)^{C\delta}}\lnm \nabla^5\wt{b}\rnm^2_{\LHb(u,v;N)} 
        \lesssim   B\frac{v^{3/2}}{(-u)^{3/2+C\delta}}.
    \end{aligned}
\end{equation*}
\end{proof}

\subsection{Energy estimates for top-order terms}
In this subsection, we close the bootstrap {argument} based on \eqref{eq:Bootstrap_Assumption}. We first list the Bianchi pairs of interest:
\[(\dv(\Omega\chic),(\Omega d\eta, \Omega K)),\ ((\Omega d\etab,\Omega K), \dv(\Omega\chibc)),\ (\nabla(\Omega e_4\phi),\Omega\nabla^2\phi),\ (\Omega\nabla^2\phi,\nabla(\Omega e_3\phi)).\]
Because their equations satisfy the model equations in Lemma \ref{Energy_Estimate_LH_LHb}, we obtain
\begin{equation*}
    \begin{aligned}
        &\lnm\nabla^4\wt{\Psi}_1 \rnm_{\LH(u,v;N)}^2+\lnm\nabla^4\wt{\Psi}_2 \rnm_{\LHb(u,v;N)}^2+\lnm\nabla^4\wt{\Psi}_2 \rnm_{\LR(u,v;N)}^2\\
        \lesssim & \int_0^v\int_{-1}^u \frac{1}{(-u^\prime)^{2+C\delta}}\lnm \nabla^4\wt{\Psi}\rnm^2_{\LS(u^\prime,v;N)}+\frac{v^2}{(-u)^2}\\
        &\qquad+\frac{1}{(-u^\prime)^2}\lnm \nabla^5\left(\wt{\Omega\chic},\wt{\Omega e_4\phi}\right)\rnm^2_{\LS(u^\prime,v;N)}-\frac{N}{(-u^\prime)^2}\lnm \nabla^4\wt{\Psi_1}\rnm^2_{\LS(u^\prime,v;N)}.
    \end{aligned}
\end{equation*}
The elliptic estimate gives
\begin{equation*}
    \lnm \nabla^5\wt{\Omega\chic}\rnm^2_{\LS(u,v;N)}\lesssim \lnm \nabla^5\wt{\Omega\tr\chi},\nabla^4\wt{\dv\Omega\chic},\nabla^4 \wt{K}\rnm^2_{\LS(u,v;N)}+\frac{v^2}{(-u)^2}.
\end{equation*}
The estimate above can be rewritten as
\begin{equation*}
    \begin{aligned}
        &\lnm\nabla^4\wt{\Psi}_1 \rnm_{\LH(u,v;N)}^2+\lnm\nabla^4\wt{\Psi}_2 \rnm_{\LHb(u,v;N)}^2+\lnm\nabla^4\wt{\Psi}_2 \rnm_{\LR(u,v;N)}^2\\
        \lesssim & \int_0^v\int_{-1}^u \frac{1}{(-u^\prime)^{2+C\delta}}\lnm \nabla^4\wt{\Psi}\rnm^2_{\LS(u^\prime,v;N)}+\frac{v^2}{(-u)^2}\\
        &\qquad+\frac{1}{(-u^\prime)^2}\lnm \nabla^4\left(\wt{\dv\Omega\chic},\wt{\nabla(\Omega e_4\phi)}\right)\rnm^2_{\LS(u^\prime,v;N)}-\frac{N}{(-u^\prime)^2}\lnm \nabla^4\wt{\Psi_1}\rnm^2_{\LS(u^\prime,v;N)}\\
        \lesssim & \int_0^v\int_{-1}^u \frac{1}{(-u^\prime)^{2+C\delta}}\lnm \nabla^4\wt{\Psi}\rnm^2_{\LS(u^\prime,v;N)}+\frac{v^2}{(-u)^2}.
    \end{aligned}
\end{equation*}
Since $\Psi$ above denotes neither $\dv(\Omega\chic)$ nor $\nabla( \Omega e_4\phi)$, the bootstrap assumption gives
$$\lnm\nabla^4\wt\Psi\rnm_{\LHb(u,v;N)}\lesssim B\frac{v^{3/2}}{(-u)^{3/2}}.$$
The first term on the right-hand side is controlled by
\begin{equation*}
    \int_0^v\int_{-1}^u \frac{1}{(-u^\prime)^{2+C\delta}}\lnm \nabla^4\wt{\Psi}\rnm^2_{\LS(u^\prime,v;N)}\lesssim B\frac{v^{5/2}}{(-u)^{5/2+C\delta}}.
\end{equation*}
We therefore arrive at
\begin{equation}\label{eq:top-order-est-1}
    \lnm\nabla^4\wt{\Psi}_1 \rnm_{\LH(u,v;N)}^2+\lnm\nabla^4\wt{\Psi}_2 \rnm_{\LHb(u,v;N)}^2+\lnm\nabla^4\wt{\Psi}_2 \rnm_{\LR(u,v;N)}^2\lesssim\frac{v^2}{(-u)^2}
\end{equation}
for the Bianchi pairs
\[(\dv(\Omega\chic),(\Omega d\eta, \Omega K)),\ ((\Omega d\etab,\Omega K), \dv(\Omega\chibc)),\ (\nabla(\Omega e_4\phi),\Omega\nabla^2\phi),\ (\Omega\nabla^2\phi,\nabla(\Omega e_3\phi)).\]
{As a direct consequence, we obtain the remaining} non-top-order estimate for $\etab$ in \eqref{Estimate_non_top_order_Ricci_coefficients},
\begin{equation*}
    \sum_{i\leq 4}\lnm \nabla^i\wt\etab\rnm^2_{\LS^2(u,v;N)}\lesssim \left(\frac{v}{-u}\right)^2,
\end{equation*}
and 
\begin{equation*}
    \sum_{i\leq 4}\lnm \nabla^{i+1}\wt{\log\Omega}\rnm^2_{\LS^2(u,v;N)}\lesssim \sum_{i\leq 4}\lnm \nabla^i\wt\eta,\nabla^i\wt\etab\rnm^2_{\LS^2(u,v;N)}\lesssim\left(\frac{v}{-u}\right)^2.
\end{equation*}
Recall the $\nabla_4\tr\chi$ and $\nabla_4(\dv\etab-K)$ equations. We have the schematic equation
\begin{equation*}
    \Lie_v\wt{\Psi_2}\sim \wt{\Psi}+\wt{\psi\Psi}+\wt{\psi\nabla(\Omega\chic)}+\wt{\psi\nabla(\Omega e_4\phi)}+\wt{\Omega^2\Rc}
\end{equation*}
for $\Psi_2\in\{\Omega^2\nabla(\Omega^{-1}\tr\chi), \Omega \dv\eta-\Omega K\}$. We do not use $\nabla(\Omega\tr\chi)$, in order to avoid top-order derivatives of $\omega$. Therefore, the transport estimate gives
\begin{equation*}
    \sum_{i\leq 4}\lnm \nabla^i\wt{\Psi_2}\rnm^2_{\LS(u,v;N)}\lesssim B\frac{v^{5/2}}{(-u)^{5/2+C\delta}}\ll \frac{v^2}{(-u)^2}.
\end{equation*}

For $\tr\chib$ and $\dv\eta-K$, we need to treat their $\nabla_3$ equations carefully. Recall that for ESE, 
\begin{equation*}
    \begin{aligned}
        \Omega\nabla_3 (\Omega^2\nabla(\Omega^{-1}\tr\chib))\sim &\Omega\omegab \Omega^2\nabla(\Omega^{-1}\tr\chib)+(\eta+\etab)\left|\Omega\chibc\right|^2+(\eta+\etab)\left|\Omega e_3\phi\right|^2\\
        &+\Omega\chib \Omega^2\nabla(\Omega^{-1}\tr\chib)+\nabla\left(\Omega\chibc\right)\cdot\Omega\chibc+\nabla\left(\Omega e_3\phi\right)\cdot\Omega e_3\phi,
    \end{aligned}
\end{equation*}
\begin{equation*}
    \begin{aligned}
        \Omega\nabla_3(\Omega\dv\eta-\Omega K)\sim & (\Omega\omegab+\Omega\chib)(\Omega\dv\eta-\Omega K)+\Omega\nabla(\eta-\etab)\Omega\chibh+\Omega(\eta-\etab)\nabla(\Omega\chibh)\\
        &+\Omega\nabla^2\phi\Omega e_3\phi+\Omega\nabla\phi\nabla(\Omega e_3\phi).
    \end{aligned}
\end{equation*}
Because 
$$\nabla^4\left(\wt{[\Omega^2\nabla(\Omega^{-1}\tr\chib)]}-\wt{[\nabla(\Omega\tr\chib)]}\right)\sim \nabla^4\left(\wt{\eta\Omega\tr\chib}+\wt{\eta\Omega\tr\chib}\right)$$
contains no top-order terms, the $\LS(u,v;N)$ norms of $\nabla^4\wt{[\Omega^2\nabla(\Omega^{-1}\tr\chib)]}$ and $\nabla^4\wt{[\nabla(\Omega\tr\chib)]}$ differ at most $(v/u)^2.$
We can use the transport estimate in Lemma \ref{Transport_Estimate_LS} and the non-top-order results in \eqref{Estimate_non_top_order_Ricci_coefficients} to obtain
\begin{equation}\label{eq:top-trchib-LS}
    \begin{aligned}
        \lnm \nabla^4\wt{\left(\Omega^2\nabla(\Omega^{-1}\tr\chib)\right)}\rnm^2_{\LS(u,v;N)}\lesssim & \lnm \nabla^5\left(\wt{\Omega\chib},\wt{\Omega e_3\phi}\right)\rnm^2_{\LHb(u,v;N)}+\frac{v^2}{(-u)^2}\\
        &- N \lnm \nabla^4\wt{\left[\nabla(\Omega\tr\chib)\right]}\rnm^2_{\LHb(u,v;N)},
    \end{aligned}
\end{equation}
and 
\begin{equation}\label{eq:top-eta-K-LS}
    \begin{aligned}
        \lnm \nabla^4\left(\wt{\Omega\dv\eta}-\wt{\Omega K}\right)\rnm^2_{\LS(u,v;N)}\lesssim & \lnm \nabla^5\left(\wt{\Omega\chib},\wt{\Omega e_3\phi},\wt{\Omega\nabla\phi},\wt{\Omega\etab},\wt{\Omega \eta}\right)\rnm^2_{\LHb(u,v;N)}\\
        &+\frac{v^2}{(-u)^2}-N \lnm \nabla^4\left(\wt{\Omega\dv\eta}-\wt{\Omega K}\right)\rnm^2_{\LHb(u,v;N)}.
    \end{aligned}
\end{equation}
Combining {these estimates} with \eqref{eq:top-order-est-1} and the elliptic estimate
\begin{equation*}
    \lnm \nabla^5\wt{\Omega\chibc}\rnm^2_{\LS(u,v;N)}\lesssim \lnm \nabla^5\wt{\Omega\tr\chib},\nabla^4\wt{\dv\Omega\chibc},\nabla^4 \wt{K}\rnm^2_{\LS(u,v;N)}+\frac{v^2}{(-u)^2},
\end{equation*}
we find 
\begin{equation}\label{eq:top-order-est-2}
    \begin{aligned}
        \lnm \nabla^5\wt{\left(\Omega\tr\chib\right)}\rnm^2_{\LS(u,v;N)}+\lnm \nabla^4\left(\wt{\Omega\dv\eta}-\wt{\Omega K}\right)\rnm^2_{\LS(u,v;N)}\lesssim \frac{v^2}{(-u)^2},
    \end{aligned}
\end{equation}
and thus
\begin{equation}\label{eq:top-order-est-3}
    \lnm \nabla^5\wt{\left(\Omega\chib\right)}\rnm^2_{\LHb(u,v;N)}\lesssim\frac{v^2}{(-u)^2},\ \ \lnm \nabla^4\wt{\etab}\rnm^2_{\LS(u,v;N)}\lesssim\frac{v^2}{(-u)^2}.
\end{equation}
By the elliptic estimate, we find that 
\begin{equation*}
    \sum_{i\leq 5}\lnm \nabla^i\wt\eta\rnm_{\LS(u,v;N)}\lesssim \sum_{i\leq 4}\lnm \nabla^i\left(\wt\eta,d\wt{\eta},\dv\wt{\eta},\wt{K}\right)\rnm_{\LS(u,v;N)},
\end{equation*}
which implies 
 \begin{equation*}
    \begin{aligned}
        &\lnm \Omega\nabla^5\wt\eta\rnm_{\LH(u,v;N)}\lesssim \sum_{i\leq 4,j\leq 5}\lnm \nabla^i\left(\Omega\wt\eta,\wt{\Omega d\etab},\wt{\Omega\dv\eta}-\wt{\Omega K},\wt{\Omega K}\right),\nabla^j\left(\wt{g},\wt{\log\Omega}\right)\rnm_{\LH(u,v;N)},\\
        &\lnm \Omega\nabla^5\wt\eta\rnm_{\LHb(u,v;N)}\lesssim \sum_{i\leq 4,j\leq 5}\lnm \nabla^i\left(\Omega\wt\eta,\wt{\Omega d\eta},\wt{\Omega\dv\eta}-\wt{\Omega K},\wt{\Omega K}\right),\nabla^j\left(\wt{g},\wt{\log\Omega}\right)\rnm_{\LHb(u,v;N)}.
    \end{aligned}
\end{equation*}
{We remark that} $\tr d\etab=\cl\etab=-\cl\eta=-\tr d\eta$. Similarly, for $\etab$ we have
 \begin{equation*}
    \begin{aligned}
        &\lnm \Omega\nabla^5\wt\etab\rnm_{\LH(u,v;N)}\lesssim \sum_{i\leq 4,j\leq 5}\lnm \nabla^i\left(\Omega\wt\etab,\wt{\Omega d\etab},\wt{\Omega\dv\etab}-\wt{\Omega K},\wt{\Omega K}\right),\nabla^j\left(\wt{g},\wt{\log\Omega}\right)\rnm_{\LH(u,v;N)},\\
        &\lnm \Omega\nabla^5\wt\etab\rnm_{\LHb(u,v;N)}\lesssim \sum_{i\leq 4,j\leq 5}\lnm \nabla^i\left(\Omega\wt\etab,\wt{\Omega d\eta},\wt{\Omega\dv\etab}-\wt{\Omega K},\wt{\Omega K}\right),\nabla^j\left(\wt{g},\wt{\log\Omega}\right)\rnm_{\LHb(u,v;N)}.
    \end{aligned}
\end{equation*}
Employing {the} estimates \eqref{eq:top-order-est-1}, \eqref{eq:top-order-est-2}, \eqref{eq:top-order-est-3}, and \eqref{Estimate_non_top_order_Ricci_coefficients}, we obtain
\begin{equation*}
    \lnm \Omega\nabla^5(\wt\eta,\wt\etab)\rnm^2_{\LH(u,v;N)}+\lnm \Omega\nabla^5(\wt\eta,\wt\etab)\rnm^2_{\LHb(u,v;N)}\lesssim \frac{v^2}{(-u)^2}.
\end{equation*}
{This completes the a priori estimates.}
 \section{Trapped Surface Formation and Marginally Outer Trapped Surfaces}\label{Section_Trapped_surface_formation} 

\subsection{Isotropic trapped surface}

In this section, we demonstrate the formation of trapped surfaces near the singularity arising from isotropic perturbations with growth rate $v^{\frac{\k}{2}-o(1)}$ along $H_{-1}$.
We recall the geodesic coordinates $(u,v,\theta)\rightarrow (S(u,v,\theta),V,\Theta)$ defined by
\begin{equation*}
    (u,v,\theta)=\exp_{(-1,V,\Theta)}\left(S\cdot (\Omega^{-1}e_3)|_{(-1,V,\Theta)}\right),
\end{equation*}
where $(-1,V,\Theta)$ is the point on $H_{-1}$ expressed in double-null coordinates. Furthermore, we define
$$U(S,V,\Theta)=\int_0^S\Omega^{-2}(S^\prime,V,\Theta)dS^\prime.$$
We then consider the coordinates $(U,V,\Theta)$. We have $u=U$ and $\partial_U=\left(\frac{\partial U}{\partial S}\right)^{-1}\partial_S=\Omega e_3$.
We emphasize that the Ricci coefficients used below are defined in the coordinates $(u,v,\theta)$ but regarded as tensor-valued functions of $(U,V,\Theta)$.
We also note that the coordinate transformation from $(u,v,\theta)$ to $(U,V,\Theta)$ can be written as $(U,V,\Theta)=(u,v,\Theta(u,v,\theta))$. The difference $\psi-\psi_0$ is still defined in the $(u,v,\theta)$ coordinates as $\psi(u,v,\theta)-\psi(u,0,\theta)$.

\begin{proposition}
    Suppose that {the following assumption holds along $H_{-1}$ for every $\theta$}:
\begin{equation*}
    \left|\left|\Omega\chih\right|^2-\left|(\Omega\chih)_0\right|^2\right|+\left|\left|\Omega e_4\phi\right|^2-\left|(\Omega e_4\phi)_0\right|^2\right|>\frac{1}{C}v^{2\lambda},\ \left|\Omega\chih-(\Omega\chih)_0\right|+\left|\Omega e_4\phi-(\Omega e_4\phi)_0\right|\leq Cv^\lambda.
\end{equation*}
If $0<\lambda<\frac{\k(2N-3)}{2(2N+3\k)}$, then along $v_0(u)=\epsilon_1(-u)^{1+\frac{3\k}{2N}}$ with $\epsilon_1$ sufficiently small and $(-u)$ sufficiently small, we have $\tr\chi(u,v_0(u),\theta)<0$ for each $\theta$.
\end{proposition}
From the equations for $\Omega\nabla_3(\Omega\chih)$ and $\Omega\nabla_3(\Omega e_4\phi)$, we obtain {the following estimate} for {either} $\psi_2=\Omega\chih$ {or} $\psi_2=\Omega e_4\phi$:
\begin{equation*}
    \left|\left(\nabla_{\partial_U}-\frac{1+o(1)}{-U}\right)\left|\psi_2-(\psi_2)_0\right|\right|\lesssim \frac{V}{(-U)^{3-\k}}+\frac{1}{(-U)^{1-\k}}\left(\Omega^{-1}\tr\chi-(\Omega^{-1}\tr\chi)_0\right).
\end{equation*}
Here $\left|o(1)\right|$ is bounded by $\frac{\k}{2N}$.
From {this estimate}, we can extract both lower and upper bounds for each $(U,V,\Theta)$:
\begin{equation*}
    \begin{aligned}
        &(-U)^{1-\frac{\k}{2N}}\left|\psi_2-(\psi_2)_0\right|\\ &\qquad\qquad \geq \frac{1}{C}V^\lambda -C\frac{V}{(-U)^{1-\k+\k/2N}}-\int_{-1}^U\frac{1}{(-U^\prime)^{-\k+\k/2N}}\left|\Omega^{-1}\tr\chi-(\Omega^{-1}\tr\chi)_0\right|,
    \end{aligned}
\end{equation*}
\begin{equation*}
    \begin{aligned}
        &(-U)^{1+\frac{\k}{2N}}\left|\psi_2-(\psi_2)_0\right| \\ &\qquad\qquad \leq CV^\lambda +C\frac{V}{(-U)^{1-\k-\k/2N}}+\int_{-1}^U\frac{1}{(-U^\prime)^{-\k-\k/2N}}\left|\Omega^{-1}\tr\chi-(\Omega^{-1}\tr\chi)_0\right|.
    \end{aligned}
\end{equation*}

Using the equation
$\left(\Omega\nabla_4+\frac{1}{2}\Omega\tr\chi\right)(\Omega^{-1}\tr\chi)=-\Omega^{-2}\left|\Omega\chih\right|^2-\Omega^{-2}(\Omega e_4\phi)^2$ and the estimate $|\Omega\chih|+|\Omega e_4\phi|\lesssim \frac{1}{(-u)^{1+\k/2N}}$, we have 
\begin{equation*}
 \begin{aligned}
       &\left|\Omega^{-1}\tr\chi-(\Omega^{-1}\tr\chi)_0\right|^2(u,v,\theta)\lesssim  v \int_0^v \frac{1}{(-u)^{2+2\k+\k/N}}\left(\left|\Omega\chih\right|^2+\left|\Omega e_4\phi\right|^2\right)(u,v^\prime,\theta)\\
       \lesssim & \frac{v^2}{(-u)^{4+2\k+2\k/N}}\left((-u)^{2\k}+v^{2\lambda}+\frac{v^2}{(-u)^{2-2\k}}\right)\\
       &+\frac{v}{(-u)^{2+2\k+2\k/N}}\int_0^v\left(\int_{-1}^u\frac{1}{(-u^\prime)^{-\k-\k/2N}}\left|\Omega^{-1}\tr\chi-(\Omega^{-1}\tr\chi)_0\right|(U(u),V(v^\prime),\Theta(u,v^\prime,\theta))\right)^2\\
       \lesssim & \frac{v^2}{(-u)^{4+2\k/N}}\left(1+\frac{v^{2\lambda}}{(-u)^{2\k}}+\frac{v^2}{(-u)^2}\right)\\
       &+\frac{v^2}{(-u)^{2+2\k+3\k/N}}\sup_{v^\prime\leq v}\left(\int_{-1}^u\frac{1}{(-u^\prime)^{-\k}}\left|\Omega^{-1}\tr\chi-(\Omega^{-1}\tr\chi)_0\right|(U(u^\prime),V(v^\prime),\Theta(u^\prime,v^\prime,\theta))\right)^2.
 \end{aligned}
\end{equation*}
If we denote 
\begin{equation*}
    \begin{aligned}
        f(u,v)=&\sup_{v^\prime\leq v}\sup_{(u,v^\prime,\theta)\in S_{u,v^\prime}}\frac{1}{(-u)^{-\k}}\left|\Omega^{-1}\tr\chi-(\Omega^{-1}\tr\chi)_0\right|(u,v,\theta)\\
        =&\sup_{v^\prime\leq v}\sup_{(U,V,\Theta)\in S_{u,v^\prime}}\frac{1}{(-U)^{-\k}}\left|\Omega^{-1}\tr\chi-(\Omega^{-1}\tr\chi)_0\right|(U,V,\Theta),
    \end{aligned}
\end{equation*}
the inequality above can be written as
\begin{equation*}
    f(u,v)\leq C \frac{v}{(-u)^{2-\k+\k/N}}\left(1+\frac{v^\lambda}{(-u)^\k}\right)+C\frac{v}{(-u)^{2+3\k/2N}}\int_{-1}^u f(u^\prime,v).
\end{equation*}
If we {impose} the following bootstrap assumption,
\begin{equation*}
    f(u)\leq D\frac{v}{(-u)^{2-\k+\k/N}}\left(1+\frac{v^\lambda}{(-u)^\k}\right),
\end{equation*}
with $D\geq C+2$,
then we have 
\begin{equation*}
    \begin{aligned}
        f(u)\leq & C \frac{v}{(-u)^{2-\k+\k/N}}\left(1+\frac{v^\lambda}{(-u)^\k}\right)+C\frac{v}{(-u)^{2+3\k/2N}}\cdot D\cdot C(\lambda,\k) \frac{v}{(-u)^{1-\k+\k/N}}\left(1+\frac{v^\lambda}{(-u)^\k}\right)\\
        \leq & \left(C+C\cdot D\cdot C(\lambda,\k)\frac{v}{(-u)^{1+3\k/2N}}\right) \frac{v}{(-u)^{2-\k+\k/N}}\left(1+\frac{v^\lambda}{(-u)^\k}\right)\\
        \leq & \left(C+1\right) \frac{v}{(-u)^{2-\k+\k/N}}\left(1+\frac{v^\lambda}{(-u)^\k}\right),
    \end{aligned}
\end{equation*}
where we {assume} $\frac{v}{(-u)^{1+3\k/2N}}\leq \epsilon_1\ll 1$.
It follows that
\begin{equation*}
    \lnm \Omega^{-1}\tr\chi-(\Omega^{-1}\tr\chi)_0\rnm_{L^\infty(S_{u,v})}\lesssim \frac{v}{(-u)^{2+\k/N}}\left(1+\frac{v^\lambda}{(-u)^\k}\right),
\end{equation*}
and for each $\theta$, 
\begin{equation*}
    \left|\psi_2-(\psi_2)_0\right|\geq \frac{1}{C}v^\lambda-C\frac{v}{(-u)^{1-\k}}-C\frac{v}{(-u)^{1-\k+\k/N}}\left(1+\frac{v^\lambda}{(-u)^\k}\right).
\end{equation*}
From the equation $ \Omega\nabla_4(\Omega^{-1}\tr\chi)=-\frac{1}{2}\Omega\tr\chi\Omega^{-1}\tr\chi-\Omega^{-2}\left|\Omega\chih\right|^2-\Omega^{-2}(\Omega e_4\phi)^2$, we obtain
\begin{equation*}
    \begin{aligned}
        \partial_v (\Omega^{-1}\tr\chi)(u,v,\theta)\leq &  \frac{C}{-u}\Omega^{-1}\tr\chi(u,v,\theta)- \frac{1}{C(-u)^\k}\left(\left|\Omega\chih\right|^2+\left|\Omega e_4\phi\right|^2\right)(u,v,\theta).
    \end{aligned}
\end{equation*}
{After integrating and evaluating} along $v_0(u)=\epsilon_1(-u)^{1+3\k/2N}$ with $\epsilon_1$ sufficiently small, we obtain
\begin{equation*}
    \begin{aligned}
        &\sup_{\theta}(-u)\Omega^{-1}\tr\chi(u,v_0(u),\theta)\\
        \leq & C-\frac{1}{C}\frac{v^{1+2\lambda}}{(-u)^{1+\k}}+C\frac{v^3}{(-u)^{3-\k}}+C\frac{v^3}{(-u)^{3-\k+2\k/N}}\left(1+\frac{v^{2\lambda}}{(-u)^{2\k}}\right)   \\
        \leq & C-\frac{1}{C}\frac{v^{1+2\lambda}}{(-u)^{1+\k}}+C\frac{v^3}{(-u)^{3-\k+2\k/N}}\frac{v^{2\lambda}}{(-u)^{2\k}}\\
        \leq & C-\frac{v^{1+2\lambda}}{(-u)^{1+\k}}\left(\frac{1}{C}-C\frac{v^2}{(-u)^{2+2\k/N}}\right)\leq  C-\frac{1}{C}\frac{v^{1+2\lambda}}{(-u)^{1+\k}}\\
        \leq & C-\frac{c(\epsilon_1)}{C}\frac{1}{(-u)^{\k-2\lambda-3\k/2N-3\lambda\k/N}}.
    \end{aligned}
\end{equation*}
For $0<\lambda<\frac{\k(2N-3)}{2(2N+3\k)}$, the exponent satisfies $\k-2\lambda-\frac{3\k}{2N}-\frac{3\lambda\k}{N}>0$, which implies $$\sup_{\theta}(-u)\Omega^{-1}\tr\chi<-1$$ for $-u$ sufficiently small. This establishes the existence of a trapped surface.

\subsection{Anisotropic trapped surface}
We continue to work in the coordinates $(U,V,\Theta)$ introduced in the previous section, with $U=u,V=v,\Theta=\Theta(u,v,\theta)$. We define the natural transport map $F_{u,v}:S_{-1,v}\to S_{u,v}$ by $F_{u,v}(x)=(U,V,\Theta)$ for $x=(-1,V,\Theta)\in S_{-1,v}$.

\begin{lemma}
    Let $\delta\gg \epsilon$. If $B_{S_{-1,v}}(x,r)\subset A$, then for all $u$ within the existence region, we have
    \begin{equation*}
        B_{S_{u,v}}\left(F_{u,v}(x),(-u)^{1+\delta}r\right)\subset A_{u,v}.
    \end{equation*}
    Conversely, if $B_{S_{u,v}}\left(F_{u,v}(x),(-u)^{1-\delta}r\right)\subset A_{u,v}$, then $B_{S_{-1,v}}(x,r)\subset A$.
\end{lemma}
\begin{proof}
    For two points $x,y\in S_{-1,v}$, we estimate $\partial_u d_{S_{u,v}}(F_{u,v}(x), F_{u,v}(y))$. First, fix $(u,v)$. Suppose that $c_{u,v}:[0,1]\rightarrow S_{u,v}$ is a minimizing geodesic with $c_{u,v}(0)=F_{u,v}(x),c_{u,v}(1)=F_{u,v}(y)$, and $\nabla_{\partial_s}\partial_s c_{u,v}=0$. Then
    \begin{equation*}
        d_{S_{u,v}} (F_{u,v}(x), F_{u,v}(y))= \int_0^1 {\left| \partial_s c_{u,v}\right|_{g_{u,v}}}ds.
    \end{equation*}
    For $u^\prime\neq u$, we define $\Phi^{u^\prime}_{u}=F_{u^\prime,v}\circ (F_{u,v})^{-1}:S_{u,v}\rightarrow S_{u^\prime,v}$. Then 
    \begin{equation*}
        d_{S_{u^\prime,v}}(F_{u^\prime,v}(x), F_{u^\prime,v}(y))\leq \int_0^1 {\left|(\Phi^{u^\prime}_{u})_*\partial_s c_{u,v}\right|_{g_{u^\prime,v}}}ds.
    \end{equation*}
If we denote $c_{u,v,\epsilon}(s)=\Phi_u^{u+\epsilon}(c_{u,v}(s))$, then we have
\begin{equation*}
    \begin{aligned}
        &\partial_u d_{S_{u,v}}(F_{u,v}(x), F_{u,v}(y))\leq \int_0^1\partial_{\epsilon}|_{\epsilon=0}\left({\left|(\Phi^{u+\epsilon}_{u})_*\partial_s c_{u,v}\right|_{g_{u+\epsilon,v}}}\right) ds\\
        \leq &\int_0^1 \frac{\langle D_{\partial_\epsilon}{\partial_s} c_{u,v,\epsilon}(s),\partial_s c_{u,v}\rangle}{\left|\partial_s c_{u,v}\right|}ds=\int_0^1\frac{\langle D_{\partial_s} (\Omega e_3)|_{c_{u,v}(s)},\partial_s c_{u,v}\rangle}{\left|\partial_s c_{u,v}\right|} ds\\
        =& \int_0^1 \frac{1}{{\left|\partial_s c_{u,v}\right|}}\Omega\chib\left(\partial_s c_{u,v},\partial_s c_{u,v}\right)d\leq\int_0^1 \frac{1}{2}\Omega\tr\chib{\left|\partial_s c_{u,v}\right|}+\left|\Omega\chibh\right|{\left|\partial_s c_{u,v}\right|}\\
        \leq & \frac{-1+\delta}{-u}d_{S_{u,v}}(F_{u,v}(x), F_{u,v}(y)).
    \end{aligned}
\end{equation*}
Hence{,}
\begin{equation*}
    (-u)^{-1+\delta}d_{S_{u,v}}(F_{u,v}(x), F_{u,v}(y))\leq d_{S_{-1,v}}(x,y).
\end{equation*}
If we consider $\wt{c}_{u,v,\epsilon}=\Phi_{u}^{u-\epsilon}\left(c_{u,v}(s)\right)$, then 
$$d_{S_{u-\epsilon,v}}(F_{u-\epsilon,v}(x), F_{u-\epsilon,v}(y))\leq \int_0^1\left|\partial_s \wt{c}_{u,v,\epsilon}\right|ds,$$
and thus{,}
\begin{equation*}
    \partial_u d_{S_{u,v}}(F_{u,v}(x), F_{u,v}(y))\geq \int_0^1\partial_{\epsilon}|_{\epsilon=0}\left|\partial_s\wt{c}_{u,v,\epsilon}\right|ds\geq \frac{-1-\delta}{-u}d_{S_{u,v}}(F_{u,v}(x), F_{u,v}(y)).
\end{equation*}
We obtain
\begin{equation*}
    (-u)^{-1-\delta}d_{S_{u,v}}(F_{u,v}(x), F_{u,v}(y))\geq d_{S_{-1,v}}(x,y).
\end{equation*}
\end{proof}
By adapting the arguments used in the isotropic case, we establish the following result for anisotropic perturbations.
\begin{proposition}
    Suppose that there exist $\theta_0$ and $r>0$ such that, for all $v\in(0,\epsilon_1)$ and $\theta\in B_{S_{-1,v}}(\theta_0,r)$, the differences are bounded below:
    \begin{equation*}
        \left|\Omega\chih\right|^2-\left|(\Omega\chih)_0\right|^2+\left|\Omega e_4\phi\right|^2-\left|(\Omega e_4\phi)_0\right|^2 > \frac{1}{C}v^{2\lambda},
    \end{equation*}
    and for all $\theta\in\SS$,
    \begin{equation*}
        \left|\Omega\chih-(\Omega\chih)_0\right|+\left|\Omega e_4\phi-(\Omega e_4\phi)_0\right|\leq Cv^\lambda.
    \end{equation*}
    If $0<\lambda<\frac{\k(2N-3)}{2(2N+3\k)}-\delta_1$, then along $v_0(u)=\epsilon_1(-u)^{1+\frac{3\k}{2N}}$ with $\epsilon_1$ sufficiently small and $(-u)$ sufficiently small, we have $$(-u)\Omega^{-1}\tr\chi(u,v_0(u),\theta)<-\frac{1}{C(-u)^{2\delta_1}}$$ on $B_{S_{u,v_0(u)}}(F_{u,v_0(u)}(\theta_0), (-u)^{1+\delta} r).$
\end{proposition}

We write $e_3^{(g)}=\partial_U =\Omega e_3$, $e_A^{(g)}=\partial_{\Theta^A}$, $e_4^{(g)}=\Omega^{-1}e_4$, and denote the corresponding connection components by $\psi^{(g)}$. For a function $W$ on $\SS$, consider the sphere $S_{W,V}=\{(U,V,\Theta):U=W(\Theta)\}$. The outer null expansion of $S_{W,V}$ is
\begin{equation}\label{Equation_of_First_Variation}
        \begin{aligned}
            \ell^+(W)=& \dv_{S_{W,V}}\left(e_4^{(g)}+2\snab W{\color{red}+}\left|\snab W\right|^2e_3^{(g)}\right)\\
            =&\tr\chi^{(g)}+2\slap W+4\zeta^{(g)}_A\snab^A W-\left|\snab W\right|^2\left(\tr\chib^{(g)}+4\omegab^{(g)}\right).
        \end{aligned}
    \end{equation}
\begin{remark}
    With respect to the null frame $e_\mu^{(g)}$, we have
    \begin{equation*}
        \chib^{(g)}_{AB}=\Omega\chib_{AB},\ \chi^{(g)}_{AB}=\Omega^{-1}\chi_{AB},\ \omegab^{(g)}=2\Omega\omegab,\ \zeta^{(g)}_A=\zeta_A+\nabla_A\log\Omega.
    \end{equation*}
\end{remark}
\begin{remark}
    \begin{equation*}
        \begin{aligned}
            \slap W=&\sqrt{g}^{-1}\partial_A \left(\sqrt{g}g^{AB}(W(\Theta),\Theta)\partial_B W\right)
            =\Delta_g W|_{U=W(\Theta)}-2\chibh\left(\snab W,\snab W\right).
        \end{aligned}
    \end{equation*}
\end{remark}
We denote $\gamma_{AB}(U,\Theta)=(-U)^{-2}g_{AB}(U,\Theta)$ and{,} for a function $f=f(\Theta)${,} we write $$\Delta_{\gamma}f(U,\Theta)=\partial_A \left(\gamma^{AB}(U,\Theta)\partial_B f(\Theta)\right).$$
Let $\gamma_0$ be the standard sphere metric. The estimate $\lnm \wt{g}\rnm_{H^3(g)}=o(1)$ then implies that
\begin{equation*}
    \Delta_\gamma f-\Delta_{\gamma_0}f= o(1)\left|\partial f\right|+o(1)\left|\partial^2 f\right|.
\end{equation*}
Set $W=-Ve^{-\Phi(\Theta)}$. Then $\sqrt{g}(W,\Theta)=(-W)^2(1+o(1))$, and hence
\begin{equation*}
    \begin{aligned}
        \frac{1}{2}(-W)\ell^+= & \frac{1}{2}(-W)\Omega^{-1}\tr\chi -2\langle\zeta^{(g)},\snab \Phi\rangle_\gamma +\Delta_\gamma \Phi -\left|\snab \Phi\right|^2_{\gamma}-\left|\snab \Phi\right|^2_{\gamma }\left(-1+\k+o(1)\right)\\
        \leq & \frac{1}{2}(-W)\Omega^{-1}\tr\chi+o(1) +\Delta_{\gamma}\Phi-\frac{\k}{2}\left|\snab \Phi\right|_\gamma^2.
    \end{aligned}
\end{equation*}
First, we have $\frac{1}{2}(-U)\Omega^{-1}\tr\chi+o(1)\leq C_0$ because $\Omega^{-1}\tr\chi$ is decreasing in the direction of $e_4$.
We can choose $\Theta_1$ and $U_1,V_1=\epsilon_1\left(-U_1\right)^{1+3\k/2N}$ such that, whenever $|U-U_1|\leq (1-\frac{1}{C_0})\left|U_1\right|,\ V=V_1,\ (U_1,V_1,\Theta)\in B_{S_{U_1,V_1}}(\Theta_0,(-U_1)^{1+\delta})$, we have $(U,V,\Theta)\in B_{S_{U,V}}(\Theta_0,\frac{1}{C(C_0)}(-U)^{1+\delta})$ and $\frac{\epsilon_1}{C(C_0)}\leq\frac{V}{(-U)^{1+3\k/2N}}\leq C(C_0)\epsilon_1$. Consequently,
\begin{equation*}
    \frac{1}{2}(-U)\Omega^{-1}\tr\chi(U,V,\Theta)+o(1)<-\frac{1}{(-U)^{\delta_1}}.
\end{equation*}

To construct the trapped surface, we first recall a useful lemma.
\begin{lemma}
    For any $\epsilon_2>0$, there exists a function $h$ such that
    \begin{equation*}
        \Delta_{\gamma(h)}h-\frac{\k}{2}\left|\snab h\right|^2_{\gamma(h)}<-1/2,\ {\rm in}\ \SS\setminus B(\Theta_0,\epsilon_2),\ h=0\ {\rm on}\ \partial B(\Theta_0,\epsilon_2),
    \end{equation*}
    and
    \begin{equation*}
        \left|\nabla h\right|\leq C_0 \epsilon_2^{-1},\ \left|\nabla^2 h\right|\leq C_0 \epsilon_2^{-2}.
    \end{equation*}
\end{lemma}

\begin{proof}
    We consider the coordinates $(\theta,\phi)=(\sin(\theta)\cos(\phi),\sin(\theta)\sin(\phi),\cos(\theta))$ on the sphere and assume $\theta(\Theta_0)=0$.\\
    We consider the function
    \[h_0(\theta)=2\log\sin\left(\frac{\theta}{2}\right)-2\log\sin\left(\frac{\epsilon_2}{2}\right).\]
    We then have 
    \begin{equation*}
        \partial_\theta h_0=\frac{\cos(\theta/2)}{\sin(\theta/2)},
    \end{equation*}
    \begin{equation*}
        \nabla_{\gamma_0}^2h_0(\partial_\theta,\partial_\theta)=\partial^2_\theta h_0=-\frac{1}{2}-\frac{\cos^2(\theta/2)}{2\sin^2(\theta/2)},\ \nabla_{\gamma_0}^2h_0(\partial_\theta,\partial_\phi)=0,
    \end{equation*}
    \begin{equation*}
        \nabla_{\gamma_0}^2h_0(\partial_\phi,d\phi)=-g^{\phi\phi}\Gamma_{\phi\phi}^\theta\partial_\theta h_0=\sin^{-1}(\theta)\cos(\theta)\frac{\cos(\theta/2)}{\sin(\theta/2)}=\frac{\cos^2(\theta/2)}{2\sin^2(\theta/2)}-\frac{1}{2},\ \Delta_{\gamma_0}h_0=-1.
    \end{equation*}
    Therefore, we have 
    \begin{equation*}
        \begin{aligned}
            \Delta_{\gamma(h_0)}h_0-\frac{\k}{2}\left|\nabla h_0\right|^2_{\gamma(h)}\leq & \Delta_{\gamma_0}h_0 +o(1)\left(\left|\nabla h_0\right|+\left|\nabla^2 h_0\right|\right)-\frac{\k}{4}\left|\nabla h_0\right|^2_{\gamma_0}\\
            \leq & -1 +o(1)\left(1+\frac{\cos^2(\theta/2)}{\sin^2(\theta/2)}\right)-\frac{\k}{4}\frac{\cos^2(\theta/2)}{\sin^2(\theta/2)}< -\frac{1}{2}.
        \end{aligned}
    \end{equation*}
    
\end{proof}
Take $\epsilon_2=(-U_1)^{\delta}/10$. We obtain the corresponding $h_0$ on $\SS\setminus B_{\SS}(\Theta_0,(-U_1)^{\delta}/10)$ and set $h_0=0$ in $B_{\SS}(\Theta_0,(-U_1)^{\delta}/10)$. We then smooth the function to obtain $h_1$ such that $h_1=h_0$ on $\SS\setminus B_{\SS}(\Theta_0,(-U_1)^{\delta}/5)$, $h_1=0$ on $ B_{\SS}(\Theta_0,(-U_1)^{\delta}/100)$, and
\[\left|\nabla h_1\right|\leq C_0 (-U_1)^{-\delta},\ \left|\nabla^2 h_1\right|\leq C_0 (-U_1)^{-2\delta}.\]
Let $\Phi=2C_0 h_1+\log\left(\frac{V_1}{-U_1}\right)$, so that $W(\Theta_0)=U_1$. For $\Theta\in B_{\SS}(\Theta_0,(-U_1)^{\delta}/5)$, we have $1\geq e^{-2C_0h_1}\geq \left(\frac{\sin(\epsilon_2/2)}{\sin((-U_1)^\delta/10)}\right)^{4C_0}\geq \frac{1}{C_0}$ and thus 
\begin{equation*}
    \left|W(\Theta)-U_1\right|=(-U_1)\left|e^{-h_1(\Theta)}-1\right|\leq \left(1-\frac{1}{C_0}\right)(-U_1).
\end{equation*}
Therefore, for $\Theta\in B_{\SS}(\Theta_0,(-U_1)^{\delta}/5)$, we have
\begin{equation*}
    \frac{1}{2}(-W)\Omega^{-1}\tr\chi(W,V_1,\Theta)+o(1)<-\frac{1}{(-W)^{\delta_1}}.
\end{equation*}
When $\Theta\in \SS\setminus B_{\SS}(\Theta_0,(-U_1)^{\delta}/5)$, we have 
\begin{equation*}
    \frac{1}{2}(-W)\ell^+\leq C_0+2C_0\Delta_\gamma h_1-\frac{\k}{2}\cdot 4C_0^2\left|\nabla h\right|_{\gamma}^2\leq 0,
\end{equation*}
and in the remaining region, we have
\begin{equation*}
    \begin{aligned}
        \frac{1}{2}(-W)\ell^+\leq & \frac{1}{2}(-W)\Omega^{-1}\tr\chi+o(1) +\Delta_{\gamma}\Phi-\frac{\k}{2}\left|\snab \Phi\right|_\gamma^2\\
        \leq & -\frac{1}{(-U_1)^{\delta_1}}+\frac{C}{(-U_1)^{2\delta}}.
    \end{aligned}
\end{equation*}
Choosing $\delta_1>2\delta$ and $(-U_1)$ sufficiently small, we conclude that $\ell^+<0$ everywhere on $S_{W,V_1}$, establishing the existence of a trapped surface.

\begin{corollary}
    In our setting, along each incoming null cone $\Hb_v$ that contains a trapped surface, there exists a marginally outer trapped surface (MOTS).
\end{corollary}
\begin{proof}
    Based on our derived estimates, we first verify a matter-focusing condition. Recall the definition of the Weyl curvature component
    $\alphab=W_{3A3B}=R_{3A3B}-\frac{1}{2}\Rc_{33}g_{AB},$
    with initial value
    $$\Omega^2\alphab(u,0)=(-\Omega\nabla_3(\Omega\chibh)-\frac{1}{2}\Omega\tr\chib\Omega\chib-4\Omega\omegab\Omega\chibh)(u,0)\sim \epsilon/(-u)^2.$$
    Recall the equations
    \[
        \Omega\nabla_4(\Omega^2\alphab)=-\frac{1}{2}\Omega\tr\chi\Omega^2\alphab -\Omega^3\nabla\hat\otimes\betab-3\Omega^3(\chibh\rho-{}^*\chibh\sigma)+\Omega^3(\zeta-4\etab)\hat\otimes\betab,
    \]
    \[\rho
=-K+\frac12R+\Rc_{34}
+\frac12\chibc\cdot\chic
-\frac12\tr\chi\,\tr\chib,\ \sigma=\cl\zeta-\frac12\chibh\wedge\chih,\]
\[\betab_A
=\dv\chibc_A-(\zeta\cdot\chibc)_A-\frac12\Rc_{3A},\]
and we find 
\begin{equation*}
    \begin{aligned}
        &\sum_{i\leq 3}\lnm \nabla^i\wt{\Omega^2\alphab}\rnm^2_{\LS(u,v;N)}\\
        &\quad\lesssim \frac{v}{(-u)^{1+C\delta}} \sum_{i\leq 4,j\leq 3}\lnm \nabla^i\left(\wt{\Omega^2\eta},\wt{\Omega^2\etab},\wt{\Omega\chi},\wt{\Omega\chib},\wt{\dv(\Omega\chibc)}\right),\nabla^j\wt{\Omega^2 K}\rnm^2_{\LH(u,v;N)}\ll 1.
    \end{aligned}
\end{equation*}
Thus{,} for $v/(-u)^{1+C\delta}\leq \epsilon$,
    {we have} 
    \begin{equation}\label{eq:matter-focusing}
    \left|\Omega D_3\phi\right|^2\sim (-u)^{-2}\gg \epsilon(-u)^{-2}\gtrsim \left|\Omega\chib\right|\cdot\left|\Omega\chibh\right|+\left|\Omega^2\alphab\right|.
      \end{equation}
     {\color{black}
    This is the desired matter-focusing condition that to be used in below parabolic flow approach.
      
   \vspace{1mm}
    We then run the flow directly on the portion of $\Hb_v$ bounded by the trapped graph constructed above and a round outer un-trapped section lying to its timelike past, taking the latter as the initial surface. In the graph coordinate $\omega$ adapted to the past-directed generator, we then consider the below scalar quasilinear parabolic equation
    \begin{equation}
        \partial_t\omega=-\frac{1}{2}\ell^+(\omega).
    \end{equation}
    Here $S_{\omega(t)}\subset\Hb_v$ is the graph determined by $\omega(t,\cdot)$, and $\ell^+(\omega)$ denotes its outgoing null expansion. With the trapped surface constructed above as an inner barrier and the round outer untrapped section as the initial surface, the maximum principle confines the evolving graphs to a fixed compact portion of $\Hb_v$ and yields the required $C^0$ bound. The essential remaining step is to control $|\nabla\omega|^2$. Set $q=\frac12|\nabla\omega|^2$. On the set where $q>0$, one then applies the maximum principle to the auxiliary quantity $\log q+f(\omega)$, with $f$ chosen appropriately. The matter-focusing condition \eqref{eq:matter-focusing} provides precisely the favorable sign in its evolution inequality, yielding a uniform bound for $q$, and hence for $|\nabla\omega|$. Together, the $C^0$ and gradient bounds make the flow uniformly parabolic, so standard parabolic regularity gives uniform higher-order estimates and continuation for all time. Moreover, since the initial section is strictly outer untrapped, the maximum principle applied to the evolution equation for $\ell^+$ preserves $\ell^+(\omega(t))>0$. The flow equation then gives $\partial_t\omega<0$, so $\omega(t,\theta)$ is pointwise decreasing in $t$. The uniform estimates imply smooth convergence to a limiting section, and integration of the flow equation shows that this limit satisfies $\ell^+=0$. Since $\tr\chib<0$, the limiting section is a MOTS.
    }
    For more details of the parabolic arguments, see \cite[Theorem 1.2]{RoeschScheuer2021}.
\end{proof}

\subsection{More remarks on MOTS}
In this section, we consider a more general geometric setting for the study of marginally outer trapped surfaces (MOTS). Let $\Hb$ be an incoming null cone embedded in a Lorentzian manifold and let $\Lb$ denote the tangent null vector field satisfying $D_{\Lb}\Lb=0$.
For any spacelike sphere $S_0$ in $\Hb$, we say that $(s,\theta)$ is the \textit{geodesic foliation} based on $S_0$ if $\theta$ denotes coordinates on $S_0$ and each point on $\Hb$ is represented as $(s,\theta)=\exp_{\theta}(s\cdot \Lb|_{\theta})$. For any positive function $f(s,\theta)$, define $u(s,\theta)$ by 
$$u(s,\theta) =\int_0^s\frac{1}{f(\tau,\theta)}d\tau-1.$$
After the coordinate change $(s,\wt{\theta})\mapsto (u,\theta)=(u(s,\wt\theta),\wt{\theta})$, we have $\partial_u =f\Lb$. We then consider the null frame $e_3=\partial_u,e_A=\partial_{\theta^A}$ on $\Hb$, with $e_4$ defined as the null companion of $e_3$ on the level sets of $u$. For a function $w=w(\theta)$, write $S_w=\{(u,\theta)\in\Hb:u=w(\theta),\theta\in\SS\}$. The null frame on $S_w$ is $e_3^\prime=e_3$, $e_A^\prime=w_*(e_A)=e_A+\snab_A w e_3$, and $e_4^\prime=e_4+2\snab w-\left|\snab w\right|^2 e_3$. We note that $\snab$ denotes the connection on $S_w$.
Now we specialize to $f(s,\theta)=\sqrt{g(s,\theta)}^a\sqrt{g(0,\theta)^{-1}}^{a}$, which yields
\begin{equation*}
    D_{e_3}e_3=f\partial_s f \Lb= \partial_u(\log f) e_3=a\tr\chib e_3.
\end{equation*}
We recall the first variation formula \ref{Equation_of_First_Variation} for the area functional in the direction of $e_4^\prime$:
\begin{equation*}
    \begin{aligned}
        \ell^+(w)=\tr\chi+4\snab^Aw\zeta_A+2\slap w+(2a-1)\left|\snab w\right|^2\tr\chib,
    \end{aligned}
\end{equation*}
where all connection components are defined with respect to the null frame $e_\mu$.
We proceed to derive the second variation formula for the null expansion in the direction of $e_3$.
\begin{proposition}[Second {variation} formula]
    For any $w=w(t,\theta)$ with $w_t=\partial_t w$, we have 
\begin{equation}\label{Second_Variation}
    \begin{aligned}
        \frac{d}{dt}\left(\ell^{+}(w)\right)=&-(1+a)w_t\tr\chib\ell^++ 2\slap w_t- w_t\left(2K^\prime-\strc\Rc\right) \\
        &\qquad+2w_t\sdiv (\zeta^\prime)+4\snab w_t\cdot \zeta^\prime+2w_t\left|\zeta^\prime\right|^2,
    \end{aligned}
\end{equation}
where $\psi^\prime$ denotes the connection components with respect to the frame $e_\mu^\prime$, $K^\prime$ is the Gauss curvature of $S_w$, and $\strc,\sdiv,\slap$ are the operators associated with $g|_{S_w}$.
\end{proposition}
\begin{proof}
    Differentiating with respect to $t$ and using the definitions of the null frame, we obtain
\begin{equation}\label{1_1_1_1}
    \begin{aligned}
        &\frac{d}{dt}\ell^+(w)
        =-2w_t\chib^{AB}\chi^\prime_{AB}+g^{AB}\langle D_{e^\prime_A}e^\prime_4,D_{\partial_t} e_B^\prime \rangle
        +g^{AB}\left(\langle D_{e^\prime_A}D_{\partial_t}e^\prime_4, e_B^\prime \rangle- R(\partial_t,e_A^\prime,e_4^\prime,e_B^\prime)\right).
    \end{aligned}
\end{equation}
To express $D_{\partial_t}e_4^\prime$ and $D_{\partial_t} e_B^\prime$, we compute their inner products with the frame vectors. {We first note that}
\begin{equation*}
    \langle D_{\partial_t}e_4^\prime,e_3 \rangle=-\langle e_4^\prime,D_{w_t e_3}e_3 \rangle=2a\tr\chib w_t ,\ \langle D_{\partial_t}e_4^\prime,e_4^\prime \rangle=0,
\end{equation*}
\begin{equation*}
    \begin{aligned}
        \langle D_{\partial_t} e_4^\prime,e_A^\prime \rangle =&-\langle  e_4^\prime,D_{e_A^\prime}\partial_t \rangle=-\langle  e_4^\prime,e_3 \rangle \snab_{A} w_t-w_t\langle e_4^\prime, D_{e_A^\prime}e_3 \rangle
        =2\snab_A w_t+2w_t\zeta^\prime_A.
    \end{aligned}
\end{equation*}
For $D_{\partial_t}e_B^\prime$, we have 
\begin{equation*}
    \langle D_{\partial_t}e_B^\prime,e_3\rangle=-\langle e_B^\prime,D_{\partial_t}e_3\rangle =0,\ \langle D_{\partial_t}e_B^\prime,e_4^\prime\rangle=-\langle e_B^\prime,D_{\partial_t}e_4^\prime\rangle=-2\snab_B w_t-2w_t\zeta^\prime_B.
\end{equation*}
\begin{equation*}
    \langle D_{\partial_t}e_B^\prime,e_A^\prime\rangle=w_t \chib_{AB}.
\end{equation*}
Combining these, we arrive at
\begin{equation*}
    D_{\partial_t} e_4^\prime=2\left(\snab w_t+w_t\zeta^\prime\right)^Ae_A^\prime -a\tr\chib w_t e_4^\prime,\ D_{\partial_t}e_B^\prime=w_t\chib_{B}^Ce^\prime_C+\left(\snab_B w_t+w_t\zeta^\prime_B\right)e_3 .
\end{equation*}
        Substituting {these expressions} into the variation formula, we find
\begin{equation*}
    \begin{aligned}
        &g^{AB}\langle D_{e^\prime_A}e^\prime_4,D_{\partial_t} e_B^\prime \rangle
        =w_t\chib^{AB}\chi^\prime_{AB}+2\left(\snab_A w_t+w_t\zeta^\prime_A\right){\zeta^\prime}^A,\\
        &g^{AB}\left(\langle D_{e^\prime_A}D_{\partial_t}e^\prime_4, e_B^\prime \rangle- R(\partial_t,e_A^\prime,e_4^\prime,e_B^\prime)\right)\\ &\qquad\qquad\qquad
        = 2\slap w_t +2\sdiv (w_t\zeta^\prime)-aw_t\tr\chib \ell^+-w_t g^{AB}R(e_3,e_A^\prime,e_4^\prime,e_B^\prime).
    \end{aligned}
\end{equation*}
    Inserting these expressions into \ref{1_1_1_1}, we conclude that
\begin{equation*}
    \begin{aligned}
        &\frac{d}{dt}\ell^{+}(w)=-w_t\left(\chib^{AB}\chih^\prime_{AB}+g^{AB}R(e_3,e_A^\prime,e_4^\prime,e_B^\prime)\right)+2\slap w_t +2w_t\sdiv (\zeta^\prime)+4\nabla w_t\cdot \zeta^\prime \\ &\qquad\qquad\qquad +2w_t\left|\zeta^\prime\right|^2-aw_t\tr\chib \ell^+.
    \end{aligned}
\end{equation*}
Recall that for any spacelike sphere $S$ with null companions $e_3,e_4$, the Gauss--Codazzi equation implies that
\begin{equation*}
K=\frac{1}{2} g^{A B} R_{3 A 4 B}+\frac{1}{2} R+\frac{1}{2} \operatorname{Ric}_{34}+\frac{1}{2} \hat{\chi} \cdot \hat{\chi}-\frac{1}{4} \operatorname{tr} \chi \operatorname{tr} \underline{\chi},
\end{equation*}
which is equivalent to $g^{AB}R_{3A4B}+\chi\cdot\chib=2K+\tr\chi\tr\chib-g^{AB}\Rc_{AB}.$
Substituting this identity into the preceding expression completes the proof of the second variation formula.

\end{proof}
Several important consequences follow directly from the second variation formula.

\begin{proposition}[Uniqueness of MOTS]\label{Uniqueness_of_MOTS}
    Suppose that two surfaces $S_{w_0}$ and $S_{w_1}$ satisfy $\ell^+(w_0)=\ell^+(w_1)=0$. Let $w(t)=w_0+t(w_1-w_0)$ and assume that $2K_{w(t)}+\frac{1}{2}\tr\chib\ell^+(w(t))-\tr_{w(t)}\Rc-2\left|\zeta_{w(t)}\right|^2> 0$. Then $w_0=w_1$.
\end{proposition}
\begin{proof}
    Let $a=\frac{1}{2}$ and $h=w_1-w_0$. By assumption, we have 
    \begin{equation*}
        \begin{aligned}
            0=&\int_{S_{w_1}}h\ell^+(w_1)-\int_{S_{w_0}}h\ell^+(w_0)\\
            =&\int_0^1\int_{S_{w(t)}}\tr\chib h^2\ell^+(w(t))\\ &\qquad+h\left(-\frac{3}{2}h\tr\chib\ell^+(w(t))+2\slap h-h(2K^\prime-\strc\Rc)+2h\sdiv\zeta^\prime+4\snab h\cdot\zeta^\prime+2h\left|\zeta^\prime\right|^2\right)\\
            =&\int_0^1\int_{S_{w(t)}}-2\left|\snab h\right|^2-\left(2K^\prime+\frac{1}{2}\tr\chib\ell^+(w(t))-\strc\Rc-2\left|\zeta^\prime\right|^2\right)h^2.
        \end{aligned}
    \end{equation*}
    Since the integrand is non-positive, we conclude that $h\equiv 0$.
\end{proof}
\begin{remark}
    We can compute $2K^\prime+\frac{1}{2}\tr\chib\ell^+(w)-\strc\Rc-2\left|\zeta^\prime\right|^2$ more explicitly for a sphere $S_w$. By the definition of $\zeta^\prime$, we have
    \begin{equation*}
        \begin{aligned}
            \zeta^\prime_A=&-\frac{1}{2}\langle D_{e_A^\prime} e_3,e_4^\prime\rangle=\zeta_A-\frac{1}{2}\langle D_{e_A} e_3,2\nabla^B w e_B+\left|\snab w\right|^2e_3\rangle-\frac{1}{2}\snab_A w\langle D_{e_3} e_3,e_4\rangle\\
            =&\zeta_A -\chib_A^B\snab_B w.
        \end{aligned}
    \end{equation*}
    At a fixed point on $\SS^2$, we consider normal coordinates $\partial_1,\partial_2$ with respect to the metric $g_{AB}(w(\theta),\theta)$. Then for any tensor $\psi$, we have $\partial_A\psi=\snab_A\psi$ since ${\Gamma^\prime}^A_{BC}= 0$.
    Let $\Gamma^\prime,\Gamma$ be the Christoffel symbols of $g(w(\theta),\theta)$ and $g(\theta)$, respectively. Then
    \begin{equation*}
        \wt{\Gamma}^C_{AB}={\Gamma^\prime}^C_{AB}-\Gamma^C_{AB}=\snab_A \chib_{BC}+\snab_B w\chib_{AC}-\snab_Cw\chib_{AB}.
    \end{equation*}
    By Theorema Egregium, the {difference between the Gauss curvatures} satisfies
    \begin{equation*}
        \begin{aligned}
            &2K^\prime-2K\\
            =&R_{BABA}(g\circ w)-R_{BABA}(g)\\
            =&-\partial_A\wt{\Gamma}^B_{BA}+\partial_B\wt{\Gamma}^B_{AA}-\wt{\Gamma}^B_{AE}\wt{\Gamma}^E_{AB}+\wt{\Gamma}^B_{BE}\wt{\Gamma}^E_{AA}\\
            =&-\snab_A\left(\snab_Aw\tr\chib\right)+\snab_B\left(2\snab_Aw\chib_{AB}+\snab_B w\tr\chib\right)-(\snab_Aw\chib_{BE})^2\\
            &+\left(\snab_Ew\chib_{AB}+\snab_B w\chib_{AE}\right)^2+\snab_Ew\tr\chib\left(2\snab_A w\chib_{AE}-\snab_E w\tr\chib\right)\\
            =&-\sdiv (\tr\chib \snab w)+2\sdiv(\snab w\cdot\chibh)+\left|\snab w\right|^2\left(\frac{1}{2}(\tr\chib)^2+\left|\chibh\right|^2\right)-2\chib_{AB}\chib^A_{E}\snab^B w\snab^E w\\
            &+2\tr\chib\chibh(\snab w,\snab w)\\
            =&-\sdiv (\tr\chib \snab w)+2\sdiv(\snab w\cdot\chibh)+\left|\snab w\right|^2\left|\chibh\right|^2-2\left|\chibh\cdot\snab w\right|^2.
        \end{aligned}
    \end{equation*}
    Thus,
    \begin{equation*}
        \begin{aligned}
            &2K^\prime+\frac{1}{2}\tr\chib\ell^+(w)-\strc\Rc-2\left|\zeta^\prime\right|^2\\
            =&2K-\sdiv (\tr\chib \snab w)+2\sdiv(\snab w\cdot\chibh)+\left|\snab w\right|^2\left|\chibh\right|^2-2\left|\chibh\cdot\snab w\right|^2+\frac{1}{2}\tr\chib\ell^+(w)\\
            &-\tr\Rc-2\snab^A w\Rc_{3A}-\Rc_{33}\left|\snab w\right|^2-2\left|\zeta^2\right|-4\chib_{AB}\zeta^A\snab^B w-2\left|\chib\cdot\snab w\right|^2.
        \end{aligned}
    \end{equation*}
    Using $\snab\tr\chib=\nabla\tr\chib+\snab w\left((a-\frac{1}{2})(\tr\chib)^2-\left|\chibh\right|^2-\Rc_{33}\right)$, we have 
    \begin{equation*}
        \begin{aligned}
            &2K^\prime+\frac{1}{2}\tr\chib\ell^+(w)-\strc\Rc-2\left|\zeta^\prime\right|^2\\
            =& \frac{1}{2}\tr\chib\tr\chi+2\tr\chib\zeta\cdot\snab w+\tr\chib\slap w\\
            &+2K+\frac{1}{2}\left|\snab w\right|^2(\tr\chib)^2-\frac{1}{2}\nabla\tr\chib\cdot\nabla w+\left|\chibh\right|^2\left|\snab w\right|^2+\Rc_{33}\left|\snab w\right|^2-\tr\chib\slap w\\
            &+2\sdiv(\snab w\cdot\chibh)+\left|\snab w\right|^2\left|\chibh\right|^2-2\left|\chibh\cdot\snab w\right|^2\\
            &-\tr\Rc-2\snab^A w\Rc_{3A}-\Rc_{33}\left|\snab w\right|^2-2\left|\zeta\right|^2-4\chib_{AB}\zeta^A\snab^B w-2\left|\chib\cdot\snab w\right|^2\\
            =&2K+\frac{1}{2}\tr\chib\tr\chi -\tr\Rc-2\left|\zeta\right|^2-\frac{1}{2}\nabla\tr\chib\cdot\nabla w+2\left|\chibh\right|^2\left|\snab w\right|^2\\
            &+2\sdiv\left(\chibh\cdot\snab w\right)-4\left|\chibh\cdot\snab w\right|^2-2\snab^Aw\Rc_{3A}-4\chibh_{AB}\zeta^A\snab^B w-2\tr\chib\chibh_{AB}\snab^A w\snab^B w.
        \end{aligned}
    \end{equation*}
\end{remark}
\begin{remark}\label{Remark_of_MOTS_uniqueness}
    For the MOTS and trapped surface considered in \cite{An2025}, the smallness of $\zeta,\chibh$, $\nabla\tr\chib$, $\nabla\phi$, $\zeta$, together with the uniform bound for $\left|\snab^2\log w\right|+\left|\snab\log w\right|$, leads to the estimate
    \begin{equation*}
        2K^\prime+\frac{1}{2}\tr\chib\ell^+(w)-\strc\Rc-2\left|\zeta^\prime\right|^2\geq 2K+\frac{1}{2}\tr\chi\tr\chib+o(1)R^{-2}(w),
    \end{equation*}
    where $R$ denotes a radius function with $K= \frac{1+o(1)}{R^2}$, $\Omega\tr\chib= -\frac{2+o(1)}{R}$, and $\Omega^{-1}\tr\chi\leq \frac{2}{(1+\k)R}$. Therefore, the MOTS is unique.
\end{remark}

\begin{proposition}[Null Comparison Principle]\label{Null_Comparison_Principle}
    Suppose that there is a MOTS $S_{w_0}$ and a surface $S_{w_1}$ with $\ell^+(w_1)\leq 0$. Let $w(t)=w_0+t(w_1-w_0)$. Suppose that $\int_0^1g^{AB}(w(t),\theta)dt$ is a pointwise positive definite matrix and that $(1+a)\tr\chib\ell^+(w(t))+2K^\prime_{w(t)}-\strc \Rc-2\sdiv\zeta^\prime_{w(t)}-2\left|\zeta^\prime_{w(t)}\right|^2>0$ for some $a\in\mathbb{R}$ and all $t\in(0,1)$. Then $w_1\geq w_0$.
\end{proposition}
\begin{proof}
    Write $h=w_1-w_0$. We have 
    \begin{equation*}
        \begin{aligned}
            0\geq & \ell^+(w_1)-\ell^+(w_0)\\
            =&2\left(\int_0^1 g^{AB}(w(t),\theta)\right)\partial^2_{AB}h+\left(\int_0^1 2\sqrt{g}^{-1}\partial_A(\sqrt{g}g^{AB})+4{\zeta^\prime}^B\right)\partial_B h\\
            &-2\int_0^1\left(\frac{1+a}{2}\tr\chib\ell^+(w(t))+K^\prime_{w(t)}-\frac{1}{2}\strc \Rc-\sdiv\zeta^\prime_{w(t)}-\left|\zeta^\prime_{w(t)}\right|^2\right)h.
        \end{aligned}
    \end{equation*}
    Suppose for contradiction that $h$ attains a negative minimum. At such a point, the right-hand side is bounded below by
    \begin{equation*}
        -2\int_0^1\left(\frac{1+a}{2}\tr\chib\ell^+(w(t))+K^\prime_{w(t)}-\frac{1}{2}\strc \Rc-\sdiv\zeta^\prime_{w(t)}-\left|\zeta^\prime_{w(t)}\right|^2\right)h,
    \end{equation*}
    which is strictly positive and leads to a contradiction.
\end{proof}

\begin{remark}\label{Remark_An2025_Null_Comparison}
{\color{black}
In the setting of \cite{An2025}, the same asymptotic estimates as in Remark~\ref{Remark_of_MOTS_uniqueness} can be applied to verify the positivity hypothesis in Proposition~\ref{Null_Comparison_Principle} along the interpolation $w(t)$. The null comparison principle therefore yields the desired comparison between the corresponding graph spheres. In particular, this provides an alternative proof of the uniqueness of the MOTS and the achronality of the apparent horizon in that setting.
}
\end{remark}

\begin{remark}
    If we consider the null foliation $x(t,\theta)=(w(t,\theta),\theta)$ and denote the corresponding quantities by $\psi^\prime$, then the strict positivity condition in Proposition~\ref{Null_Comparison_Principle} is equivalent to $\sqrt{g}^{-a-1}\nabla_{e_3^\prime}\left(\sqrt{g}^{a+1}\tr\chi^\prime\right)< 0,$ where $\sqrt{g}$ satisfies $dV_{g(t,\theta)}=\sqrt{g}(t,\theta)dV_{g(0,\theta)}$.
\end{remark}



 
\bibliographystyle{plain}

\end{document}